\documentclass[prx,aps,twocolumn,groupedaddress,nofootinbib,superscriptaddress,preprintnumbers]{revtex4-2}

\usepackage{graphicx}
\usepackage[nointegrals]{wasysym} %
\usepackage[export]{adjustbox}
\usepackage{amsmath,amsfonts,amssymb,latexsym}
\usepackage{hhline}
\usepackage{bm}
\usepackage{verbatim}
\usepackage{enumitem}
\usepackage{mathrsfs}
\usepackage{slashed}
\usepackage{empheq}

\newcommand{\Tr}{\text{Tr}}

\newcommand{\p}{\partial}

\newcommand{\lan}{\langle}
\newcommand{\ran}{\rangle}

\newcommand{\unit}{\mathbf{1}}

\newcommand{\da}{{\dagger}}

\newcommand{\ob}[1]{\mkern 1.5mu\overline{\mkern-1.5mu#1\mkern-1.5mu}\mkern 1.5mu}

\newcommand{\ra}{\rightarrow}

\newcommand{\wt}{\widetilde}

\renewcommand{\(}{\left(}
\renewcommand{\)}{\right)}
\renewcommand{\[}{\left[}
\renewcommand{\]}{\right]}
\newcommand{\mt}{\mapsto}

\newcommand{\tp}{\otimes}

\newcommand{\twp}{{2\pi}}

\newcommand\bpm            {\begin{pmatrix}}
	\newcommand\epm           {\end{pmatrix}}

\newcommand{\ms}{\medskip}

\def\app#1#2{%
	\mathrel{%
		\setbox0=\hbox{$#1\sim$}%
		\setbox2=\hbox{%
			\rlap{\hbox{$#1\propto$}}%
			\lower1.1\ht0\box0%
		}%
		\raise0.25\ht2\box2%
	}%
}

\newcommand{\spec}{\text{Spec}}

\newcommand{\tw}{\textwidth}

\renewcommand{\cp}{\Phi}
\newcommand{\ct}{\Theta}

\newcommand{\inv}{^{-1}}

\newcommand{\ope}\odot

\usepackage{manfnt}

\newcommand{\bi}{\begin{itemize}}
	\newcommand{\ei}{\end{itemize}}

\usepackage{amsthm}
\newtheorem{theorem}{Theorem}
\newtheorem{corollary}{Corollary}
\newtheorem{proposition}{Proposition}

\newtheorem{lemma}{Lemma}

\theoremstyle{definition}

\newcommand\bpro		  {\begin{proposition}}
	\newcommand\epro 		  {\end{proposition}}
\newcommand\bproof			  {\begin{proof}}
	\newcommand\eproof 		  {\end{proof}}

\newcommand\ed            {\end{definition}}

\newcommand\be            {\begin{equation}}
\newcommand\ee            {\end{equation}}
\newcommand\ba            {\begin{aligned}}
\newcommand\ea            {\end{aligned}}
\newcommand\bea{\begin{equation}\begin{aligned}}
	\newcommand\eea{\end{aligned}\end{equation}}

\usepackage{hyperref} 
\hypersetup{final}
\hypersetup{colorlinks, citecolor=red, linkcolor=red, urlcolor=red} 

\newcommand{\sss}{\subsubsection}
\renewcommand{\ss}{\subsection}

\renewcommand{\a}{\alpha}
\renewcommand{\b}{\beta}
\renewcommand{\d}{\delta}
\newcommand{\De}{\Delta}
\newcommand{\g}{\gamma}
\newcommand{\G}{\Gamma}
\newcommand{\s}{\sigma}
\renewcommand{\S}{\Sigma} 

\newcommand{\ep}{\varepsilon} %
\renewcommand{\l}{\lambda}
\renewcommand{\L}{\Lambda}

\renewcommand{\O}{\Omega}

\renewcommand{\r}{\rho}

\newcommand{\z}{\zeta}

\newcommand{\bfs}{\mathbf{s}}

\newcommand{\EE}{\mathbb{E}}

\newcommand{\nn}{\mathbb{N}}

\newcommand{\rr}{\mathbb{R}}
\newcommand{\qq}{\qquad}

\newcommand{\zz}{\mathbb{Z}}

\newcommand{\mcc}{\mathcal{C}}
\newcommand{\mck}{\mathcal{K}}

\newcommand{\mco}{\mathcal{O}}
\newcommand{\mcu}{\mathcal{U}}

\newcommand{\mch}{\mathcal{H}}

\newcommand{\mcm}{\mathcal{M}}
\newcommand{\mcn}{\mathcal{N}}

\newcommand{\sfd}{\mathsf{d}}

\newcommand{\sfg}{\mathsf{g}}

\usepackage[mathscr]{eucal} %

\newcommand{\kb}[2]{|{#1}\rangle\langle{#2}|}
\renewcommand{\k}[1]{|#1\rangle}

\newcommand{\proj}[1]{|#1\rangle\langle#1|}

\usepackage{dcolumn}
\usepackage{pifont}
\usepackage{ulem}

\usepackage[caption=false]{subfig}

\usepackage{enumitem}

\usepackage{amsthm}
\usepackage{physics}

\usepackage{bm}

\renewcommand{\bot}{\bigotimes}

\DeclarePairedDelimiterX{\infdivx}[2]{(}{)}{%
	#1\;\delimsize\|\;#2%
}

\newtheorem*{theorem*}{Theorem}

\newtheorem{definition}{Definition}

\renewcommand\qq{\qquad}
\renewcommand\cp{\Phi}

\renewcommand{\bot}{\bigotimes}

\usepackage{accents}

\newcommand{\oEE}{\mathop{\mathbb{E}}}

\newcommand{\irr}{{\mathsf{irr}}}
\newcommand{\ltot}{{L}}
\newcommand{\hfroz}{{\mathcal{H}_{\rm froz}}}
\newcommand{\lc}{{L_{\rm cons}}}
\begin{document}

	\title{Exponentially slow thermalization and the robustness of Hilbert space fragmentation}
	
	\author{Yiqiu Han}
	\affiliation{Department of Physics, Boston College, Chestnut Hill, MA 02467, USA}
    \affiliation{Department of Physics and Center for Theory of Quantum Matter, University of Colorado, Boulder, CO 80309, USA}
	\author{Xiao Chen}
	\affiliation{Department of Physics, Boston College, Chestnut Hill, MA 02467, USA}
	
	\author{Ethan Lake}
	\email[Electronic address:$~~$]{elake@berkeley.edu}
	\affiliation{Department of Physics, University of California Berkeley, Berkeley, CA 94720, USA}
	\begin{abstract}
			The phenomenon of Hilbert space fragmentation, whereby dynamical constraints fragment Hilbert space into many disconnected sectors, provides a simple mechanism by which thermalization can be arrested. However, little is known about how thermalization occurs in situations where the constraints are not exact. To study this, we consider a situation in which a fragmented 1d chain with pair-flip constraints is coupled to an ergodicity-restoring thermal bath at its boundary. 
            We numerically observe an exponentially long thermalization time in Hamiltonian dynamics, manifested in both entanglement dynamics and the relaxation of local observables. To understand this, we study an analogous model of random unitary circuit dynamics, whose thermalization time we prove scales exponentially with system size. Slow thermalization in this model is shown to be a consequence of strong bottlenecks in configuration space, which restrict how the system can explore Hilbert space, and demonstrate a new way of producing anomalously slow thermalization dynamics.
	\end{abstract}
	
	\maketitle

	\paragraph*{Introduction: } A central theme in quantum dynamics is the understanding of mechanisms which impede or arrest thermalization \cite{deutsch1991quantum, srednicki1994chaos, rigol2008thermalization, nandkishore2015many-body, dalessio2016from}. Many such mechanisms, most prominently many body localization \cite{alet2018many-body, abanin2019colloquium}, rely crucially on some form of spatial disorder. %
	A recent body of work has demonstrated that even without recourse to strong disorder, the imposition of certain dynamical constraints can be rigorously shown to prevent thermalization \cite{sala2020ergodicity, khemani2020localization,motrunich,pancotti2020quantum,van2015dynamics,mukherjee2021minimal,yang2020hilbert,brighi2023hilbert,lan2018quantum,garrahan2018aspects, detomasi2019dynamics,moudgalya2022thermalization,will2023realization,kwan2023minimal,li2023hilbert}. In these systems, thermalization is evaded by virtue of the dynamics being non-ergodic, with the space of product states splitting into exponentially many dynamically-disconnected ``fragments'', in a phenomenon known as {\it Hilbert space fragmentation} (HSF). 

	Unfortunately, this form of ergodicity breaking relies on fine-tuning the dynamics to ensure that the dynamical constraints leading to HSF are exactly obeyed. 
	What happens when the constraints are weakly broken is a relatively unexplored question (see however \cite{li2023hilbert,frey2024probing}), despite the fact that many experimental systems are close to fine-tuned points where the constraints are exact  \cite{kohlert2023exploring,kinoshita2006quantum,scheie2021detection,malvania2021generalized,wei2022quantum,guardado2020subdiffusion,scherg2021observing}. A natural question thus remains: how does the structure of HSF imprint itself on thermalization dynamics once its associated constraints are broken? In particular, can there exist models of constrained dynamics which display anomalous thermalization even in the presence of constraint-breaking, ergodicity-restoring perturbations?

	In this work, we answer this question affirmatively {in a restricted setting}, by studying what happens when a 1d spin chain with ``pair-flip'' constraints \cite{caha2018pair} is connected at one end to a chain undergoing generic unconstrained dynamics, {with the remainder of the chain left unperturbed} (Fig.~\ref{fig:inspirational_schematic}~a). 
	The coupling to the unconstrained system can be viewed as a coupling to a thermal bath, and it renders the dynamics fully ergodic. {The thermalization time $t_{\rm th}$ is then the time needed for the effects of a thermal bath at the boundary to spread throughout the bulk. Letting $\lc$ be the size of the constrained region, for generic dynamics we may therefore expect $t_{\rm th}$ to scale as either $\lc$ or $\lc^2$, depending on whether the bulk dynamics possess a local conserved charge. 
	The main result of this paper is to rigorously show that the influence of the bath in fact spreads much more slowly than this, with $t_{\rm th}$ instead scaling {\it exponentially} in $\lc$.}
	This slowness is due to strong bottlenecks that the system encounters as it tries to explore Hilbert space, a phenomenon which arises from the type of constraints and the local nature of the coupling to the bath. We rigorously prove an exponentially large lower bound on $t_{\rm th}$ in the setting where the system undergoes a constrained form of random unitary (RU) dynamics, and provide numerical evidence that $t_{\rm th}$ for Hamiltonian dynamics is similarly long. 
	Remarkably, the long approach to equilibrium can be diagnosed simply by measuring expectation values of certain local operators, which take exponentially long to reach their steady-state values.
	
	\begin{figure}
		\includegraphics[width=.5\tw]{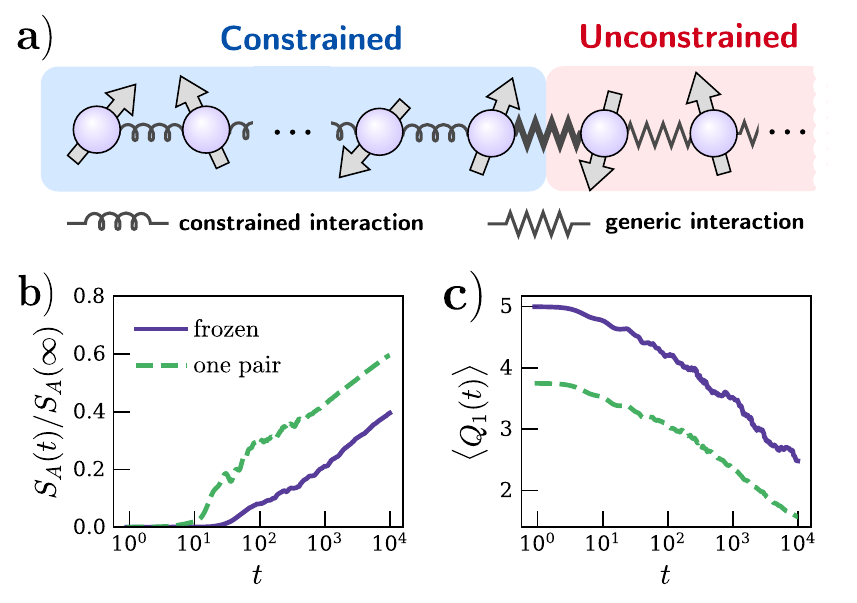}
		\caption{\label{fig:inspirational_schematic} {\bf a)} A schematic of the setup considered in this paper: a 1d chain is partitioned into a constrained region at sites $1<i\leq \lc$ within which the dynamics is fragmented, and an unconstrained region at sites $\lc<i\leq \ltot$ where the dynamics is generic. {\bf b)} Growth of the normalized entanglement entropy for a quench from under Hamiltonian dynamics, calculated for the region $A = [1,L/2]$ with $\lc=8, L=10, N=3$. The solid line shows a quench from the frozen state $\k{12}^{\tp \ltot/2}$; for the dashed line the first site is changed to $2$, creating a single flippable pair. In both cases, a slow logarithmic growth is observed. {\bf c)} Relaxation of the charge $Q_1$ computed in region $A$ for the same initial state, showing a similarly slow decay.}
	\end{figure} 
	
	\paragraph*{Slow thermalization in the pair-flip model:} 
	We begin by studying Hamiltonian dynamics. We consider a spin-$(N-1)/2$ model of the form $H=H_0+H_{\rm imp}$. Here the constrained Hamiltonian $H_0$ takes the form \cite{caha2018pair}
	\begin{equation}
		H_0=\sum_{i=1}^{\ltot-1}\sum_{a,b=1}^{N} g^{a,b}_{i} \kb{aa}{bb}_{i,i+1},
	\end{equation}
	with $g_i$ arbitrary $N\times N$ Hermitian matrices. The ``impurity'' Hamiltonian $H_{\rm imp}$ acts to break the constraints on sites $\lc<i\leq\ltot$, with $\lc$ the size of the constrained region. 
	$H_0$ only flips neighboring spins with identical values. As a result, it preserves the $U(1)$ charges
	\be \label{charges} Q_a \equiv \sum_{i} (-1)^i \proj{a}_i.\ee
	When $N>2$, which we will specify to unless explicitly noted otherwise, $H_0$ additionally possesses an exponentially large number of non-local conserved quantities~\cite{moudgalya2018exact}. Among these are the $N(N-1)^{\ltot-1}$ ``frozen'' states of the form $|a_1,\dots,a_\ltot\ran, \, a_i \neq a_{i+1}$, which are annihilated by $H_0$. We will denote the space spanned by these states as $\mch_{\rm froz}$.

	To numerically study the robustness of the constrained dynamics with respect to $H_{\rm imp}$, we investigate quantum quenches performed on initial states in $\hfroz$.
	For concreteness we will specify to the case where $g^{ab}_i = (-1)^i (1/N + \kappa\d_{a,b})$; here the $(-1)^i$ ensures that the states in $\hfroz$ lie approximately in the middle of $H_0$'s spectrum (since for this choice of $g^{ab}_i$, $H_0$ is an alternating-sign sum of frustration free projectors that annihilate $\mch_{\rm froz}$).
	The model obtained by setting $\kappa=0$ is a form of Temperley-Lieb model \cite{batchelor1991temperley,aufgebauer2010quantum}, which has $SU(N)$ symmetry, fragments in an entangled basis \cite{moudgalya2022hilbert}, and possesses a very large number of degenerate states. 
	In Sec.~\ref{App:TL}
	we prove that this model fails to thermalize even at infinite times when the constraint is broken only on a single site. In our numerics we will however fix $\kappa = 2/3$, breaking $SU(N)$ and yielding a more generic pair-flip model. 
	For $H_{\rm imp}$, we take for definiteness 
	\be \label{himp} H_{\rm imp}= \mcn\inv\sum_{i=\lc+1}^{\ltot}(e^{i\pi/4} X_i + e^{-i\pi/4} X_i^T), \ee
	where $X_i \equiv \sum_{a=1}^N |a\ran\lan [a+1]_N|_i$ with $[\cdot]_N$ denoting reduction modulo $N$, and with $\mcn$ ensuring that each term in $H_{\rm imp}$ has unit norm. This choice fully restores ergodicity, and can be numerically checked to render the spectrum of $H$ completely non-degenerate. 

	\begin{figure*}
		\centering 
		\includegraphics[width=\tw]{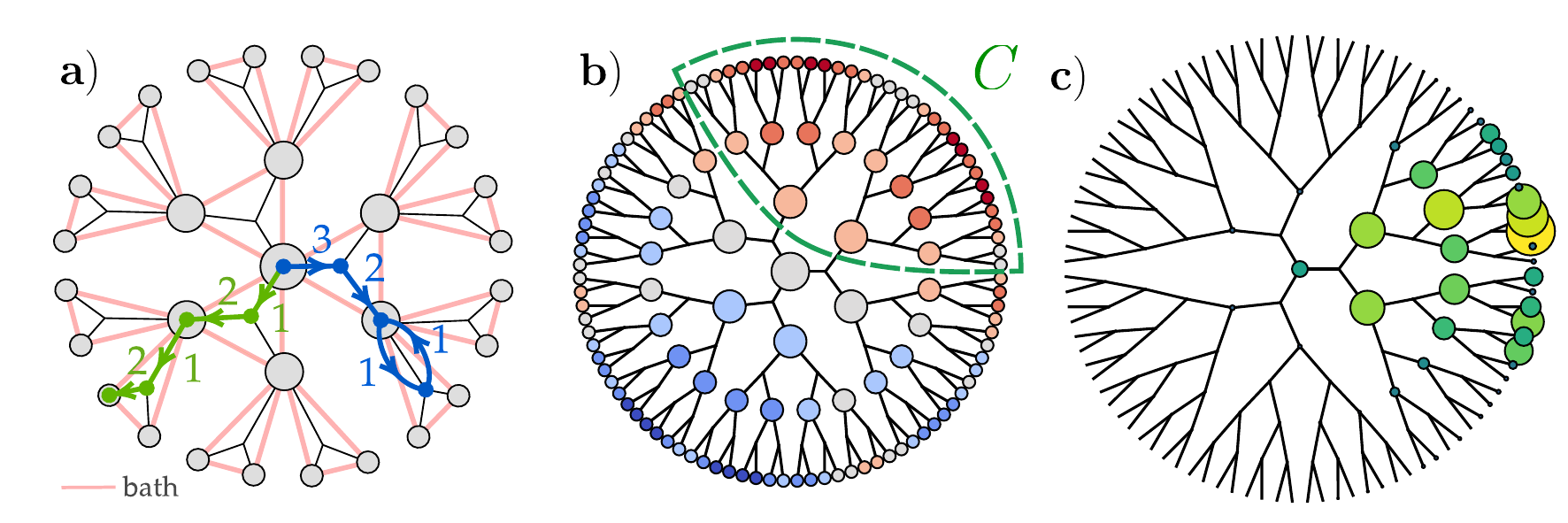} 
		\caption{\label{fig:trees} Krylov graph illustrations for $N=3$. {\bf a)} Each computational basis product state $ \k{a_1,\dots, a_L}$ 
			defines a length-$L$ walk on the degree 3 tree, with the walk backtracking when two identical labels are encountered in a row. The constrained dynamics preserves walk endpoints (grey circles), each of which defines a Krylov sector; the sectors are thus enumerated by the even (odd) sublattice of the depth-$L$ tree for even (odd) $L$. 
			Paths which reach the edge of the tree are frozen under the constrained dynamics, while those with backtracks belong to sectors of dimension $>1$. 
			Breaking the constraints at the edge induces transitions between next-nearest-neighbor tree vertices in the manner indicated by the pink lines. {\bf b)} Values of the charge $Q_1$ in each sector (shown for $L=6$). Redder (bluer) colors indicate more positive (negative) values. Dynamics begun from a state in the region $C$ will take exponentially long to escape $C$ due to the bottleneck imposed by the tree structure. {\bf c)} Krylov sector occupations for a quench under $e^{-iHt}$, starting from a frozen state on the edge of the tree. Circle sizes and colors are drawn according to $\lan \psi(t)|\Pi_\mck|\psi(t)\ran$ at time $t=10^3$, where $\Pi_\mck$ projects onto the sector $\mck$.}
	\end{figure*}

	Since states in $\hfroz$ are near the middle of $H$'s spectrum, 
	energy conservation does not present an obstacle for states in $\hfroz$ to thermalize to infinite temperature, even when the size of the impurity region is small.  
	Furthermore, as we will see momentarily, only $O(L)$ applications of $H_{\rm imp}$ are required to connect any two computational basis product states. 
	From these facts, a reasonable prior might be that states in $\hfroz$ rapidly thermalize to a volume-law infinite-temperature state, with this occurring on the time scale needed for the influence of $H_{\rm imp}$ to propagate throughout the full extent of the constrained region. 
	Simulating the dynamics with TEBD for the simplest choice of $N=3$ reveals that this is {\it not} what happens. In Fig.~\ref{fig:inspirational_schematic} we compute the bipartite entanglement entropy $S_A(t)$, with $A = [1,\lc/2]$ half of the constrained region, together with the charge expectation values $\lan Q_a\ran$, with $Q_a$ computed on the sites $[1,\lc]$ of the constrained region. Both of these quantities indicate an {\it exponentially} long thermalization time, with $S_A(t)$ exhibiting a slow logarithmic growth and $\lan Q_a\ran$ a similarly slow decay. Slow thermalization is also observed for initial product states which are mostly frozen but contain a small number of flippable nearest-neighbor pairs; for these states $S_A(t)$ increases quickly at short times $t\lesssim L$ but then grows as $\sim\log(t)$ thereafter. 
	
	\paragraph*{Hilbert space connectivity and random walks:}
	To understand these observations, it will be helpful to have a geometric understanding of how the dynamics acts in Hilbert space. To best illustrate this we will momentarily fix $\ltot=\lc+1$, so that the pair-flip constraint is broken only on the last site of the chain.
	This understanding is obtained by associating each product state $\k{a_1,\dots,a_L}$ with a length-$L$ walk on the $N$-valent tree $T_N$ \cite{caha2018pair}. The walk is determined by reading the product state from left to right: $a_1$ determines the direction of the first step of the walk, $a_2$ the second, and so on. The direction of the walk's travel is fixed by the convention that if two identical labels are encountered in a row ($a_i=a_{i+1}$), the walk backtracks (see Fig.~\ref{fig:trees}~{a} for an illustration with $N=3$). The merit of this is that the allowed processes implementable by $H_0$ are precisely those which {\it preserve walk endpoints}, since $H_0$ acts nontrivially only on locations with backtracks. Each vertex of $T_N$ (more precisely, each vertex of $T_N$'s even / odd sublattice, depending on the parity of $L$) thus defines a disconnected sector of the constrained dynamics. 
	The sectors on the edge of $T_N$ contain precisely those states whose walks have no backtracks; these states thus span $\mch_{\rm froz}$ and define $O((N-1)^L)$ one-dimensional sectors. As shown in \cite{si}, 
 % App.~\ref{app:sector_dims}, 
 the sectors increase exponentially in size towards the center of the tree, with the largest sector at the tree center having dimension $|\mck_{\rm max}|\sim L^{3/2} (2\sqrt{N-1})^L$.

	The action of $H_{\rm imp}$ at the chain end breaks the constraint by allowing the {\it last step} of the walk to be changed. \footnote{{Putting baths at {\it both} ends of the chain allows for both the first and last steps of the walk to be changed, resulting in a significantly more complicated graph structure. In the pair-flip model, this change is in fact significant enough to reduce the thermalization time to $t_{\rm th} = {\rm poly}(L)$; for other strongly-fragmented models $t_{\rm th}$ remains exponential in $L$. See Ref.~\cite{wang2025exponentially} for details.}} This restores ergodicity, connecting the sectors to their nearest neighbors on one sublattice of $T_N$, in the manner shown in Fig.~\ref{fig:trees}~{a}. We will refer to the graph $G_\mck$ of Krylov sectors so obtained as the {\it Krylov graph}.	%
	The dynamics thus induces a random walk on $G_\mck$, and for a system to thermalize, its wavefunction must spread out across the entirety of $G_\mck$ under the action of this walk. 
	
	The tree structure of the Krylov graph suggests that this spreading is slow, since for a vertex of $T_N$ at depth $1<d<L$, there are $N-1$ ways of going ``out'' towards the boundary, but only one way of going ``in'' towards the center. This implies that a simple random walk on $T_N$ will have an ``outward'' bias with velocity $v_N=(N-1)/N-1/N = 1-2/N$, and suggests that the dynamics on $G_\mck$ will take exponentially long to overcome this bias and thermalize. 
	However, this argument ignores the fact that the sectors increase exponentially in size as one moves towards the center of $G_\mck$, which introduces an opposing ``inward'' bias. Our main goal in the following will be to understand which one of these competing effects dominates.

	As a first pass, one can quench a state in $\mch_{\rm froz}$ under $e^{-iHt}$ and numerically compute the weight of $\k{\psi(t)}$ in each sector. 
	Doing so gives weights which are highly clustered around $G_\mck$'s edge even when $t\gg L$, as shown in Fig.~\ref{fig:trees}~{c}. This suggests that the ``outward'' bias---which acts to slow thermalization---wins out. To understand why, we shift our focus from Hamiltonian to RU dynamics, where rigorous bounds on thermalization times can be proven.

	\paragraph*{RU dynamics and Hilbert space bottlenecks:} To simplify the dynamics, we replace the unconstrained region by a thermal bath which subjects the spin on the end of the chain to depolarizing noise. This models the situation where the unconstrained region is taken to be infinitely large (so that it can exchange an arbitrary amount of energy with the constrained region), and then traced out. Since there is no conserved energy in this setup, we will replace time evolution under $H_0$ by a constrained form RU dynamics, consisting of local gates which preserve the pair-flip constraint but are otherwise Haar random.
	The quantum channel implementing one step of the dynamics is 
	\be \label{channel} \mcc_\sfd(\r) = \mcu^\da_\sfd( \Tr_L[\r] \tp \unit/N) \mcu_\sfd,\ee 
	where $\Tr_L[\cdot]$ denotes tracing out the spin at site $L$, and $\mcu_\sfd$ is a random depth-$\sfd$ constrained brickwork circuit. In our analytic arguments we will take $\sfd\gg L$, which simplifies things by making the {\it intra}-sector dynamics thermalize instantaneously. Regardless of $\sfd$, $\mcc_\sfd$ should generically thermalize product states {\it faster} than the Hamiltonian dynamics studied above. We will nevertheless prove that the thermalization time under $\mcc_\sfd$ is {\it exponentially long} in $L$. Detailed proofs of the statements to follow are deferred to the supplementary information \cite{si}; 
 %App.~\ref{app:bounds}, 
    in what follows we will only discuss the most salient aspects. 
	
	Let us first examine the circuit-averaged state $\ob \r_\psi(t) \equiv \oEE_{\{\mcu\}} \mcc_\sfd^t(\proj\psi)$ obtained by evolving a computational basis product state $\k{\psi}$ for time $t$. Using by-now standard techniques \cite{Singh,group_dynamics}, the circuit average can be performed exactly, mapping the RU evolution to a certain kind of Markov process. One finds $\ob\r_\psi(t) = \proj{\mcm^t \psi}$, where $\k{\mcm^t \psi}$ is the state obtained from $\k\psi$ by $t$ applications of a Markov generator of the form $\mcm = \mcm_{PF}\mcm_L$. In this expression $\mcm_L = \unit_{L-1} \tp \frac1N \sum_{a,b} \kb{a}{b}_L$ randomizes the state of the spin on the end of the chain, and $\mcm_{PF}$ is a depth-$\sfd$ stochastic brickwork circuit implementing pair-flip dynamics, each brick being the 2-site gate $\frac1N \sum_{a,b} \kb{aa}{bb} + \sum_{a\neq b} \proj{ab}$ (see \cite{si} 
 % App.~\ref{app:stochastics} 
 for details).

	Let $G_\mch$ be the graph with a vertex for each computational basis product state, and an edge drawn between all pairs of vertices connected by a single step of the dynamics. 
  The Krylov graph $G_\mck$ is thus a ``coarse-grained'' version of $G_\mch$, obtained by gathering all states in a given sector into a single ``supernode'', and merging all connections between states in bath-connected sectors into a single ``superedge''. 
	$\mcm$ implements a random walk on $G_\mch$, with $\r=\unit$ as its unique steady state. Thus all initial states thermalize under $\mcc_\sfd$, provided one waits long enough. How long one must wait is determined by the relaxation time $t_{\rm rel}$, which is inversely proportional to the spectral gap $\De_\mcm$ of $\mcm$ (the difference between $\mcm$'s lowest and second-lowest eigenvalues). As a first result, we prove 
	\begin{theorem} \label{thm:maintextsmallgap}
		$\De_\mcm$ is exponentially small in system size: 
		\be \label{gapbound} \De_\mcm \leq |\mck_{\rm max}| N^{-L} \sim  L^{-3/2} \r_N^L,\ee 
		where $\r_N \equiv 2 \sqrt {N-1}/ N < 1$.
	\end{theorem} 
	
	This result follows from the fact that states evolving under $\mcm$ encounter severe bottlenecks as they move throughout Hilbert space. To see this, define the {\it expansion} $\cp(G)$ of a graph $G$ as \cite{lyons2017probability}
	\be \label{conddef} \cp(G) \equiv \min_{C \, : |C| \leq |G|/2} |\p C| / |C|,\ee 
	where $|\p C|$ denotes the number of edges connecting the subgraph $C$ to $G \setminus C$. Graphs with strong bottlenecks have smaller values of $\cp(G)$, and random walks on graphs with strong bottlenecks mix slowly. This is quantified using Cheeger's inequality \cite{lyons2017probability}, which reads 
	\be \label{cheeger} \frac12 \cp(G_\mch)^2 \leq \De_\mcm \leq 2 \cp(G_\mch).\ee 
    Since $G_\mck$ is formed by coarse-graining $G_\mch$, we have $\Phi(G_\mch) \leq \Phi(G_\mck)$, provided that $|\p C|$ ($|C|$) in \eqref{conddef} are appropriately weighted by the sizes of the superedges (supernodes). 
	The upper bound in \eqref{gapbound} then follows by letting $C$ be one full branch of the tree (shown in Fig.~\ref{fig:trees}~{b} for $N=3$), for which $|\p C| \sim |\mck_{\rm max}|$ and $|C|\sim N^L/N$. This gives an upper bound on $\cp(G_\mck)$---and hence on $\De_\mcm$---scaling as $|\mck_{\rm max}|/N^L$, which is exponentially small in $L$ (exact diagonalization indicates that this bound is in fact saturated \cite{si}).
 % ; see App.~\ref{app:bounds}). 
 This demonstrates the existence of a large bottleneck and shows that the ``outward'' bias discussed above ultimately wins out at long times (the effect of the ``inward'' bias turns out to be to reduce the base of the exponential in \eqref{gapbound} from $1/N$ to $\r_N$). 
	
	An initial state $\k\psi$ chosen randomly from $C$ will thus typically take an exponentially long time to leave $C$, immediately yielding a bound on the circuit-averaged entropy $\ob S_\psi(t) = \oEE_{\{ \mcu\} } S[\mcc_\sfd^t (\proj\psi)]$. To this end, define $t_S(\g) \equiv \min\{t \, : \, \ob S_\psi(t) \geq \g L \ln(N)\}$ as the time at which $\ob S_\psi(t)$ first reaches a fraction $\g$ of its maximal value. We find 	\begin{theorem}\label{thm:maintextsmalls}
		Let $\g$ satisfy $\g_*<\g<1$, $\g_*\equiv 2 (1-v_N\ln(N-1) /  \ln N)$. Then
		\be t_S(\g) \geq C_\g \sqrt L e^{L \l_\g},\ee 
		where $C_\g$ is an unimportant $O(1)$ constant and $\l_\g \equiv \frac12((1-\g/2)\ln(N)/\ln(N-1)-v_N)^2$. 
	\end{theorem} 
	This indicates that when initialized from a typical computational basis product state, the system rapidly reaches an entropy of $S_* = \gamma_* L \ln N$, but that after reaching $S_*$ the entropy growth slows dramatically, taking exponentially long to fully saturate. This can be understood by appealing to the biased random walk introduced above, which implies that almost all product states are located at a depth near $d_* = v_NL$. A state initialized on the boundary of $G_\mck$ may rapidly move inwards to a depth of $d_*$ (during which the ``inward'' bias dominates), but then gets ``stuck'' at $d_*$, where the number of states is largest (and where the ``outward'' bias now dominates).

	Remarkably, the small gap \eqref{gapbound} imprints itself in the expectation values of the charges $Q_a \in [-L/2,L/2]$, despite the fact that expectation values of local operators do not allow one to distinguish different Krylov sectors. This occurs because the pattern of values that $Q_a$ takes on different sectors is strongly anisotropic across $G_\mck$, as illustrated in Fig.~\ref{fig:trees}~{b}. The exponential smallness of $\cp(G_\mch)$ means that the charge of a generic product state will also take exponentially long to relax. 
	The upshot of this is that similarly to Thm.~\ref{thm:maintextsmalls}, we can bound $t_{Q}(\g)$, the time for $\lan Q_a(t)\ran$ to drop below $\g L/2$ when initialized in a state $\k{\psi_{\rm max}}$ of maximal $Q_a$ charge, 
	as follows:
	\begin{theorem}\label{thm:maintextsmallq}
		Let $\g$ satisfy $0<\g<v_N/2$. Then when $\g=\Theta(L^0)$,
		\be \label{chargerel}	t_{Q}(\g) \geq D_\g \sqrt L e^{L (2\g - v_N)^2 /2} \ee 
		where $D_\g$ is another unimportant $O(1)$ constant. In the $\gamma \ra 0$ limit, $t_{Q}(\g) \geq D_\g / \Phi(G_\mch)$. 
	\end{theorem}

	\begin{figure}
		\centering 
		\includegraphics[width=0.48\textwidth]{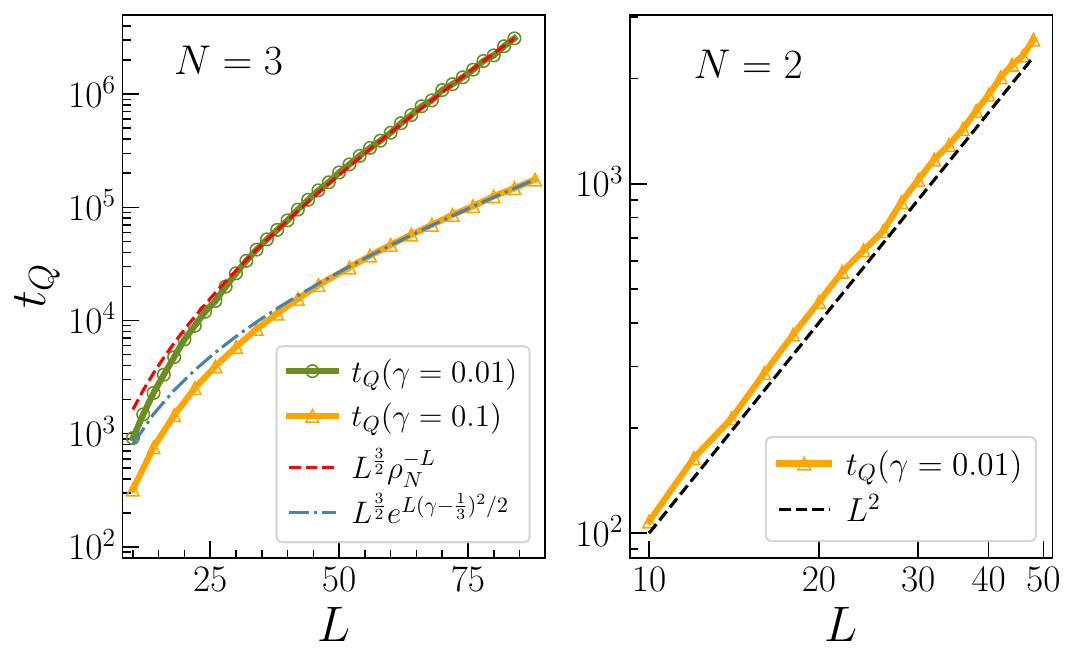}
		\caption{\label{fig:T_relax} The charge relaxation time $t_Q$ for the maximal-$Q_1$ initial state $|\psi(0)\rangle=|21\ran^{\tp L/2}$. {\it Left:} Slow charge relaxation for $N=3$: $t_Q$ for $\gamma=0.1$ and $\gamma=0.01$. The former fits well to $L$ times the bound in \eqref{chargerel} (blue dot-dashed line), while the latter saturates the $\g\ra0$ bound of $\sim 1/\Phi(G_\mch)$ (red dashed line). 
			% are normalized to fit the last data points. 
			{\it Right:} Fast relaxation for $N=2$: $t_Q$ on a log-log scale, showing clear diffusive behavior.}
	\end{figure}
	
	The above theorems were all stated for the case of $N>2$. When $N=2$, the constrained dynamics is not strongly fragmented: it instead possesses only $O(L)$ sectors, the largest of which has a dimension smaller than $2^L$ by only a factor of $1/\sqrt L$. These sectors are connected by the bath to form a 2-valent tree (i.e. a line), yielding $\cp(G_\mch) \sim 1/\sqrt L$. The $N=2$ dynamics thus possesses only a weak bottleneck, and
	thermalizes exponentially faster than the $N\geq 3$ models.  

	These results are corroborated in Fig.~\ref{fig:T_relax} by numerically simulating charge relaxation in the classical stochastic model defined by $\mcm$. In these simulations we fix the depth of each RU brickwork layer at $\sfd = 2$, constraining the intra-sector dynamics by spatial locality. On general grounds we expect this to decrease $\De_{\mcm}$ by a factor of $L\inv$, 
 % (which we confirm numerically in App.~\ref{app:bounds}), 
 so that Cheeger's inequality for $N=2$ reads $L^{-2}\lesssim \De_\mcm \lesssim L^{-3/2}$. For $N=2$ our numerics give a diffusive $t_Q\sim L^{-2}$, matching this lower bound. For $N=3$ and $\g = \ct(L^0)$, $t_Q(\g)$ is similarly observed to exceed the lower bound of Thm.~\ref{thm:maintextsmallq} by a factor of $L$. Interestingly however, when $\g\ra0$ the lower bound of $1/\cp(G_\mch)\sim L^{3/2}\r_3^{-L}$ appears to be quite nearly saturated, without an extra factor of $L$.
	% we simulate the relaxation time for $\gamma=0.1$ and $\gamma=0.01$. We find that when $\gamma=0.1$, $t_{Q_a}\sim L^{3/2} e^{L(\gamma-1/3)^2/2}$, which is the lower bound in \eqref{chargerel} multiplied by $L$. Interestingly, $t_{Q_a}(\gamma=0.01)\sim L^{3/2}\rho_3^{-L}$ which is exactly the lower bound of $t_{\rm rel}\equiv 1/\Delta_M$. This is consistent with the $\gamma\to 0$ lower bound limit of $t_{Q_a}$ (see App.~\ref{app:bounds}). 
	Finally, we note that when $N=3$, the base of the exponent in \eqref{chargerel} is only $e^{1/18}\approx 1.06$ for small $\g$; this makes the ${\rm poly}(L)$ contribution to $t_Q(\g)$ dominate for moderate system sizes. 
	% which when combined with the presence of a ${\rm poly}(L)$ prefactor frustrates a decisive numerical observation of the exponential scaling at easily accessible system sizes. 
	In the future it would be interesting to explore models with stronger bottlenecks, where the exponential scaling may be more pronounced. 
	
	\paragraph*{Discussion:} In this work, we have seen how HSF can remain robust in the presence of local coupling to a thermal bath, with signatures of fragmentation remaining present even on times exponentially long in system size {(provided the constraints are unbroken away from the bath)}. Our results have focused on models with pair-flip constraints, where thermalization is arrested by the presence of strong bottlenecks encountered by the dynamics as it explores Hilbert space. {However, in Ref.~\cite{wang2025exponentially}, we show that this graph-theoretic picture of thermalization generalizes to a very broad family of strongly-fragmented systems, including all examples currently known in the literature, from the dipole-conserving models of Refs.~\cite{sala2020ergodicity,khemani2020localization} to the general family of semigroup-based dynamics introduced in Ref.~\cite{group_dynamics}. In Ref.~\cite{wang2025exponentially}, strong evidence is built for the conjecture that in {\it all} models of strongly-fragmented dynamics, it always takes at least exponentially long (in system size) for a bath at one end of a chain to thermalize the entire system. This conjecture is argued for based on the ubiquity of the expanding structure of Krylov graphs in strongly-fragmented models, and can be rigorously proven to hold for all of the examples of strongly-fragmented dynamics known to the authors (a proof in full generality appears to require resolving an open conjecture in the theory of expander graphs).}
    
    By viewing one half of a fully-constrained system as a bath for the other half, our results imply that bottlenecks inevitably produce slow {\it intra}-sector dynamics in unperturbed models, a result which may be useful for understanding the slow charge transport of Refs.~\cite{group_dynamics,brighi2023hilbert}. 
	
	Our analytic results have focused on the case of RU dynamics. Hamiltonian dynamics should generically thermalize at least as slowly, although
 there are other reasons why disorder-free Hamiltonian dynamics may be parametrically slower,
 % indeed the numerics of App.~\ref{app:numerics} suggest that eigenstates may localize on the Krylov graph, 
 and in the future it would be interesting to investigate this possibility in more detail. 
	It would also be interesting to understand couplings to different types of baths and the resulting steady states, as in Ref.~\cite{li2023hilbert}. 
 % Finally, it would be valuable to investigate which class of bottlenecked dynamics is most easily realized in cold atom systems, which provide the most natural route towards experimentally studying the phenomena investigated in this work. 
	
	\paragraph*{Acknowledgments:} E.L. is grateful to Shankar Balasubramanian, Sarang Gopalakrishnan, and Alexey Khudorozhkov for stimulating discussions and collaborations on a related project. He also thanks Ehud Altman, Soonwon Choi, Ilya Gruzberg, and the residents of the 5th floor of Birge Hall, Sam Garratt in particular, for helpful discussions. Y.H. thanks Hyeongjin Kim, Rahul Nandkishore, Anatoli Polkovnikov and Zhi-Cheng Yang for discussions. Some simulations were performed with the help of the ITensor library \cite{fishman2022itensor}. E.L. was supported by a Miller research fellowship. This research is supported in part by the National Science Foundation under Grant No. DMR-2219735 (Y. H. and X. C.).
	
	\bibliography{hsf_impurities_bib}

%apsrev4-2.bst 2019-01-14 (MD) hand-edited version of apsrev4-1.bst
%Control: key (0)
%Control: author (8) initials jnrlst
%Control: editor formatted (1) identically to author
%Control: production of article title (0) allowed
%Control: page (0) single
%Control: year (1) truncated
%Control: production of eprint (0) enabled
\begin{thebibliography}{48}%
\makeatletter
\providecommand \@ifxundefined [1]{%
 \@ifx{#1\undefined}
}%
\providecommand \@ifnum [1]{%
 \ifnum #1\expandafter \@firstoftwo
 \else \expandafter \@secondoftwo
 \fi
}%
\providecommand \@ifx [1]{%
 \ifx #1\expandafter \@firstoftwo
 \else \expandafter \@secondoftwo
 \fi
}%
\providecommand \natexlab [1]{#1}%
\providecommand \enquote  [1]{``#1''}%
\providecommand \bibnamefont  [1]{#1}%
\providecommand \bibfnamefont [1]{#1}%
\providecommand \citenamefont [1]{#1}%
\providecommand \href@noop [0]{\@secondoftwo}%
\providecommand \href [0]{\begingroup \@sanitize@url \@href}%
\providecommand \@href[1]{\@@startlink{#1}\@@href}%
\providecommand \@@href[1]{\endgroup#1\@@endlink}%
\providecommand \@sanitize@url [0]{\catcode `\\12\catcode `\$12\catcode `\&12\catcode `\#12\catcode `\^12\catcode `\_12\catcode `\%12\relax}%
\providecommand \@@startlink[1]{}%
\providecommand \@@endlink[0]{}%
\providecommand \url  [0]{\begingroup\@sanitize@url \@url }%
\providecommand \@url [1]{\endgroup\@href {#1}{\urlprefix }}%
\providecommand \urlprefix  [0]{URL }%
\providecommand \Eprint [0]{\href }%
\providecommand \doibase [0]{https://doi.org/}%
\providecommand \selectlanguage [0]{\@gobble}%
\providecommand \bibinfo  [0]{\@secondoftwo}%
\providecommand \bibfield  [0]{\@secondoftwo}%
\providecommand \translation [1]{[#1]}%
\providecommand \BibitemOpen [0]{}%
\providecommand \bibitemStop [0]{}%
\providecommand \bibitemNoStop [0]{.\EOS\space}%
\providecommand \EOS [0]{\spacefactor3000\relax}%
\providecommand \BibitemShut  [1]{\csname bibitem#1\endcsname}%
\let\auto@bib@innerbib\@empty
%</preamble>
\bibitem [{\citenamefont {Deutsch}(1991)}]{deutsch1991quantum}%
  \BibitemOpen
  \bibfield  {author} {\bibinfo {author} {\bibfnamefont {J.~M.}\ \bibnamefont {Deutsch}},\ }\bibfield  {title} {\bibinfo {title} {Quantum statistical mechanics in a closed system},\ }\href {https://doi.org/10.1103/PhysRevA.43.2046} {\bibfield  {journal} {\bibinfo  {journal} {Phys. Rev. A}\ }\textbf {\bibinfo {volume} {43}},\ \bibinfo {pages} {2046} (\bibinfo {year} {1991})}\BibitemShut {NoStop}%
\bibitem [{\citenamefont {Srednicki}(1994)}]{srednicki1994chaos}%
  \BibitemOpen
  \bibfield  {author} {\bibinfo {author} {\bibfnamefont {M.}~\bibnamefont {Srednicki}},\ }\bibfield  {title} {\bibinfo {title} {Chaos and quantum thermalization},\ }\href {https://doi.org/10.1103/PhysRevE.50.888} {\bibfield  {journal} {\bibinfo  {journal} {Phys. Rev. E}\ }\textbf {\bibinfo {volume} {50}},\ \bibinfo {pages} {888} (\bibinfo {year} {1994})}\BibitemShut {NoStop}%
\bibitem [{\citenamefont {Rigol}\ \emph {et~al.}(2008)\citenamefont {Rigol}, \citenamefont {Dunjko},\ and\ \citenamefont {Olshanii}}]{rigol2008thermalization}%
  \BibitemOpen
  \bibfield  {author} {\bibinfo {author} {\bibfnamefont {M.}~\bibnamefont {Rigol}}, \bibinfo {author} {\bibfnamefont {V.}~\bibnamefont {Dunjko}},\ and\ \bibinfo {author} {\bibfnamefont {M.}~\bibnamefont {Olshanii}},\ }\bibfield  {title} {\bibinfo {title} {Thermalization and its mechanism for generic isolated quantum systems},\ }\href@noop {} {\bibfield  {journal} {\bibinfo  {journal} {Nature}\ }\textbf {\bibinfo {volume} {452}},\ \bibinfo {pages} {854} (\bibinfo {year} {2008})}\BibitemShut {NoStop}%
\bibitem [{\citenamefont {Nandkishore}\ and\ \citenamefont {Huse}(2015)}]{nandkishore2015many-body}%
  \BibitemOpen
  \bibfield  {author} {\bibinfo {author} {\bibfnamefont {R.}~\bibnamefont {Nandkishore}}\ and\ \bibinfo {author} {\bibfnamefont {D.~A.}\ \bibnamefont {Huse}},\ }\bibfield  {title} {\bibinfo {title} {Many-body localization and thermalization in quantum statistical mechanics},\ }\href {https://doi.org/10.1146/annurev-conmatphys-031214-014726} {\bibfield  {journal} {\bibinfo  {journal} {Annual Review of Condensed Matter Physics}\ }\textbf {\bibinfo {volume} {6}},\ \bibinfo {pages} {15} (\bibinfo {year} {2015})}\BibitemShut {NoStop}%
\bibitem [{\citenamefont {D'Alessio}\ \emph {et~al.}(2016)\citenamefont {D'Alessio}, \citenamefont {Kafri}, \citenamefont {Polkovnikov},\ and\ \citenamefont {Rigol}}]{dalessio2016from}%
  \BibitemOpen
  \bibfield  {author} {\bibinfo {author} {\bibfnamefont {L.}~\bibnamefont {D'Alessio}}, \bibinfo {author} {\bibfnamefont {Y.}~\bibnamefont {Kafri}}, \bibinfo {author} {\bibfnamefont {A.}~\bibnamefont {Polkovnikov}},\ and\ \bibinfo {author} {\bibfnamefont {M.}~\bibnamefont {Rigol}},\ }\bibfield  {title} {\bibinfo {title} {From quantum chaos and eigenstate thermalization to statistical mechanics and thermodynamics},\ }\href {https://doi.org/10.1080/00018732.2016.1198134} {\bibfield  {journal} {\bibinfo  {journal} {Advances in Physics}\ }\textbf {\bibinfo {volume} {65}},\ \bibinfo {pages} {239} (\bibinfo {year} {2016})}\BibitemShut {NoStop}%
\bibitem [{\citenamefont {Alet}\ and\ \citenamefont {Laflorencie}(2018)}]{alet2018many-body}%
  \BibitemOpen
  \bibfield  {author} {\bibinfo {author} {\bibfnamefont {F.}~\bibnamefont {Alet}}\ and\ \bibinfo {author} {\bibfnamefont {N.}~\bibnamefont {Laflorencie}},\ }\bibfield  {title} {\bibinfo {title} {Many-body localization: An introduction and selected topics},\ }\href {https://doi.org/https://doi.org/10.1016/j.crhy.2018.03.003} {\bibfield  {journal} {\bibinfo  {journal} {Comptes Rendus Physique}\ }\textbf {\bibinfo {volume} {19}},\ \bibinfo {pages} {498} (\bibinfo {year} {2018})},\ \bibinfo {note} {quantum simulation / Simulation quantique}\BibitemShut {NoStop}%
\bibitem [{\citenamefont {Abanin}\ \emph {et~al.}(2019)\citenamefont {Abanin}, \citenamefont {Altman}, \citenamefont {Bloch},\ and\ \citenamefont {Serbyn}}]{abanin2019colloquium}%
  \BibitemOpen
  \bibfield  {author} {\bibinfo {author} {\bibfnamefont {D.~A.}\ \bibnamefont {Abanin}}, \bibinfo {author} {\bibfnamefont {E.}~\bibnamefont {Altman}}, \bibinfo {author} {\bibfnamefont {I.}~\bibnamefont {Bloch}},\ and\ \bibinfo {author} {\bibfnamefont {M.}~\bibnamefont {Serbyn}},\ }\bibfield  {title} {\bibinfo {title} {Colloquium: Many-body localization, thermalization, and entanglement},\ }\href {https://doi.org/10.1103/RevModPhys.91.021001} {\bibfield  {journal} {\bibinfo  {journal} {Rev. Mod. Phys.}\ }\textbf {\bibinfo {volume} {91}},\ \bibinfo {pages} {021001} (\bibinfo {year} {2019})}\BibitemShut {NoStop}%
\bibitem [{\citenamefont {Sala}\ \emph {et~al.}(2020)\citenamefont {Sala}, \citenamefont {Rakovszky}, \citenamefont {Verresen}, \citenamefont {Knap},\ and\ \citenamefont {Pollmann}}]{sala2020ergodicity}%
  \BibitemOpen
  \bibfield  {author} {\bibinfo {author} {\bibfnamefont {P.}~\bibnamefont {Sala}}, \bibinfo {author} {\bibfnamefont {T.}~\bibnamefont {Rakovszky}}, \bibinfo {author} {\bibfnamefont {R.}~\bibnamefont {Verresen}}, \bibinfo {author} {\bibfnamefont {M.}~\bibnamefont {Knap}},\ and\ \bibinfo {author} {\bibfnamefont {F.}~\bibnamefont {Pollmann}},\ }\bibfield  {title} {\bibinfo {title} {Ergodicity breaking arising from hilbert space fragmentation in dipole-conserving hamiltonians},\ }\href {https://doi.org/10.1103/PhysRevX.10.011047} {\bibfield  {journal} {\bibinfo  {journal} {Phys. Rev. X}\ }\textbf {\bibinfo {volume} {10}},\ \bibinfo {pages} {011047} (\bibinfo {year} {2020})}\BibitemShut {NoStop}%
\bibitem [{\citenamefont {Khemani}\ \emph {et~al.}(2020)\citenamefont {Khemani}, \citenamefont {Hermele},\ and\ \citenamefont {Nandkishore}}]{khemani2020localization}%
  \BibitemOpen
  \bibfield  {author} {\bibinfo {author} {\bibfnamefont {V.}~\bibnamefont {Khemani}}, \bibinfo {author} {\bibfnamefont {M.}~\bibnamefont {Hermele}},\ and\ \bibinfo {author} {\bibfnamefont {R.}~\bibnamefont {Nandkishore}},\ }\bibfield  {title} {\bibinfo {title} {Localization from hilbert space shattering: From theory to physical realizations},\ }\href {https://doi.org/10.1103/PhysRevB.101.174204} {\bibfield  {journal} {\bibinfo  {journal} {Phys. Rev. B}\ }\textbf {\bibinfo {volume} {101}},\ \bibinfo {pages} {174204} (\bibinfo {year} {2020})}\BibitemShut {NoStop}%
\bibitem [{\citenamefont {Moudgalya}\ and\ \citenamefont {Motrunich}(2022{\natexlab{a}})}]{motrunich}%
  \BibitemOpen
  \bibfield  {author} {\bibinfo {author} {\bibfnamefont {S.}~\bibnamefont {Moudgalya}}\ and\ \bibinfo {author} {\bibfnamefont {O.~I.}\ \bibnamefont {Motrunich}},\ }\bibfield  {title} {\bibinfo {title} {Hilbert space fragmentation and commutant algebras},\ }\href {https://doi.org/10.1103/PhysRevX.12.011050} {\bibfield  {journal} {\bibinfo  {journal} {Phys. Rev. X}\ }\textbf {\bibinfo {volume} {12}},\ \bibinfo {pages} {011050} (\bibinfo {year} {2022}{\natexlab{a}})}\BibitemShut {NoStop}%
\bibitem [{\citenamefont {Pancotti}\ \emph {et~al.}(2020)\citenamefont {Pancotti}, \citenamefont {Giudice}, \citenamefont {Cirac}, \citenamefont {Garrahan},\ and\ \citenamefont {Ba{\~n}uls}}]{pancotti2020quantum}%
  \BibitemOpen
  \bibfield  {author} {\bibinfo {author} {\bibfnamefont {N.}~\bibnamefont {Pancotti}}, \bibinfo {author} {\bibfnamefont {G.}~\bibnamefont {Giudice}}, \bibinfo {author} {\bibfnamefont {J.~I.}\ \bibnamefont {Cirac}}, \bibinfo {author} {\bibfnamefont {J.~P.}\ \bibnamefont {Garrahan}},\ and\ \bibinfo {author} {\bibfnamefont {M.~C.}\ \bibnamefont {Ba{\~n}uls}},\ }\bibfield  {title} {\bibinfo {title} {Quantum east model: Localization, nonthermal eigenstates, and slow dynamics},\ }\href@noop {} {\bibfield  {journal} {\bibinfo  {journal} {Physical Review X}\ }\textbf {\bibinfo {volume} {10}},\ \bibinfo {pages} {021051} (\bibinfo {year} {2020})}\BibitemShut {NoStop}%
\bibitem [{\citenamefont {van Horssen}\ \emph {et~al.}(2015)\citenamefont {van Horssen}, \citenamefont {Levi},\ and\ \citenamefont {Garrahan}}]{van2015dynamics}%
  \BibitemOpen
  \bibfield  {author} {\bibinfo {author} {\bibfnamefont {M.}~\bibnamefont {van Horssen}}, \bibinfo {author} {\bibfnamefont {E.}~\bibnamefont {Levi}},\ and\ \bibinfo {author} {\bibfnamefont {J.~P.}\ \bibnamefont {Garrahan}},\ }\bibfield  {title} {\bibinfo {title} {Dynamics of many-body localization in a translation-invariant quantum glass model},\ }\href@noop {} {\bibfield  {journal} {\bibinfo  {journal} {Physical Review B}\ }\textbf {\bibinfo {volume} {92}},\ \bibinfo {pages} {100305} (\bibinfo {year} {2015})}\BibitemShut {NoStop}%
\bibitem [{\citenamefont {Mukherjee}\ \emph {et~al.}(2021)\citenamefont {Mukherjee}, \citenamefont {Banerjee}, \citenamefont {Sengupta},\ and\ \citenamefont {Sen}}]{mukherjee2021minimal}%
  \BibitemOpen
  \bibfield  {author} {\bibinfo {author} {\bibfnamefont {B.}~\bibnamefont {Mukherjee}}, \bibinfo {author} {\bibfnamefont {D.}~\bibnamefont {Banerjee}}, \bibinfo {author} {\bibfnamefont {K.}~\bibnamefont {Sengupta}},\ and\ \bibinfo {author} {\bibfnamefont {A.}~\bibnamefont {Sen}},\ }\bibfield  {title} {\bibinfo {title} {Minimal model for hilbert space fragmentation with local constraints},\ }\href@noop {} {\bibfield  {journal} {\bibinfo  {journal} {Physical Review B}\ }\textbf {\bibinfo {volume} {104}},\ \bibinfo {pages} {155117} (\bibinfo {year} {2021})}\BibitemShut {NoStop}%
\bibitem [{\citenamefont {Yang}\ \emph {et~al.}(2020)\citenamefont {Yang}, \citenamefont {Liu}, \citenamefont {Gorshkov},\ and\ \citenamefont {Iadecola}}]{yang2020hilbert}%
  \BibitemOpen
  \bibfield  {author} {\bibinfo {author} {\bibfnamefont {Z.-C.}\ \bibnamefont {Yang}}, \bibinfo {author} {\bibfnamefont {F.}~\bibnamefont {Liu}}, \bibinfo {author} {\bibfnamefont {A.~V.}\ \bibnamefont {Gorshkov}},\ and\ \bibinfo {author} {\bibfnamefont {T.}~\bibnamefont {Iadecola}},\ }\bibfield  {title} {\bibinfo {title} {Hilbert-space fragmentation from strict confinement},\ }\href@noop {} {\bibfield  {journal} {\bibinfo  {journal} {Physical review letters}\ }\textbf {\bibinfo {volume} {124}},\ \bibinfo {pages} {207602} (\bibinfo {year} {2020})}\BibitemShut {NoStop}%
\bibitem [{\citenamefont {Brighi}\ \emph {et~al.}(2023)\citenamefont {Brighi}, \citenamefont {Ljubotina},\ and\ \citenamefont {Serbyn}}]{brighi2023hilbert}%
  \BibitemOpen
  \bibfield  {author} {\bibinfo {author} {\bibfnamefont {P.}~\bibnamefont {Brighi}}, \bibinfo {author} {\bibfnamefont {M.}~\bibnamefont {Ljubotina}},\ and\ \bibinfo {author} {\bibfnamefont {M.}~\bibnamefont {Serbyn}},\ }\bibfield  {title} {\bibinfo {title} {Hilbert space fragmentation and slow dynamics in particle-conserving quantum east models},\ }\href@noop {} {\bibfield  {journal} {\bibinfo  {journal} {SciPost Physics}\ }\textbf {\bibinfo {volume} {15}},\ \bibinfo {pages} {093} (\bibinfo {year} {2023})}\BibitemShut {NoStop}%
\bibitem [{\citenamefont {Lan}\ \emph {et~al.}(2018)\citenamefont {Lan}, \citenamefont {van Horssen}, \citenamefont {Powell},\ and\ \citenamefont {Garrahan}}]{lan2018quantum}%
  \BibitemOpen
  \bibfield  {author} {\bibinfo {author} {\bibfnamefont {Z.}~\bibnamefont {Lan}}, \bibinfo {author} {\bibfnamefont {M.}~\bibnamefont {van Horssen}}, \bibinfo {author} {\bibfnamefont {S.}~\bibnamefont {Powell}},\ and\ \bibinfo {author} {\bibfnamefont {J.~P.}\ \bibnamefont {Garrahan}},\ }\bibfield  {title} {\bibinfo {title} {Quantum slow relaxation and metastability due to dynamical constraints},\ }\href@noop {} {\bibfield  {journal} {\bibinfo  {journal} {Physical review letters}\ }\textbf {\bibinfo {volume} {121}},\ \bibinfo {pages} {040603} (\bibinfo {year} {2018})}\BibitemShut {NoStop}%
\bibitem [{\citenamefont {Garrahan}(2018)}]{garrahan2018aspects}%
  \BibitemOpen
  \bibfield  {author} {\bibinfo {author} {\bibfnamefont {J.~P.}\ \bibnamefont {Garrahan}},\ }\bibfield  {title} {\bibinfo {title} {Aspects of non-equilibrium in classical and quantum systems: Slow relaxation and glasses, dynamical large deviations, quantum non-ergodicity, and open quantum dynamics},\ }\href@noop {} {\bibfield  {journal} {\bibinfo  {journal} {Physica A: Statistical Mechanics and its Applications}\ }\textbf {\bibinfo {volume} {504}},\ \bibinfo {pages} {130} (\bibinfo {year} {2018})}\BibitemShut {NoStop}%
\bibitem [{\citenamefont {De~Tomasi}\ \emph {et~al.}(2019)\citenamefont {De~Tomasi}, \citenamefont {Hetterich}, \citenamefont {Sala},\ and\ \citenamefont {Pollmann}}]{detomasi2019dynamics}%
  \BibitemOpen
  \bibfield  {author} {\bibinfo {author} {\bibfnamefont {G.}~\bibnamefont {De~Tomasi}}, \bibinfo {author} {\bibfnamefont {D.}~\bibnamefont {Hetterich}}, \bibinfo {author} {\bibfnamefont {P.}~\bibnamefont {Sala}},\ and\ \bibinfo {author} {\bibfnamefont {F.}~\bibnamefont {Pollmann}},\ }\bibfield  {title} {\bibinfo {title} {Dynamics of strongly interacting systems: From fock-space fragmentation to many-body localization},\ }\href {https://doi.org/10.1103/PhysRevB.100.214313} {\bibfield  {journal} {\bibinfo  {journal} {Phys. Rev. B}\ }\textbf {\bibinfo {volume} {100}},\ \bibinfo {pages} {214313} (\bibinfo {year} {2019})}\BibitemShut {NoStop}%
\bibitem [{\citenamefont {Moudgalya}\ \emph {et~al.}(2022)\citenamefont {Moudgalya}, \citenamefont {Prem}, \citenamefont {Nandkishore}, \citenamefont {Regnault},\ and\ \citenamefont {Bernevig}}]{moudgalya2022thermalization}%
  \BibitemOpen
  \bibfield  {author} {\bibinfo {author} {\bibfnamefont {S.}~\bibnamefont {Moudgalya}}, \bibinfo {author} {\bibfnamefont {A.}~\bibnamefont {Prem}}, \bibinfo {author} {\bibfnamefont {R.}~\bibnamefont {Nandkishore}}, \bibinfo {author} {\bibfnamefont {N.}~\bibnamefont {Regnault}},\ and\ \bibinfo {author} {\bibfnamefont {B.~A.}\ \bibnamefont {Bernevig}},\ }\bibfield  {title} {\bibinfo {title} {Thermalization and its absence within krylov subspaces of a constrained hamiltonian},\ }in\ \href@noop {} {\emph {\bibinfo {booktitle} {Memorial Volume for Shoucheng Zhang}}}\ (\bibinfo  {publisher} {World Scientific},\ \bibinfo {year} {2022})\ pp.\ \bibinfo {pages} {147--209}\BibitemShut {NoStop}%
\bibitem [{\citenamefont {Will}\ \emph {et~al.}(2023)\citenamefont {Will}, \citenamefont {Moessner},\ and\ \citenamefont {Pollmann}}]{will2023realization}%
  \BibitemOpen
  \bibfield  {author} {\bibinfo {author} {\bibfnamefont {M.}~\bibnamefont {Will}}, \bibinfo {author} {\bibfnamefont {R.}~\bibnamefont {Moessner}},\ and\ \bibinfo {author} {\bibfnamefont {F.}~\bibnamefont {Pollmann}},\ }\bibfield  {title} {\bibinfo {title} {Realization of hilbert space fragmentation and fracton dynamics in 2d},\ }\href@noop {} {\bibfield  {journal} {\bibinfo  {journal} {arXiv preprint arXiv:2311.05695}\ } (\bibinfo {year} {2023})}\BibitemShut {NoStop}%
\bibitem [{\citenamefont {Kwan}\ \emph {et~al.}(2023)\citenamefont {Kwan}, \citenamefont {Wilhelm}, \citenamefont {Biswas},\ and\ \citenamefont {Parameswaran}}]{kwan2023minimal}%
  \BibitemOpen
  \bibfield  {author} {\bibinfo {author} {\bibfnamefont {Y.~H.}\ \bibnamefont {Kwan}}, \bibinfo {author} {\bibfnamefont {P.~H.}\ \bibnamefont {Wilhelm}}, \bibinfo {author} {\bibfnamefont {S.}~\bibnamefont {Biswas}},\ and\ \bibinfo {author} {\bibfnamefont {S.}~\bibnamefont {Parameswaran}},\ }\bibfield  {title} {\bibinfo {title} {Minimal hubbard models of maximal hilbert space fragmentation},\ }\href@noop {} {\bibfield  {journal} {\bibinfo  {journal} {arXiv preprint arXiv:2304.02669}\ } (\bibinfo {year} {2023})}\BibitemShut {NoStop}%
\bibitem [{\citenamefont {Li}\ \emph {et~al.}(2023)\citenamefont {Li}, \citenamefont {Sala},\ and\ \citenamefont {Pollmann}}]{li2023hilbert}%
  \BibitemOpen
  \bibfield  {author} {\bibinfo {author} {\bibfnamefont {Y.}~\bibnamefont {Li}}, \bibinfo {author} {\bibfnamefont {P.}~\bibnamefont {Sala}},\ and\ \bibinfo {author} {\bibfnamefont {F.}~\bibnamefont {Pollmann}},\ }\bibfield  {title} {\bibinfo {title} {Hilbert space fragmentation in open quantum systems},\ }\href@noop {} {\bibfield  {journal} {\bibinfo  {journal} {Physical Review Research}\ }\textbf {\bibinfo {volume} {5}} (\bibinfo {year} {2023})}\BibitemShut {NoStop}%
\bibitem [{\citenamefont {Frey}\ \emph {et~al.}(2024)\citenamefont {Frey}, \citenamefont {Mikhail}, \citenamefont {Rachel},\ and\ \citenamefont {Hackl}}]{frey2024probing}%
  \BibitemOpen
  \bibfield  {author} {\bibinfo {author} {\bibfnamefont {P.}~\bibnamefont {Frey}}, \bibinfo {author} {\bibfnamefont {D.}~\bibnamefont {Mikhail}}, \bibinfo {author} {\bibfnamefont {S.}~\bibnamefont {Rachel}},\ and\ \bibinfo {author} {\bibfnamefont {L.}~\bibnamefont {Hackl}},\ }\bibfield  {title} {\bibinfo {title} {Probing hilbert space fragmentation and the block inverse participation ratio},\ }\href@noop {} {\bibfield  {journal} {\bibinfo  {journal} {Physical Review B}\ }\textbf {\bibinfo {volume} {109}},\ \bibinfo {pages} {064302} (\bibinfo {year} {2024})}\BibitemShut {NoStop}%
\bibitem [{\citenamefont {Kohlert}\ \emph {et~al.}(2023)\citenamefont {Kohlert}, \citenamefont {Scherg}, \citenamefont {Sala}, \citenamefont {Pollmann}, \citenamefont {Hebbe~Madhusudhana}, \citenamefont {Bloch},\ and\ \citenamefont {Aidelsburger}}]{kohlert2023exploring}%
  \BibitemOpen
  \bibfield  {author} {\bibinfo {author} {\bibfnamefont {T.}~\bibnamefont {Kohlert}}, \bibinfo {author} {\bibfnamefont {S.}~\bibnamefont {Scherg}}, \bibinfo {author} {\bibfnamefont {P.}~\bibnamefont {Sala}}, \bibinfo {author} {\bibfnamefont {F.}~\bibnamefont {Pollmann}}, \bibinfo {author} {\bibfnamefont {B.}~\bibnamefont {Hebbe~Madhusudhana}}, \bibinfo {author} {\bibfnamefont {I.}~\bibnamefont {Bloch}},\ and\ \bibinfo {author} {\bibfnamefont {M.}~\bibnamefont {Aidelsburger}},\ }\bibfield  {title} {\bibinfo {title} {Exploring the regime of fragmentation in strongly tilted fermi-hubbard chains},\ }\href {https://doi.org/10.1103/PhysRevLett.130.010201} {\bibfield  {journal} {\bibinfo  {journal} {Phys. Rev. Lett.}\ }\textbf {\bibinfo {volume} {130}},\ \bibinfo {pages} {010201} (\bibinfo {year} {2023})}\BibitemShut {NoStop}%
\bibitem [{\citenamefont {Kinoshita}\ \emph {et~al.}(2006)\citenamefont {Kinoshita}, \citenamefont {Wenger},\ and\ \citenamefont {Weiss}}]{kinoshita2006quantum}%
  \BibitemOpen
  \bibfield  {author} {\bibinfo {author} {\bibfnamefont {T.}~\bibnamefont {Kinoshita}}, \bibinfo {author} {\bibfnamefont {T.}~\bibnamefont {Wenger}},\ and\ \bibinfo {author} {\bibfnamefont {D.~S.}\ \bibnamefont {Weiss}},\ }\bibfield  {title} {\bibinfo {title} {A quantum newton's cradle},\ }\href@noop {} {\bibfield  {journal} {\bibinfo  {journal} {Nature}\ }\textbf {\bibinfo {volume} {440}},\ \bibinfo {pages} {900} (\bibinfo {year} {2006})}\BibitemShut {NoStop}%
\bibitem [{\citenamefont {Scheie}\ \emph {et~al.}(2021)\citenamefont {Scheie}, \citenamefont {Sherman}, \citenamefont {Dupont}, \citenamefont {Nagler}, \citenamefont {Stone}, \citenamefont {Granroth}, \citenamefont {Moore},\ and\ \citenamefont {Tennant}}]{scheie2021detection}%
  \BibitemOpen
  \bibfield  {author} {\bibinfo {author} {\bibfnamefont {A.}~\bibnamefont {Scheie}}, \bibinfo {author} {\bibfnamefont {N.}~\bibnamefont {Sherman}}, \bibinfo {author} {\bibfnamefont {M.}~\bibnamefont {Dupont}}, \bibinfo {author} {\bibfnamefont {S.}~\bibnamefont {Nagler}}, \bibinfo {author} {\bibfnamefont {M.}~\bibnamefont {Stone}}, \bibinfo {author} {\bibfnamefont {G.}~\bibnamefont {Granroth}}, \bibinfo {author} {\bibfnamefont {J.}~\bibnamefont {Moore}},\ and\ \bibinfo {author} {\bibfnamefont {D.}~\bibnamefont {Tennant}},\ }\bibfield  {title} {\bibinfo {title} {Detection of kardar--parisi--zhang hydrodynamics in a quantum heisenberg spin-1/2 chain},\ }\href@noop {} {\bibfield  {journal} {\bibinfo  {journal} {Nature Physics}\ }\textbf {\bibinfo {volume} {17}},\ \bibinfo {pages} {726} (\bibinfo {year} {2021})}\BibitemShut {NoStop}%
\bibitem [{\citenamefont {Malvania}\ \emph {et~al.}(2021)\citenamefont {Malvania}, \citenamefont {Zhang}, \citenamefont {Le}, \citenamefont {Dubail}, \citenamefont {Rigol},\ and\ \citenamefont {Weiss}}]{malvania2021generalized}%
  \BibitemOpen
  \bibfield  {author} {\bibinfo {author} {\bibfnamefont {N.}~\bibnamefont {Malvania}}, \bibinfo {author} {\bibfnamefont {Y.}~\bibnamefont {Zhang}}, \bibinfo {author} {\bibfnamefont {Y.}~\bibnamefont {Le}}, \bibinfo {author} {\bibfnamefont {J.}~\bibnamefont {Dubail}}, \bibinfo {author} {\bibfnamefont {M.}~\bibnamefont {Rigol}},\ and\ \bibinfo {author} {\bibfnamefont {D.~S.}\ \bibnamefont {Weiss}},\ }\bibfield  {title} {\bibinfo {title} {Generalized hydrodynamics in strongly interacting 1d bose gases},\ }\href@noop {} {\bibfield  {journal} {\bibinfo  {journal} {Science}\ }\textbf {\bibinfo {volume} {373}},\ \bibinfo {pages} {1129} (\bibinfo {year} {2021})}\BibitemShut {NoStop}%
\bibitem [{\citenamefont {Wei}\ \emph {et~al.}(2022)\citenamefont {Wei}, \citenamefont {Rubio-Abadal}, \citenamefont {Ye}, \citenamefont {Machado}, \citenamefont {Kemp}, \citenamefont {Srakaew}, \citenamefont {Hollerith}, \citenamefont {Rui}, \citenamefont {Gopalakrishnan}, \citenamefont {Yao} \emph {et~al.}}]{wei2022quantum}%
  \BibitemOpen
  \bibfield  {author} {\bibinfo {author} {\bibfnamefont {D.}~\bibnamefont {Wei}}, \bibinfo {author} {\bibfnamefont {A.}~\bibnamefont {Rubio-Abadal}}, \bibinfo {author} {\bibfnamefont {B.}~\bibnamefont {Ye}}, \bibinfo {author} {\bibfnamefont {F.}~\bibnamefont {Machado}}, \bibinfo {author} {\bibfnamefont {J.}~\bibnamefont {Kemp}}, \bibinfo {author} {\bibfnamefont {K.}~\bibnamefont {Srakaew}}, \bibinfo {author} {\bibfnamefont {S.}~\bibnamefont {Hollerith}}, \bibinfo {author} {\bibfnamefont {J.}~\bibnamefont {Rui}}, \bibinfo {author} {\bibfnamefont {S.}~\bibnamefont {Gopalakrishnan}}, \bibinfo {author} {\bibfnamefont {N.~Y.}\ \bibnamefont {Yao}}, \emph {et~al.},\ }\bibfield  {title} {\bibinfo {title} {Quantum gas microscopy of kardar-parisi-zhang superdiffusion},\ }\href@noop {} {\bibfield  {journal} {\bibinfo  {journal} {Science}\ }\textbf {\bibinfo {volume} {376}},\ \bibinfo {pages} {716} (\bibinfo {year} {2022})}\BibitemShut {NoStop}%
\bibitem [{\citenamefont {Guardado-Sanchez}\ \emph {et~al.}(2020)\citenamefont {Guardado-Sanchez}, \citenamefont {Morningstar}, \citenamefont {Spar}, \citenamefont {Brown}, \citenamefont {Huse},\ and\ \citenamefont {Bakr}}]{guardado2020subdiffusion}%
  \BibitemOpen
  \bibfield  {author} {\bibinfo {author} {\bibfnamefont {E.}~\bibnamefont {Guardado-Sanchez}}, \bibinfo {author} {\bibfnamefont {A.}~\bibnamefont {Morningstar}}, \bibinfo {author} {\bibfnamefont {B.~M.}\ \bibnamefont {Spar}}, \bibinfo {author} {\bibfnamefont {P.~T.}\ \bibnamefont {Brown}}, \bibinfo {author} {\bibfnamefont {D.~A.}\ \bibnamefont {Huse}},\ and\ \bibinfo {author} {\bibfnamefont {W.~S.}\ \bibnamefont {Bakr}},\ }\bibfield  {title} {\bibinfo {title} {Subdiffusion and heat transport in a tilted two-dimensional fermi-hubbard system},\ }\href@noop {} {\bibfield  {journal} {\bibinfo  {journal} {Physical Review X}\ }\textbf {\bibinfo {volume} {10}},\ \bibinfo {pages} {011042} (\bibinfo {year} {2020})}\BibitemShut {NoStop}%
\bibitem [{\citenamefont {Scherg}\ \emph {et~al.}(2021)\citenamefont {Scherg}, \citenamefont {Kohlert}, \citenamefont {Sala}, \citenamefont {Pollmann}, \citenamefont {Hebbe~Madhusudhana}, \citenamefont {Bloch},\ and\ \citenamefont {Aidelsburger}}]{scherg2021observing}%
  \BibitemOpen
  \bibfield  {author} {\bibinfo {author} {\bibfnamefont {S.}~\bibnamefont {Scherg}}, \bibinfo {author} {\bibfnamefont {T.}~\bibnamefont {Kohlert}}, \bibinfo {author} {\bibfnamefont {P.}~\bibnamefont {Sala}}, \bibinfo {author} {\bibfnamefont {F.}~\bibnamefont {Pollmann}}, \bibinfo {author} {\bibfnamefont {B.}~\bibnamefont {Hebbe~Madhusudhana}}, \bibinfo {author} {\bibfnamefont {I.}~\bibnamefont {Bloch}},\ and\ \bibinfo {author} {\bibfnamefont {M.}~\bibnamefont {Aidelsburger}},\ }\bibfield  {title} {\bibinfo {title} {Observing non-ergodicity due to kinetic constraints in tilted fermi-hubbard chains},\ }\href@noop {} {\bibfield  {journal} {\bibinfo  {journal} {Nature Communications}\ }\textbf {\bibinfo {volume} {12}},\ \bibinfo {pages} {4490} (\bibinfo {year} {2021})}\BibitemShut {NoStop}%
\bibitem [{\citenamefont {Caha}\ and\ \citenamefont {Nagaj}(2018)}]{caha2018pair}%
  \BibitemOpen
  \bibfield  {author} {\bibinfo {author} {\bibfnamefont {L.}~\bibnamefont {Caha}}\ and\ \bibinfo {author} {\bibfnamefont {D.}~\bibnamefont {Nagaj}},\ }\bibfield  {title} {\bibinfo {title} {The pair-flip model: a very entangled translationally invariant spin chain},\ }\href@noop {} {\bibfield  {journal} {\bibinfo  {journal} {arXiv preprint arXiv:1805.07168}\ } (\bibinfo {year} {2018})}\BibitemShut {NoStop}%
\bibitem [{\citenamefont {Moudgalya}\ \emph {et~al.}(2018)\citenamefont {Moudgalya}, \citenamefont {Rachel}, \citenamefont {Bernevig},\ and\ \citenamefont {Regnault}}]{moudgalya2018exact}%
  \BibitemOpen
  \bibfield  {author} {\bibinfo {author} {\bibfnamefont {S.}~\bibnamefont {Moudgalya}}, \bibinfo {author} {\bibfnamefont {S.}~\bibnamefont {Rachel}}, \bibinfo {author} {\bibfnamefont {B.~A.}\ \bibnamefont {Bernevig}},\ and\ \bibinfo {author} {\bibfnamefont {N.}~\bibnamefont {Regnault}},\ }\bibfield  {title} {\bibinfo {title} {Exact excited states of nonintegrable models},\ }\href {https://doi.org/10.1103/PhysRevB.98.235155} {\bibfield  {journal} {\bibinfo  {journal} {Phys. Rev. B}\ }\textbf {\bibinfo {volume} {98}},\ \bibinfo {pages} {235155} (\bibinfo {year} {2018})}\BibitemShut {NoStop}%
\bibitem [{\citenamefont {Batchelor}\ and\ \citenamefont {Kuniba}(1991)}]{batchelor1991temperley}%
  \BibitemOpen
  \bibfield  {author} {\bibinfo {author} {\bibfnamefont {M.~T.}\ \bibnamefont {Batchelor}}\ and\ \bibinfo {author} {\bibfnamefont {A.}~\bibnamefont {Kuniba}},\ }\bibfield  {title} {\bibinfo {title} {Temperley-lieb lattice models arising from quantum groups},\ }\href@noop {} {\bibfield  {journal} {\bibinfo  {journal} {Journal of Physics A: Mathematical and General}\ }\textbf {\bibinfo {volume} {24}},\ \bibinfo {pages} {2599} (\bibinfo {year} {1991})}\BibitemShut {NoStop}%
\bibitem [{\citenamefont {Aufgebauer}\ and\ \citenamefont {Kl{\"u}mper}(2010)}]{aufgebauer2010quantum}%
  \BibitemOpen
  \bibfield  {author} {\bibinfo {author} {\bibfnamefont {B.}~\bibnamefont {Aufgebauer}}\ and\ \bibinfo {author} {\bibfnamefont {A.}~\bibnamefont {Kl{\"u}mper}},\ }\bibfield  {title} {\bibinfo {title} {Quantum spin chains of temperley--lieb type: periodic boundary conditions, spectral multiplicities and finite temperature},\ }\href@noop {} {\bibfield  {journal} {\bibinfo  {journal} {Journal of Statistical Mechanics: Theory and Experiment}\ }\textbf {\bibinfo {volume} {2010}},\ \bibinfo {pages} {P05018} (\bibinfo {year} {2010})}\BibitemShut {NoStop}%
\bibitem [{\citenamefont {Moudgalya}\ and\ \citenamefont {Motrunich}(2022{\natexlab{b}})}]{moudgalya2022hilbert}%
  \BibitemOpen
  \bibfield  {author} {\bibinfo {author} {\bibfnamefont {S.}~\bibnamefont {Moudgalya}}\ and\ \bibinfo {author} {\bibfnamefont {O.~I.}\ \bibnamefont {Motrunich}},\ }\bibfield  {title} {\bibinfo {title} {Hilbert space fragmentation and commutant algebras},\ }\href@noop {} {\bibfield  {journal} {\bibinfo  {journal} {Physical Review X}\ }\textbf {\bibinfo {volume} {12}},\ \bibinfo {pages} {011050} (\bibinfo {year} {2022}{\natexlab{b}})}\BibitemShut {NoStop}%
\bibitem [{si()}]{si}%
  \BibitemOpen
  \href@noop {} {}\bibinfo {note} {See the supplemental materials for details.}\BibitemShut {Stop}%
\bibitem [{\citenamefont {Wang}\ \emph {et~al.}(2025)\citenamefont {Wang}, \citenamefont {Balasubramanian}, \citenamefont {Han}, \citenamefont {Lake}, \citenamefont {Chen},\ and\ \citenamefont {Yang}}]{wang2025exponentially}%
  \BibitemOpen
  \bibfield  {author} {\bibinfo {author} {\bibfnamefont {C.}~\bibnamefont {Wang}}, \bibinfo {author} {\bibfnamefont {S.}~\bibnamefont {Balasubramanian}}, \bibinfo {author} {\bibfnamefont {Y.}~\bibnamefont {Han}}, \bibinfo {author} {\bibfnamefont {E.}~\bibnamefont {Lake}}, \bibinfo {author} {\bibfnamefont {X.}~\bibnamefont {Chen}},\ and\ \bibinfo {author} {\bibfnamefont {Z.-C.}\ \bibnamefont {Yang}},\ }\bibfield  {title} {\bibinfo {title} {Exponentially slow thermalization in 1d fragmented dynamics},\ }\href@noop {} {\bibfield  {journal} {\bibinfo  {journal} {arXiv preprint arXiv:2501.13930}\ } (\bibinfo {year} {2025})}\BibitemShut {NoStop}%
\bibitem [{\citenamefont {Singh}\ \emph {et~al.}(2021)\citenamefont {Singh}, \citenamefont {Ware}, \citenamefont {Vasseur},\ and\ \citenamefont {Friedman}}]{Singh}%
  \BibitemOpen
  \bibfield  {author} {\bibinfo {author} {\bibfnamefont {H.}~\bibnamefont {Singh}}, \bibinfo {author} {\bibfnamefont {B.~A.}\ \bibnamefont {Ware}}, \bibinfo {author} {\bibfnamefont {R.}~\bibnamefont {Vasseur}},\ and\ \bibinfo {author} {\bibfnamefont {A.~J.}\ \bibnamefont {Friedman}},\ }\bibfield  {title} {\bibinfo {title} {Subdiffusion and many-body quantum chaos with kinetic constraints},\ }\href {https://doi.org/10.1103/PhysRevLett.127.230602} {\bibfield  {journal} {\bibinfo  {journal} {Phys. Rev. Lett.}\ }\textbf {\bibinfo {volume} {127}},\ \bibinfo {pages} {230602} (\bibinfo {year} {2021})}\BibitemShut {NoStop}%
\bibitem [{\citenamefont {Balasubramanian}\ \emph {et~al.}(2023)\citenamefont {Balasubramanian}, \citenamefont {Gopalakrishnan}, \citenamefont {Khudorozhkov},\ and\ \citenamefont {Lake}}]{group_dynamics}%
  \BibitemOpen
  \bibfield  {author} {\bibinfo {author} {\bibfnamefont {S.}~\bibnamefont {Balasubramanian}}, \bibinfo {author} {\bibfnamefont {S.}~\bibnamefont {Gopalakrishnan}}, \bibinfo {author} {\bibfnamefont {A.}~\bibnamefont {Khudorozhkov}},\ and\ \bibinfo {author} {\bibfnamefont {E.}~\bibnamefont {Lake}},\ }\bibfield  {title} {\bibinfo {title} {Glassy word problems: ultraslow relaxation, hilbert space jamming, and computational complexity},\ }\href@noop {} {\bibfield  {journal} {\bibinfo  {journal} {arXiv preprint arXiv:2312.04562}\ } (\bibinfo {year} {2023})}\BibitemShut {NoStop}%
\bibitem [{\citenamefont {Lyons}\ and\ \citenamefont {Peres}(2017)}]{lyons2017probability}%
  \BibitemOpen
  \bibfield  {author} {\bibinfo {author} {\bibfnamefont {R.}~\bibnamefont {Lyons}}\ and\ \bibinfo {author} {\bibfnamefont {Y.}~\bibnamefont {Peres}},\ }\href@noop {} {\emph {\bibinfo {title} {Probability on trees and networks}}},\ Vol.~\bibinfo {volume} {42}\ (\bibinfo  {publisher} {Cambridge University Press},\ \bibinfo {year} {2017})\BibitemShut {NoStop}%
\bibitem [{\citenamefont {Fishman}\ \emph {et~al.}(2022)\citenamefont {Fishman}, \citenamefont {White},\ and\ \citenamefont {Stoudenmire}}]{fishman2022itensor}%
  \BibitemOpen
  \bibfield  {author} {\bibinfo {author} {\bibfnamefont {M.}~\bibnamefont {Fishman}}, \bibinfo {author} {\bibfnamefont {S.}~\bibnamefont {White}},\ and\ \bibinfo {author} {\bibfnamefont {E.}~\bibnamefont {Stoudenmire}},\ }\bibfield  {title} {\bibinfo {title} {The itensor software library for tensor network calculations},\ }\href@noop {} {\bibfield  {journal} {\bibinfo  {journal} {SciPost Physics Codebases}\ ,\ \bibinfo {pages} {004}} (\bibinfo {year} {2022})}\BibitemShut {NoStop}%
\bibitem [{\citenamefont {Oganesyan}\ and\ \citenamefont {Huse}(2007)}]{Huse2007rstatistics}%
  \BibitemOpen
  \bibfield  {author} {\bibinfo {author} {\bibfnamefont {V.}~\bibnamefont {Oganesyan}}\ and\ \bibinfo {author} {\bibfnamefont {D.~A.}\ \bibnamefont {Huse}},\ }\bibfield  {title} {\bibinfo {title} {Localization of interacting fermions at high temperature},\ }\href {https://doi.org/10.1103/PhysRevB.75.155111} {\bibfield  {journal} {\bibinfo  {journal} {Phys. Rev. B}\ }\textbf {\bibinfo {volume} {75}},\ \bibinfo {pages} {155111} (\bibinfo {year} {2007})}\BibitemShut {NoStop}%
\bibitem [{\citenamefont {Santos}(2004)}]{santos2004integrability}%
  \BibitemOpen
  \bibfield  {author} {\bibinfo {author} {\bibfnamefont {L.}~\bibnamefont {Santos}},\ }\bibfield  {title} {\bibinfo {title} {Integrability of a disordered heisenberg spin-1/2 chain},\ }\href@noop {} {\bibfield  {journal} {\bibinfo  {journal} {Journal of Physics A: Mathematical and General}\ }\textbf {\bibinfo {volume} {37}},\ \bibinfo {pages} {4723} (\bibinfo {year} {2004})}\BibitemShut {NoStop}%
\bibitem [{\citenamefont {Bulchandani}\ \emph {et~al.}(2022)\citenamefont {Bulchandani}, \citenamefont {Huse},\ and\ \citenamefont {Gopalakrishnan}}]{bulchandani2022onset}%
  \BibitemOpen
  \bibfield  {author} {\bibinfo {author} {\bibfnamefont {V.~B.}\ \bibnamefont {Bulchandani}}, \bibinfo {author} {\bibfnamefont {D.~A.}\ \bibnamefont {Huse}},\ and\ \bibinfo {author} {\bibfnamefont {S.}~\bibnamefont {Gopalakrishnan}},\ }\bibfield  {title} {\bibinfo {title} {Onset of many-body quantum chaos due to breaking integrability},\ }\href@noop {} {\bibfield  {journal} {\bibinfo  {journal} {Physical Review B}\ }\textbf {\bibinfo {volume} {105}},\ \bibinfo {pages} {214308} (\bibinfo {year} {2022})}\BibitemShut {NoStop}%
\bibitem [{\citenamefont {Hart}(2023)}]{hart2023exact}%
  \BibitemOpen
  \bibfield  {author} {\bibinfo {author} {\bibfnamefont {O.}~\bibnamefont {Hart}},\ }\bibfield  {title} {\bibinfo {title} {Exact mazur bounds in the pair-flip model and beyond},\ }\href@noop {} {\bibfield  {journal} {\bibinfo  {journal} {arXiv preprint arXiv:2308.00738}\ } (\bibinfo {year} {2023})}\BibitemShut {NoStop}%
\bibitem [{\citenamefont {Woess}(2000)}]{woess2000random}%
  \BibitemOpen
  \bibfield  {author} {\bibinfo {author} {\bibfnamefont {W.}~\bibnamefont {Woess}},\ }\href@noop {} {\emph {\bibinfo {title} {Random walks on infinite graphs and groups}}},\ \bibinfo {number} {138}\ (\bibinfo  {publisher} {Cambridge university press},\ \bibinfo {year} {2000})\BibitemShut {NoStop}%
\bibitem [{\citenamefont {Levin}\ and\ \citenamefont {Peres}(2017)}]{levin2017markov}%
  \BibitemOpen
  \bibfield  {author} {\bibinfo {author} {\bibfnamefont {D.~A.}\ \bibnamefont {Levin}}\ and\ \bibinfo {author} {\bibfnamefont {Y.}~\bibnamefont {Peres}},\ }\href@noop {} {\emph {\bibinfo {title} {Markov chains and mixing times}}},\ Vol.\ \bibinfo {volume} {107}\ (\bibinfo  {publisher} {American Mathematical Soc.},\ \bibinfo {year} {2017})\BibitemShut {NoStop}%
\bibitem [{\citenamefont {Gao}\ \emph {et~al.}(2023)\citenamefont {Gao}, \citenamefont {Zhang},\ and\ \citenamefont {Chen}}]{gao2023information}%
  \BibitemOpen
  \bibfield  {author} {\bibinfo {author} {\bibfnamefont {Q.}~\bibnamefont {Gao}}, \bibinfo {author} {\bibfnamefont {P.}~\bibnamefont {Zhang}},\ and\ \bibinfo {author} {\bibfnamefont {X.}~\bibnamefont {Chen}},\ }\href@noop {} {\bibinfo {title} {Information scrambling in free fermion systems with a sole interaction}} (\bibinfo {year} {2023}),\ \Eprint {https://arxiv.org/abs/2310.07043} {arXiv:2310.07043 [quant-ph]} \BibitemShut {NoStop}%
\end{thebibliography}%

\pagebreak
\newpage 
\bigskip 
\onecolumngrid
\begin{center}
	
	\bigskip 
	\bigskip 
	\textbf{\large Exponentially slow thermalization and the robustness of Hilbert space fragmentation: \\[.25cm] Supplementary information}\\[.3cm]

	Yiqiu Han$^1$, Xiao Chen$^1$, and Ethan Lake$^2$ \\[.1cm]
 {\itshape $^1$Department of Physics, Boston College, Chestnut Hill, MA 02467, USA} \\
  {\itshape $^2$Department of Physics, University of California Berkeley, Berkeley, CA 94720, USA} \\
\end{center}

\setcounter{equation}{0}
\setcounter{figure}{0}
\setcounter{table}{0}
\setcounter{page}{1}
\renewcommand{\theequation}{S\arabic{equation}}
\renewcommand{\thefigure}{S\arabic{figure}}

\ss*{Contents} 
		\begin{itemize}
			\item Section~\ref{app:notation}: preliminary remarks on notation 
			\item Section~\ref{app:numerics}: extended numerical results: entanglement dynamics, localization of eigenstates, and level statistics
			\item Section~\ref{app:sector_dims}: calculation of Krylov sector dimensions
			\item Section~\ref{app:stochastics}: details of the mapping between random unitary circuits and stochastic Markov chains
			\item Section~\ref{app:bounds}: proofs of the theorems pertaining to slow thermalization quoted in the main text
			\item Section~\ref{App:TL}: results on the (non)thermalization of Temperley-Lieb chains 
			
		\end{itemize}
		\section{Notational preface} \label{app:notation}
		
		As discussed in the main text, each computational product state $\k{s}$, $s_i \in \{1,\dots,N\}$ can be associated to a random walk on the balanced $N$-tree $T_N$. Proceeding from left to right, each $s_i$ determines where the walk proceeds, with the walk backtracking when $s_i = s_{i-1}$. Each walk which ends at a point of depth $d$ on the tree can be associated with a length-$d$ string which corresponds to the shortest path on the tree from the origin to that point. We will refer to the shortest path associated a string $s$ as $s$'s {\it irreducible string}, which we write as $\irr(s)$. We will let $v_s$ denote the vertex of $T_N$ at which $s$ ends. Since the walks on $T_N$ are non-lazy, a string of length $|s|=L$ can only reach those nodes at a depth whose parity matches that of $L$. Writing $|v_s|$ for the depth of the node $v_s$ (i.e. $|v_s| = |\irr(s)|$), this means that $[v_s]_2 = [L]_2$, where $[\cdot ]_2$ denotes reduction mod 2.
		
		We will let the $\mck_{v_s}^{(L)}$ denote the set of all length-$L$ strings with the same $\irr(s)$. The states in $\mck^{(L)}_{v_s}$ form a basis of the Krylov sector associated with the node $v_s$. We will write $\mck^{(L)}_d$ when we wish to refer to an arbitrary Krylov sector whose associated vertex is at a depth $d$ on $T_N$. In addition, we will write $N_\mck(L)$ for the total number of Krylov sectors on a system of size $L$. Finally, we will represent the projection operator onto the Krylov sector $\mck_{v_s}^{(L)}$ as
		\be \Pi_{\mck_{v_s}^{(L)}}\equiv \sum_{|s\ran\in \mck_{v_s}^{(L)}} |s\ran\lan s|. \ee

		\section{More numerical results} \label{app:numerics}
		In this section we collect additional numerical results regarding the thermalization of Hamiltonian dynamics, with $N=3$ unless explicitly stated otherwise. All of the Hamiltonians we will consider will act on open 1d chains, and will be of the form 
		\be H = H_0 + \l H_{\rm imp},\ee 
		where $H_0$ acts on the entire system (of length $L$), and $H_{\rm imp}$ acts only on sites $\lc < i \leq \ltot$. $H_{\rm tot}$ will always be normalized such that $||H_{\rm imp}||_\infty = \ltot-\lc$. 

		\subsection{Entanglement and Krylov entropy}
		
		\begin{figure*}
			\centering
			\includegraphics[width=0.32\textwidth]{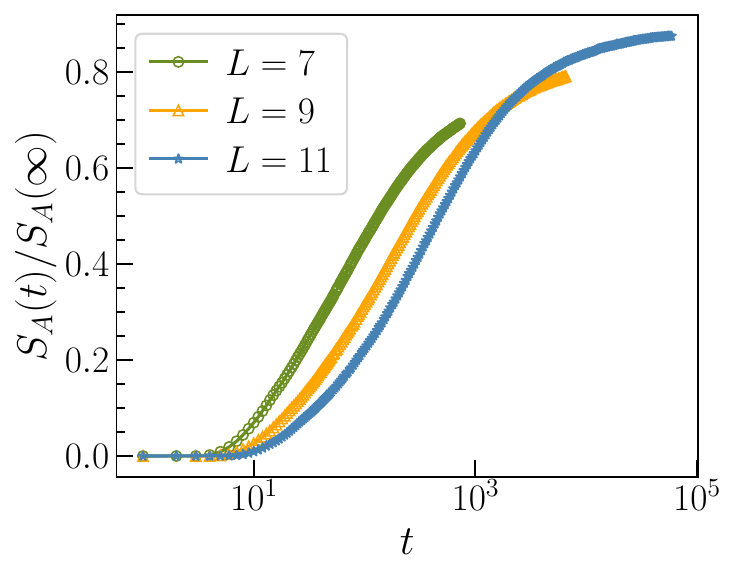} 
			\includegraphics[width=0.32\textwidth]{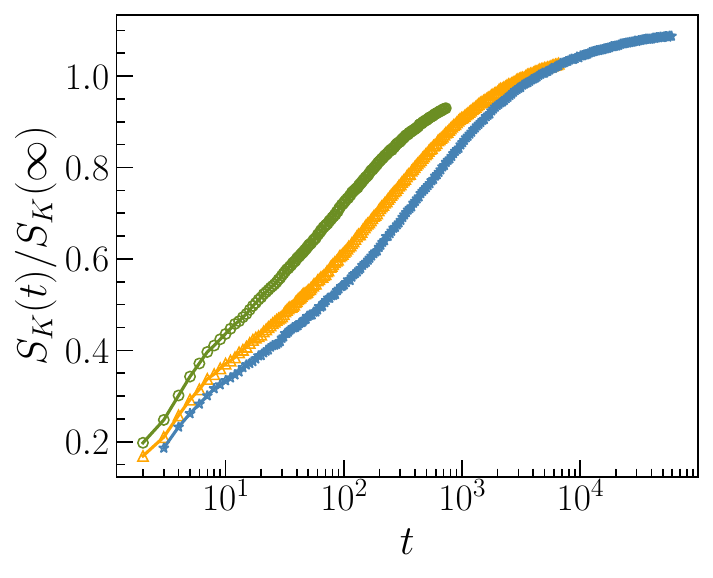}
			\includegraphics[width=0.32\textwidth]{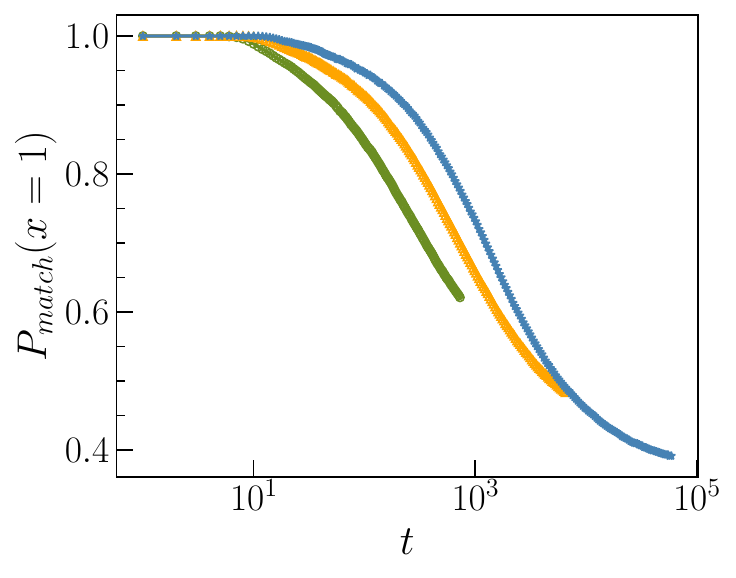}
			\caption{\label{fig:PF_rand} Numerical results of the $N=3$ translation-invariant PF model with random $g^{ab} = (g^{ab})^*$ under a single-site random impurity of strength $\lambda=1$ at $\ltot$, starting from an initial state in $\mch_{\rm froz}$. {\it Left}: The normalized von-Neumann entropy $S_A(t)/S_A(\infty)$, where $A$ is half of the constrained region and $S_A(\infty) = |A|\ln(N)/2-1/2$. {\it Center}: The Krylov entropy $S_K(t)$ normalized by $S_K(\infty)$, where $S_A(\infty)$ is calculated using $p_{v_s} = |\mck_{v_s}| / |\mch|$. {\it Right}: The probability of the spin on the first site to match with its original value.}
		\end{figure*}
		
		In this subsection we show the entanglement dynamics of a more general locally-perturbed PF model, including the half-chain entanglement entropy of the constrained system, and Krylov entropy, a quantity that measures the spreading of the wave function over different Krylov subspaces, i.e.
		\be S_K(t)\equiv -\sum_{v_s} p_{v_s}\ln p_{v_s}, \ee
		where 
		\be p_{v_s}\equiv \lan \psi(t)|\Pi_{\mck_{v_s}}^{(\lc)}\otimes \unit_{\ltot-\lc}|\psi(t)\ran\ee 
		is the probability that the first $\lc$ sites of the state are in sector $\mck_{v_s}$.
		
		We consider a {\it translation-invariant} PF model under a single-site perturbation, described by the Hamiltonian 
		\begin{equation} \label{eq:generalPF}
			H_0=\sum_{a,b=1}^{3}g^{ab}\sum_{i=1}^{L-1} \kb{aa}{bb}_{i,i+1},\qq  H_{\rm imp} = \sum_{a,b=1}^3|a\ran\lan b|_{L},
		\end{equation}
		where $g$ is an arbitrary $N\times N$ Hermitian matrix that we take to be site-independent to avoid possible many-body localization caused by disorder. In Fig.\ref{fig:PF_rand}, we take $\lambda=1$ and average over different realizations of $g^{ab}$ and different initial frozen states. We first consider the half-chain von-Neumann entropy $S_A$ where $A=[1,\lc/2]$ is half of the constrained region. As shown in the left panel, $S_A(t)\sim \ln(t)$ as it approaches its thermal value $S_A(\infty)=\lc\ln(3)/2$. Similarly, in the middle panel, the Krylov entropy increases logarithmically at early times. Here we plot $S_K(t)$ normalized by the value it takes when $p_{v_s} = |\mck_{v_s}|/|\mch|$, i.e. the value it takes when the wavefunction is spread out uniformly across Hilbert space. Since this is not the value of the $p_{v_s}$ which maximizes the entropy, {values greater than 1 are reached at late times}. 
		
		We additionally calculate $P_{\rm match}(x)$, the probability of the spin on site $x$ to be the same as its original value. As shown in the right panel, $P_{\rm match}(x=1)$ starts to deviate from 1 after the influence of the bath has propagated through the whole system, decreasing logarithmically in time to a value that approaches $1/3$ as $L\to\infty$. This further verifies that the time scale of thermalization of the general PF model under a constraint-breaking perturbation is exponentially long in system sizes.

		\ss{Localization of eigenstates}
		
		To determine the extent that eigenstates of $H$ are localized on $T_N$, we can measure the ``expected depth'' $d_\mu$ of each eigenstate $\k\mu$, defined as 
		\be \label{ddef} d_\mu  \equiv \sum_{d=0}^{L/2} \sum_{v_s \, : \, \irr(s) = d} d \lan \mu | \Pi_{\mck_{v_s}} | \mu \ran. \ee 
		
		\begin{figure*}
			\centering \includegraphics[width=.32\tw]{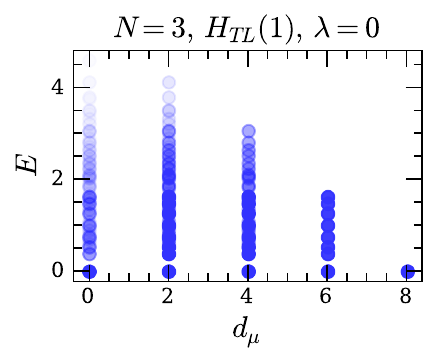} \hfill 
			\includegraphics[width=.32\tw]{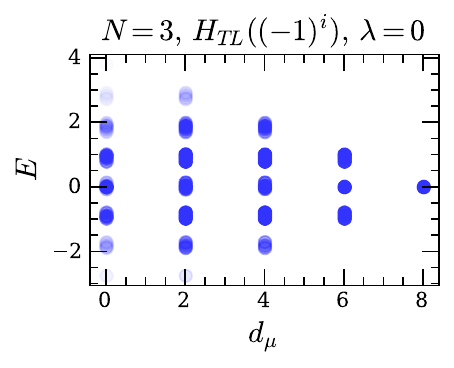} \hfill 
			\includegraphics[width=.32\tw]{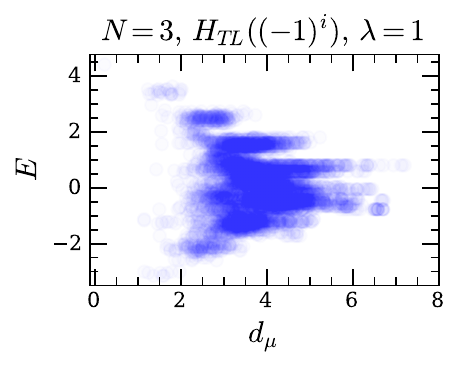} 
			\caption{\label{fig:tl_sorted_spec} Eigenenergies in different variants of $SU(3)$-symmetric models for $L=8$, arranged according to the Krylov distance $d_\mu$ \eqref{ddef} of each eigenstate. {\it Left:} Unperturbed TL model without sublattice staggering ($g_i = 1$). The frozen states at $d=8$ all lie at the bottom of the spectrum. {\it Center:} The same model with with $g_i=(-1)^i$: the frozen states now lie in the middle of the spectrum. {\it Right:} The same model but now with a two-site impurity of strength $\l=1$ at the end of the chain. States with large $d_\mu$ continue to be located roughly in the middle of the spectrum. 
			}  
		\end{figure*}
		\begin{figure*}
			\centering  
			\includegraphics[width=.32\tw]{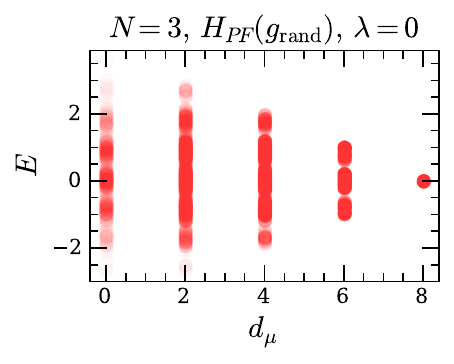} \hfill 
			\includegraphics[width=.32\tw]{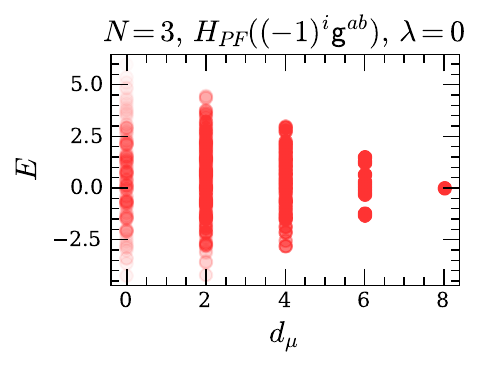} \hfill 
			\includegraphics[width=.32\tw]{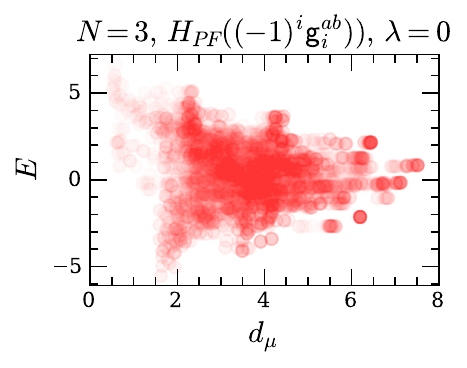} 
			\caption{\label{fig:pf_sorted_spec} Eigenenergies in different $SU(3)$-breaking models for $L=8$, arranged according to the Krylov distance $d_\mu$ \eqref{ddef} of each eigenstate. {\it Left:} A random translation-invariant choice of the pair-flip matrix $g^{ab}$. {\it Center:} The choice $g_i^{ab} = (-1)^i \sfg^{ab}$ adopted in the main text with $\l=0$. The frozen states at $d_\mu=0$ continue to lie roughly in the middle of the spectrum. {\it Right:} The same Hamiltonian but with $\l = 1$, showing a distribution of $d_\mu$ which continues to remain very broad.
			}  
		\end{figure*}
		
		\begin{figure*}
			\centering  
			\includegraphics[width=.32\tw]{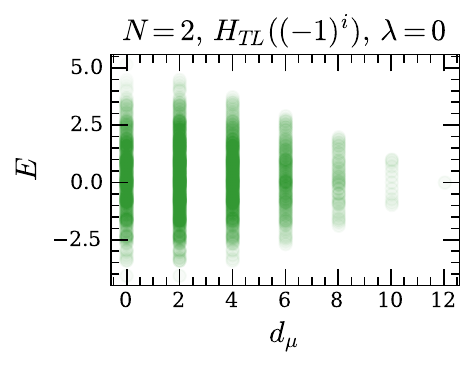} \hfill 
			\includegraphics[width=.32\tw]{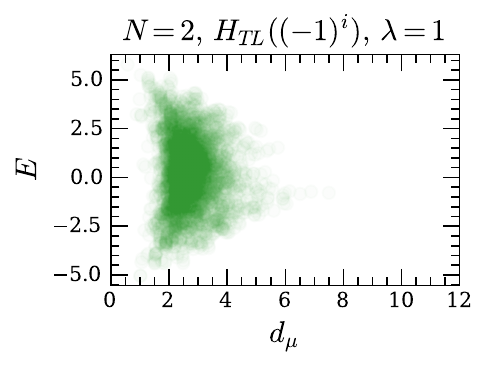} \hfill 	\includegraphics[width=.32\tw]{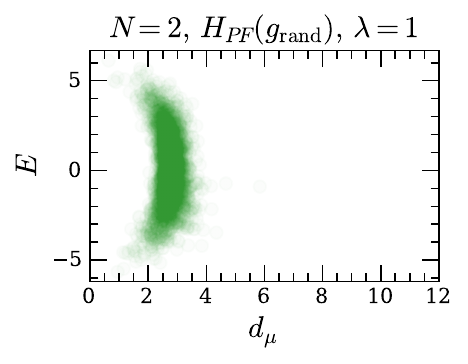} 
			\caption{\label{fig:n2_sorted_spec} Eigenenergies for $N=2$ models. {\it Left:} The $SU(2)$-symmetric model with $g_i = (-1)^i$ and $\l=0$.  {\it Center:} The same model but with $\l=1$, with most eigenstates clustering around a small value of $d_\mu$, as expected from a thermalizing Hamiltonian (on account of the absence of strong fragmentation). {\it Right:} An $SU(2)$-breaking model with random pair-flip matrix $g^{ab}$. The distribution of $d_\mu$s is even more tightly concentrated about the value one would obtain for a thermalizing Hamiltonian.}  
		\end{figure*}
		
		Histograms of $d_\mu$ and eigenstate energy $E$ are shown in Figs.~\ref{fig:tl_sorted_spec},~\ref{fig:pf_sorted_spec}, and~\ref{fig:n2_sorted_spec} for $SU(3)$-symmetric pair-flip models, $SU(3)$-breaking pair-flip models, and the $N=2$ variant of the symmetric model (unitarily equivalent to a perturbed $XXX$ chain), respectively. For the $N=3$ models we use a total system size of $L=8$, while for $N=2$ we set $L=12$. 
		
		In Fig.~\ref{fig:tl_sorted_spec} we study $SU(3)$-symmetric Temperley-Lieb chains of the form 
		\be \label{htl} H_{TL}(g_i) = \sum_i g_i P_{i,i+1} + \l H_{\rm imp}, \qq P_{i,i+1}  \equiv \frac1N \sum_{a,b=1,\dots,N} \kb{a,a}{b,b}_{i,i+1}, \ee 	 	
		where $H_{\rm imp}$ is a random matrix supported only on the last two sites of the chain ($i=7,8$) and normalized so that $||H_{\rm imp}||_\infty = 1$. In the left panel we set $g_i=1, \l=0$ and observe a large dengeneracy of $E=0$ states at the bottom of the spectrum. The model is strongly fragmented since $\l=0$, and so $d_\mu$ is an integer for each $\k\mu$. In the middle panel we take $g_i = (-1)^i$, which puts the frozen states at $E=0$ into the middle of the spectrum. In the right panel we break the fragmentation by taking $\l=1$, which is the value of $\l$ for which $H_{\rm imp}$ has the strongest effect on the spectrum. Now the $d_\mu$ are no longer integers, but nevertheless a fairly broad distribution of $d_\mu$s is observed. A generic model with no localization on the Krylov space tree would have $d_\mu \approx v_NL = L/3$ for all eigenstates. The broad distribution of $d_\mu$ observed here suggests that some degree of localization persists, although future work will be needed to understand to what degree this is due to finite size effects. 
		
		We now consider pair-flip models which lack $SU(3)$ symmetry at $\l=0$. The most general Hamiltonian we will consider is of the form 
		\be \label{hpf} H_{PF}(g_i^{ab}) = \sum_i \sum_{a,b=1}^N g^{ab}_i \kb{aa}{bb}_{i,i+1} + \l H_{\rm imp}.\ee 
		In the left panel of Fig.~\ref{fig:pf_sorted_spec} we show the spectrum at $\l=0$ for a random translation-invariant choice of $g^{ab}_i$. The frozen states at $d_\mu=8$ are observed to lie roughly in the middle of the spectrum. In the center panel we show the choice of $g^{ab}_i$ adopted in the main text, viz. 
		\be g^{ab}_i = (-1)^i\sfg^{ab} \equiv (-1)^i \(\frac13 + \kappa\d^{a,b}\)\ee 
		with $\kappa$ fixed at $2/3$, as in the main text. 
		
		Fig.~\ref{fig:n2_sorted_spec} shows the spectrum of analogous models with $N=2$, which are not strongly fragmented. In the left panel we show \eqref{htl} with $g_i = (-1)^i$ and $\l=0$, which is unitarily equivalent to a staggered XXX chain. In the center panel we set $\l=1$, which brings nearly all of the $\k\mu$ down to $d_\mu/L \approx 1/4$, which is close to what we would expect for a generic Hamiltonian. Finally, in the right panel we let $H$ be a random pair-flip Hamiltonian, observing an even more tightly clustered distribution of $d_\mu$.

		\ss{$r$ statistics}
		
		\begin{figure*}
			\centering
			\includegraphics[width=0.35\textwidth]{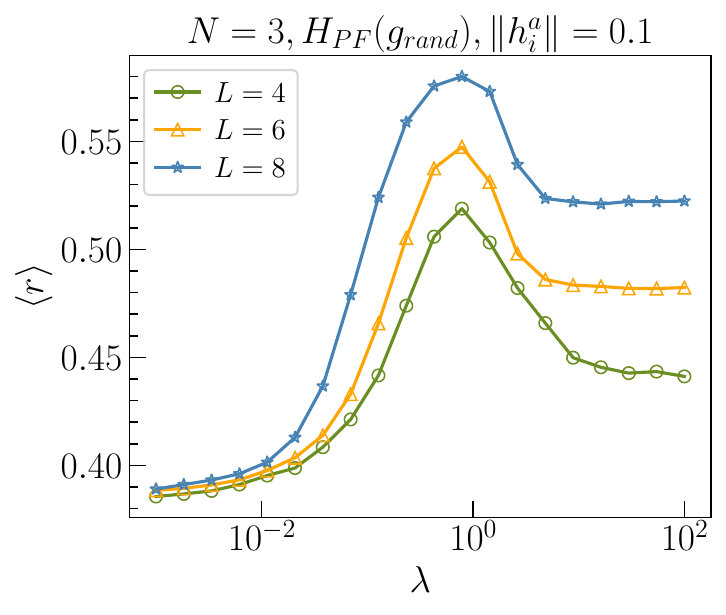} 
			\includegraphics[width=0.38\textwidth]{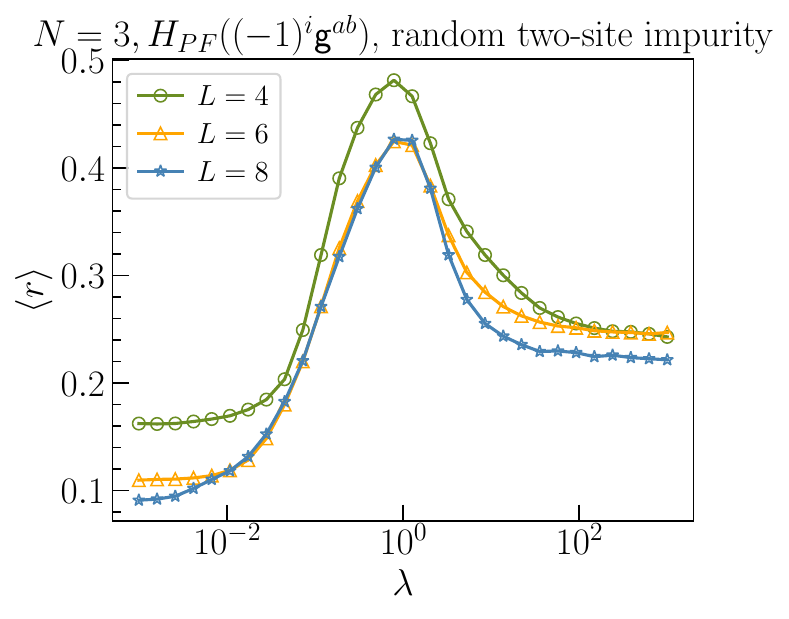}
			\caption{\label{fig:r_stat} The $\lan r\ran$ statistics of $N=3$ models vs. perturbation strength $\lambda$ for $\ltot=4, 6, 8$. {\it Left:} The translation-invariant PF model with random pair-flip matrix $g^{ab}$, a longitudinal field with strength $\|h^a_i\|=0.1$, and a single-site perturbation. When $\lambda=\mco(1)$, $\lan r\ran\to 0.6$ as $L\to\infty$. {\it Right:} The model with $g_i^{ab}=(-1)^i\mathsf{g}^{ab}$ discussed in the main text, under a random two-site impurity. }
		\end{figure*}
		
		The models with $N>2$ that we study all exhibit strong HSF, with $|\mck_{\rm max}|$ being exponentially smaller than $|\mch|$. This implies that if one examines the spectrum of $H$, consecutive eigenstates will almost certainly belong to distinct Krylov sectors, and hence the spectrum of $H$ will exhibit no nearest-neighbor level repulsion in the absence of constraint-breaking terms. How strongly the eigenstates in different sectors hybridize as a constraint-breaking term is applied provides a characterization of the severity by which thermalization is impeded, since such hybridization is a prerequisite for getting initial product states in $\hfroz$ to thermalize. In this subsection we will take some first steps towards studying this question numerically by computing the {\it r-statistic} \cite{Huse2007rstatistics}
		\be \lan r\ran\equiv \lan r_n\ran, \qq r_n\equiv \frac{\min(\delta_n, \delta_{n+1})}{\max(\delta_n,\delta_{n+1})}, \ee
		where $\delta_n\equiv E_{n+1}-E_n$ is the gap between adjacent non-degenerate energy levels $E_{n+1}>E_n$. Computing $\lan r\ran$ also allows us to make contact with research examining how the addition of local and/or weak generic perturbations to integrable Hamiltonians leads to the onset of chaos (see  \cite{santos2004integrability,bulchandani2022onset} for two almost-randomly chosen references on this broad topic). 
		
		As mentioned above, when $\l=0$ --- viz. when the constraint-breaking term is turned off --- we expect Poisson statistics, with $\lan r\ran\approx 0.38$. On the other hand, if the constraint-breaking term strongly hybridizes the states in different sectors, we expect $\lan r\ran$ to be given by the Gaussian unitary ensemble (GUE) value of $\lan r\ran\approx 0.6$. 
		In Fig.\ref{fig:r_stat}, we study the $\lan r\ran$ statistics of the $N=3$ locally-perturbed PF models as a function of the perturbation strength $\lambda$. In the left panel, we consider the translation-invariant random PF model as in Eq.\ref{eq:generalPF}, with an additional longitudinal field $H_f=\sum_i h^a_i|a\ran\lan a|$. As $\lambda\to 0$, $\lan r\ran\approx 0.38$ indicates that the unperturbed PF model possesses a spectrum with Poisson distribution, as expected. In the finite-size numerics, $\lan r\ran$ continues to grow as $\lambda$ increases and peaks at $\lambda\sim1$, approaching 0.6 as $L$ increases, followed by a slight decrease to a plateau as $\lambda\to\infty$. In the right panel, we consider the staggered PF model adopted in the main text, with the perturbation being replaced by a two-site random impurity. This model has similar but smaller finite-size $\lan r \ran$ statistics which does not increase as $L$ increases. This is consistent with the fact that the perturbed staggered PF model thermalizes slower than the perturbed random PF model.
		We conjecture that in the thermodynamic limit, $\lan r \ran$ approaches the GUE value for all $\lambda = \O(1)$. 
		If true, this indicates that the different Krylov sectors become well-hybridized in the thermodynamic limit. This of course does not preclude slow thermalization arising due to the Hilbert space bottleneck mechanism discussed in the main text, although it may mean that the thermalization times of typical Hamiltonian dynamics and the constrained RU dynamics studied in the main text are not parametrically different. 
		
		The plateau for large $\lambda$ can be understood as follows. When $\lambda\to\infty$, $\lambda H_{\rm imp}$ decouples the perturbed site from the system, splitting the spectrum of $H$ into sectors labeled by the eigenstate $|\phi_m\rangle$ of $H_{\rm imp}$, and the total eigenstate of $H$ becomes $|\psi^m_{\ltot}\ran\approx |\psi_\lc\ran \otimes |\phi_m\ran$. Taking the impurity to act on only a single site for simplicity, so that $\ltot=\lc+1$, the effective Hamiltonian in the sector labeled by $m$ is
		\begin{equation}
			\begin{aligned}
				H_{\rm eff}^m &= (\unit_{\lc}\otimes|\phi_m\ran\lan\phi_m|) H (\unit_{\lc}\otimes|\phi_m\ran\lan\phi_m|) \\
				& = \sum_{a,b=1}^N\left(\sum_{i=1}^{\lc-1} g_i^{ab}|aa\ran\lan bb|_{i,i+1}+g_{\lc}^{ab}\lan\phi_m|a\ran\lan b|\phi_m\ran |a\ran\lan b|_{\lc}\right)\otimes |\phi_m\ran\lan\phi_m|\\
				& \equiv \Tilde{H}_{\rm eff}^m \otimes |\phi_m\ran\lan\phi_m|,
			\end{aligned}
		\end{equation}
		with some constant terms ignored. $\Tilde{H}_{\rm eff}^m$ is essentially the same PF model as before, but now defined on a length $\ltot-1$ chain with an impurity Hamiltonian acting on the boundary whose matrix elements are 
		\be [H^m_{\rm imp, eff}]_{ab}=g_{\lc}^{ab}\lan\phi_m|a\ran\lan b|\phi_m\ran\sim O(1).\ee  
		Therefore, the $\lan r \ran$ statistics on a system of size $L$ at $\lambda\to\infty$ should be the same as that of  a system of size $L-1$ at $\lambda=O(1)$. 
		
		From the numerics, there does not naively seem to be  a well-defined transition in the $\lan r \ran$ statistics from Poisson to GUE, as one finds in certain types of perturbed integrable models \cite{bulchandani2022onset}, although a more detailed numerical study will need to be carried out to properly address this question.

		\section{Krylov sector dimensions} \label{app:sector_dims}

		\ss{$N=2$}
		
		We first dispatch with the easy case of $N=2$, for which the dynamics is not fragmented. The tree $T_2$ is simply a line, and the different Krylov sectors can be fully distinguished by the charge $Q_1$ defined in \eqref{charges}, the value of which gives the distance of the Krylov sector along the line. The number of Kyrlov sectors is simply 
		\be N_\mck(L) = L+1.\ee 
		The dimension of a sector whose irreducible string has length $d$ is determined by counting the number of length-$L$ non-lazy random walks on the line which end at a distance of $d>0$ from the origin. This number is simply 
		\bea \dim[\mck_{d}(L)] &= \d_{[d]_2,[L]_2}{L \choose \frac{L+d}2} \\
		& \approx \sqrt{\frac{2L}{\pi(L^2-d^2)}} \exp(L H(p_d)), \eea 
		where $H(x) = -x\ln x - (1-x)\ln(1-x)$ is the binary Shannon entropy and $p_d \equiv (1+d/L)/2$. 
		
		\ss{$N>2$} 
		
		When $N>2$ the dynamics is strongly fragmented, and determining the sizes of the different Krylov sectors is less trivial. 
		We start with the total number of Kyrlov sectors $N_\mck(L)$. 
		From thinking about the tree structure of $T_N$, it is clear that this number is 
		\bea N_\mck(L) = \begin{dcases} 1+ N \sum_{l=1}^{L/2} (N-1)^{2l-1}  & \text{($L$ even)} \\ 
			N \sum_{l=1}^{(L+1)/2} (N-1)^{2l-2} & \text{($L$ odd)} \end{dcases} \quad = \frac{(N-1)^{L+1} - 1}{N-2} = \ct((N-1)^L) .\eea 
		Note that while this number is exponentially large in $L$, it is still exponentially smaller than the Hilbert space dimension $\dim \mch = N^L$. 
		
		We now begin in our determination of the Krylov sector sizes. For simplicity of notation we will restrict our attention to the case when $L$ is even. We start with the size of the largest Krylov sector $\mck_0^{(L)}$, which is identified with the vertex at the center of the tree:
		\begin{figure}
			\centering 
			\includegraphics[width=\textwidth]{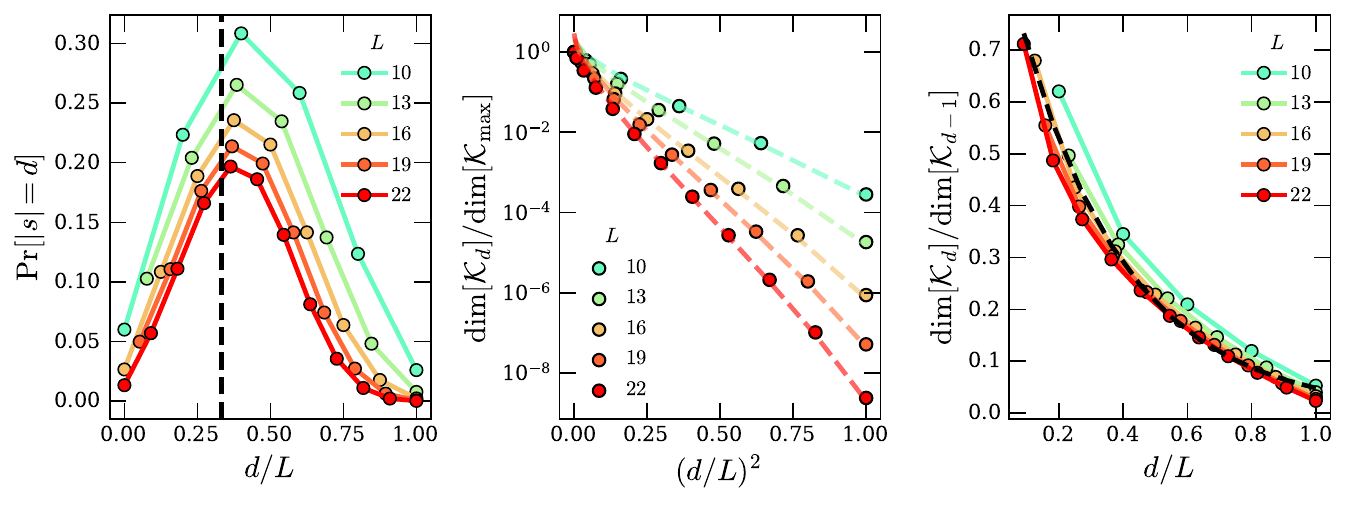}
			\caption{\label{fig:random_words} {\it Left:} The probability that a randomly chosen product state will lie in a Krylov sector with irreducible string of length $d$. The dashed black line lies at $d/L = (1-2/N)$ with $N=3$, the most probable size in the $L\ra\infty$ limit. {\it Center:} the sizes of Krylov sectors $\mck_d$ with length-$d$ irreducible string, compared with the size of the largest Krylov sector (dots). Dashed lines are plotted using the approximate expression in \eqref{approxdim}. {\it Right:} the relative size of succesive Krylov sectors arranged by distance $d$, exhibiting a decay scaling approximately exponentially with $d/L$. The dashed black line is drawn according to \eqref{approxdim}.}
		\end{figure}
		\begin{proposition} \label{prop:kmaxdim} 
			For even $L$, the size of the largest Krylov sector $\mck_0^{(L)}$ is 
			\be \label{kmaxexact} |\mck_0^{(L)}| =  N^{L} \(1 + \frac{1}2 \sum_{n=1}^{L/2} N^{-2n}{1/2\choose n} (-1)^n \g^{2n}\).\ee 

			For both even and odd $L$, in the large $L$ limit $|\mck_0^{(L)}|$ scales as 
			\be \label{kmaxdim} |\mck_0^{(L)}| \sim L^{-3/2}  (N\r)^L,\ee 
			where 
			\be \r \equiv \frac{2 \sqrt{N-1}}{N}\ee 
			is the spectral radius of $T_N$ \cite{lyons2017probability}. 
		\end{proposition}
		\begin{proof}
			
			As in Ref.~\cite{hart2023exact}, our proof will use generating functions to obtain an exact expression for $|\mck_0^{(L)}|$. Since we are assuming $L$ is even, the size of $\mck_0^{(L)}$ can be determined by counting the number of non-lazy simple length-$L$ random walks on $T_N$ which begin and end at the origin. 
			
			Let $R(x)$ be the generating function for non-lazy simple random walks on $T_N$. $R(x)$ is easily seen to obey the recursion relation 
			\be R(x) = 1 +N x^2R(x) B(x),\ee 
			where $B(x)$ is the generating function for returning walks on the {\it rooted} $N$-regular tree $T_{N,r}$ which begin and end at the root vertex (which has degree $N-1$). We thus need a recursion relation for $B(x)$, which is readily obtained as 
			\be B(x) = 1 + x^2(N-1)B(x)^2,\ee 
			which when solved yields \cite{hart2023exact}
			\be B(x) = 2\frac{1-\sqrt{1-(x\g)^2}}{(x\g)^2} \qq \g \equiv 2\sqrt{N-1}.\ee	
			We can now use this expression to get $R(x)$, which we may write after some algebra as\footnote{An aside: one should not be alarmed that $R(1)$ is imaginary. If one wants the expected number of times an infinitely long walk returns to the origin, one needs to write down generating functions for probabilities, rather than for number of paths. Since each individual move has an equal probability of $1/N$, this amounts to sending $x \mt x / N$, which gives 
				\be R(x/N) = \frac{2(N-1)}{N-2+N\sqrt{1-(x\r )^2}},\ee 
				and so sending $x\ra 1$ then gives 
				\be \label{expected_returns}\lan \text{number of returns} \ran = \frac{N-1}{N-2}.\ee 
				This appropriately diverges when $N=2$ but is finite for all $N>2$, and correctly approaches 1 as $N\ra\infty$.} 
			\be\label{nicer} R(x) = \frac{2 + N(\sqrt{1-(\g x)^2 } - 1)}{2(1-(Nx)^2)}.\ee 
			
			We now want to determine the long-walk asymptotics, which requires that we perform the series expansion 
			\be R(x) \equiv \sum_k x^{2k} |\mck_0^{(2k)}|,\ee which takes the form of a convolution between a geometric series and the series coming from the expansion of the square root in the denominator. From \eqref{nicer}, a Taylor expansion gives 
			\bea R(x) & =\frac12\sum_{m=0}^\infty (Nx)^{2m} \( 2-N+N\sum_{n=0}^\infty {1/2\choose k}(-1)^k (\g x)^{2k} \).\eea 
			The $L$th coefficient is then
			\be |\mck_{0}^{(L)}| = N^{L} \(1 + \frac{1}2 \sum_{n=1}^{L/2} N^{-2n+1}{1/2\choose n} (-1)^n \g^{2n}\),\ee
			giving the exact result \eqref{kmaxexact}. 
			
			We now obtain the asymptotic result \eqref{kmaxdim}. For this we rewrite the fractional binomial coefficient above as 
			\be {1/2 \choose n} (-1)^n = - \frac{(2n-3)!!}{2^n n!}.\ee  
			Since the fraction of length-$2k$ walks which return to the origin vanishes as $l\ra\infty$, we have $(|\mck_0^{(L)}|/N^{L}) |_{L\ra \infty} = 0$, which using the above implies 
			\be \sum_{n=1}^\infty \frac{N^{1-2n} \g^{2n} (2n-3)!!}{2^{n+1} n!} = 1.\ee 
			Thus for large $L$ we may write
			\bea \frac{|\mck_0^{(L)}|}{N^{L}} & = \sum_{n=k+1}^\infty \frac{ N^{1-2n} \g^{2n} (2n-3)!!}{2^{n+1} n!} \\ 
			& \approx \sum_{n=L/2+1}^\infty \sqrt{\frac{e^3N^2(1-3/2n)}{16\pi}}(n-3/2)^{-3/2} \exp\(n[-2\ln(N/\gamma) + \ln(1-3/2n) \) \\ 
			& \approx \int_{L/2}^\infty dn \, \sqrt{\frac{N^2 e^3}{16\pi}} n^{-3/2} \exp(-2 n \ln (N/\g)), \eea
			where we used $k!! \approx \sqrt{2k} (k/e)^{k/2}$ at large $k$. 

			Doing the integral (whose exact expression is written in terms of $\G(1/2,L \ln \g)$) then gives
			\be \label{asymptotic_kmax} |\mck_0^{(L)}| \sim L^{-3/2} (N\r)^{L},\ee 
			where $\r = \g / N$ is the spectral radius introduced above. That the base of the exponential is $N\r$ also follows more directly from the fact that $\r$ characterizes the return probability $p_{\rm ret}(L)$ as $\r \equiv \lim_{L\ra \infty} (p_{\rm ret}(L))^{1/L}$ \cite{woess2000random}. 
		\end{proof}
		
		\begin{figure}
			\centering 
			\includegraphics{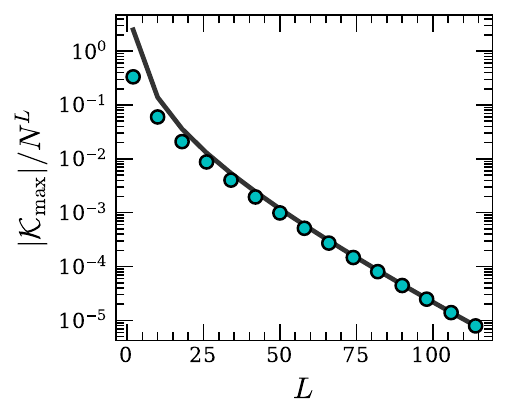}
			\caption{\label{fig:kmaxscaling} Scaling of the size of the largest Krylov sector $\mck_0$ as a function of $L$, plotted for $N=3$. Cyan circles are the exact result \eqref{kmaxexact}, and the gray line is the asymptotic expression \eqref{kmaxdim}.  } 
		\end{figure}
		
		Since $\r < 1$ for all $N>2$, the largest Krylov sector occupies an exponentially small subset of Hilbert space, and the PF model is strongly fragmented for all $N>2$. 
		
		We will also have occasion to know the number of length-$L$ random walks that end a distance $d$ from the origin. To this end, we can generalize the above result about $|\mck_0^{(L)}|$ to 
		\begin{proposition}
			The size of a Krylov sector $\mck_d^{(L)}$ whose irreducible strings have length $d$ is  
			\be \label{kdexact} |\mck_d^{(L)}| = 2^d \sum_{n=0}^{(L+d)/2} \sum_{k=0}^d |\mck_0^{(L+d-2n)}|(-1)^{k+n} {d\choose k} {k/2\choose n} \g^{2(n-d)}.\ee  
		\end{proposition} 
		
		\begin{proof}
			Let $G(x;d)$ be the generating function for walks which travel to a specific vertex at distance $d$. In terms of the generating functions introduced above, this is seen to be 
			\be G(x;d) = x^d R(x) B^d(x).\ee 
			Performing an expansion of $B^d(x)$, we have 
			\bea G(x;d) & = 2^d \sum_{l,n=0}^\infty\sum_{k=0}^d (-1)^{k+n} {d \choose k} {k/2 \choose n} \g^{2(n-d)} |\mck_0^{(2l)}| x^{2l+2n-d},   \eea 
			and so writing $G(x;d) = \sum_{m=0}^\infty x^{d+2m} |\mck_d^{(d+2m)}|$, we have 
			\be \label{dmckex} |\mck_d^{(d+2m)}| = 2^d \sum_{n=0}^{d+m} \sum_{k=0}^d |\mck_0^{(2(d+m-n))}|(-1)^{k+n} {d\choose k} {k/2\choose n} \g^{2(n-d)}.\ee
			For a length $L$ random walk we will have $m=(L-d)/2$; inserting this above gives the exact expression \eqref{kdexact}. 
		\end{proof} 
		
		\begin{figure}
			\centering 
			\includegraphics{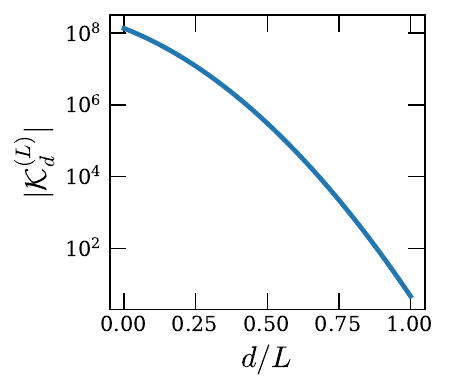}
			\caption{\label{fig:sector_sizes} The asymptotic expression \eqref{approxdim} for the size of $\mck_d^{(L)}$, shown here with $L=20$. }
		\end{figure}
		
		\sss*{Intuition: biased random walks}
		
		We have already obtained an asymptotic expression for $|\mck_0^{(L)}|$ valid at large $L$. When $d/L$ is not very small, we can obtain a complementary asymptotic expression for $|\mck_d^{(L)}|$. This can be done from a straightforward (if unilluminating) expansion of the binomial coefficients in \eqref{dmckex}. More physically, we can argue by realizing that $N(N-1)^{d-1} |\mck_d^{(L)}|$ is equal to the probability for a length-$L$ biased random walk on $\nn$ to end a distance $d$ from the origin; here the radial direction of $T_N$ is identified with $\nn$ and the factor of $N(N-1)^{d-1}$ is equal to the number of sectors at depth $d$. The bias of this walk is probability of moving outward $(1-1/N)$ minus the probability of moving inward $(1/N)$, and hence the walker has velocity 
		\be v_N \equiv 1-2/N.\ee  
		Therefore 
		\bea \label{approxdim} |\mck_d^{(L)}| & \approx  \frac{(1-1/N)^{\frac{L+d}2} N^{- \frac{L-d}2} {L \choose \frac{L+d}2 }}{N(N-1)^{d-1}} N^L \\ 
		& \approx \frac{2(N-1)}{N\sqrt{2\pi L}} N^L \exp\(-\frac{(d-v_NL)^2}{2L} - d \ln (N-1) \),\eea 
		where in the second line we have used that the probability distribution of a random walker on the real line with velocity $v_N$ is $p(x,t) = \frac1{\sqrt{2\pi t}} e^{-(x-v_Nt)^2/2t}$ (with $t=L$ and $x=d$ in the above), and the factor of $2$ comes from the fact that $d$ must have the same parity as $L$ (with the $-d \ln (N-1)$ ensuring that $|\mck_d^{(L)}|$ is always monotonically decreasing with $d$). This function is shown in Fig.~\ref{fig:sector_sizes}. 
		
		%
		%
		%
		%
		%
		%
		%
		%
		
		\begin{comment} 
			\begin{proposition}
				Let $s,s'$ be two random length-$L$ strings. Then the distance between $s$ and $s'$ goes as $L^2$ with probability exponentially close to 1 in the thermodynamic limit:  
				\be 1 - {\rm Pr}[d(s,s') = \ct(L^2)] = O(e^{-L}).\ee 
			\end{proposition}
			
			\begin{proof}
				An $s$ drawn uniformly from $\S^L$ defines a simple non-lazy length-$L$ random walk on the $N$-regular tree with endpoint $\irr(s)$. At each step of this walk, the walk proceeds to greater depths with probability $1-1/N$, and returns to a smaller depth with probability $1/N$. The walk thus moves outwards with velocity $1-2/N$, and so the expected depth of the walk, which we denote as $|\irr(s)|$, is 
				\be \EE_s [ |\irr(s)| ] = L (1-2/N). \ee 
				Moreover, fluctuations about the expected value are small. 
				
			\end{proof}
		\end{comment}

		\section{From RU to stochastic dynamics} \label{app:stochastics}
		
		In this section we derive the generator of the stochastic dynamics that arises from studying circuit-averaged evolution in the model of open-system RU dynamics introduced in the main text. A single timestep of the dynamics corresponds to evolution under the channel 
		\be \label{mccdef} \mcc(\r) = \mcu^\dag \(\Tr_L[\r] \tp \frac{\unit}N\) \mcu,\ee
		where $\Tr_L$ denotes tracing out the last site of the system, the $\unit$ acts on the $L$th site, and $\mcu$ is a random constraint-preserving unitary defined via a depth-$2$ brickwork circuit: 
		\be \mcu = \( \bigotimes_{i=1}^{L/2-1}  U_{2i,2i+1}\) \(\bigotimes_{i=1}^{L/2}  U_{2i-1,2i}\) ,\ee 
		where each $U_{i,i+1}$ are independent unitaries which preserve the constraint. We will view the $\Tr_L[\r]\tp \unit$ part of \eqref{mccdef} as arising from dynamics in which a ``bath'' system at sites $i>L$, which is coupled to the length-$L$ system through a generic interaction at site $L$, is acted on by generic Haar-random unitary dynamics and then traced out, resulting in completely depolarizing noise being applied to the $L$th site. 
		In the following we will first consider the most general case in which the $U_{i,i+1}$ are constrained only by their preservation of walk endpoints in the computational basis; we refer to this as the case of ``pair-flip'' constraints. In a subsequent subsection we will consider a more restrictive case where the $U_{i,i+1}$ preserve walk endpoints in {\it all} single-site bases; this corresponds to a RU realization of the constraints present in the Temperley-Lieb model.

		\ss{Pair-flip} 
		
		For general pair-flip dynamics, the elementary gates $U_{i,i+1}$ take the form
		\be U_{i,i+1} = \sum_{a,b} U^{PF}_{ab} \kb{aa}{bb} + \sum_{a\neq b} e^{i\phi_{ab}} \kb{ab}{ab} = U^{PF} \oplus \bigoplus_{a\neq b} e^{i\phi_{ab}},\ee 
		where the matrix $U^{PF}$ is drawn from the Haar ensemble on $U(N)$, and the second term in the direct sum acts as a diagonal matrix of random phases on the subspace of states frozen under the PF dynamics, viz. those of the form $\k{a,b}, a\neq b$. 
		
		We will be interested in understanding how an operator $\mco$ evolves under the circuit-averaged dynamics, following \cite{Singh} and the slightly modified treatment given in \cite{group_dynamics}. We let 
		\be \ob\mco(t) \equiv \oEE_{\mcc_t} [ \mcc_t(\mco)]\ee 
		denote the circuit-averaged evolution of $\mco$ over time $t$, where the $\oEE_{\mcc_t}$ denotes averaging over the unitaries constituting $\mcu$. 
		To help with notation, we will divide each unit time interval into three steps of length $t=1/3$: at $t\in \nn$ the depolarizing noise is applied, at $t\in \nn+1/3$ the first layer of $\mcu$ is applied, and at $t\in \nn+2/3$ the second layer is applied. Thus 
		\be \label{t13} \ob\mco(t+1/3) = \Tr_L[\ob\mco(t)] \tp \unit.\ee 
		Decomposing $\ob\mco$ as 
		\be \ob\mco(t) = \bot_{i=1}^{L/2-1} \ob\mco_{2i,2i+1}(t)\ee 
		without loss of generality, performing the Haar average gives 
		\bea \label{t23} \ob\mco_{2i,2i+1}(t+2/3) & = \frac{\Tr [\ob\mco_{2i,2i+1}(t+1/3)  \Pi^{PF}]}{N} \Pi^{PF}  + \sum_{a\neq b} \Tr[\ob\mco_{2i,2i+1}(t+1/3) \Pi^{ab}] \Pi^{ab}, \eea 
		where we have defined the projectors 
		\be \Pi^{PF} \equiv \sum_{a} \proj{a,a}, \qq \Pi^{ab} \equiv \proj{a,b}.\ee 
		The action of the second layer of the brickwork is obtained similarly. 
		
		Note that since $\Pi^{PF}$, $\Pi^{ab}$ are diagonal in the computational basis, so too is $\ob \mco$ after any nontrivial amount of time evolution. We may thus focus on diagonal operators, i.e. states, without loss of generality. From \eqref{t13} and \eqref{t23}, the diagonal operator $\proj\psi$ evolves according to 
		\be \k{\ob\psi(t+1)} = \mcm_o\mcm_e\mcm_L \k{\ob\psi(t)},\ee 
		where 
		\bea \label{mcmdef} \mcm_L & = \unit_{L-1} \tp \frac1N \sum_{a,b} \kb ab \\ 
		\mcm_e & = \bot_{i=1}^{L/2-1} M^{PF}_{2i,2i+1} \\ 
		\mcm_o & = \bot_{i=1}^{L/2} M^{PF}_{2i-1,2i}, \eea 
		where we have defined the 2-site stochastic matrix 
		\be M^{PF}  \equiv  \frac1N\sum_{a,b} \kb{a,a}{b,b} + \sum_{a\neq b} \proj{a,b}.\ee 
		Note that $\mcm_o\mcm_e\mcm_L$ is doubly stochastic and irreducible, and hence its unique steady state is the uniform distribution on $\mch$ (irreducibility would of course fail if $\mcm_L$ were absent). This guarantees that RU dynamics will always thermalize to the uniform distribution at long enough times. 
		
		\ss{Temperley-Lieb} 
		
		For Temperley-Lieb dynamics, the elementary unitary gates are constrained to preserve the pair-flip constraint in {\it any} onsite basis, which is done by enriching the previously studied pair-flip dynamics with $SU(N)$ symmetry. The elementary unitary gates for this model take the form  
		\be U_{i,i+1} = e^{i\phi_{i,i+1}} \Pi^{TL} + (\unit-\Pi^{TL}) ,\ee 
		where 
		\be \Pi^{TL} \equiv \frac1N \sum_{a,b} \kb{a,a}{b,b},\ee and where each $\phi_{i,i+1}$ is sampled randomly from $[0,\twp)$. 
		
		The first part of each timestep, whereby the spin at site $L$ is depolarized, is of course unchanged from the more general pair-flip case. After averaging over the $\phi_{i,i+1}$, one sees that the first layer of the brickwork maps operators as 
		\bea \label{tl_op_evol} \ob \mco_{2i,2i+1}(t+2/3) & =\Pi^{TL} \ob\mco_{2i,2i+1}(t+1/3) \Pi^{TL}+ (\unit - \Pi^{TL}) \ob\mco_{2i,2i+1}(t+1/3) (\unit - \Pi^{TL}), \eea  
		with the second layer of the brickwork acting analogously. 
		
		Because of the $U(N)$ invariance of TL dynamics, an operator which is diagonal in any single site product state basis will remain diagonal in that basis. For concreteness we will continue to use the computational basis, although any other basis is equally fine. Diagonal operators, or equivalently states, evolve as $\k{\ob\psi(t+1)} = \mcm\k{\ob\psi(t)}$, where 
		the Markov generator $\mcm$ has the same form as \eqref{mcmdef}, except with $M^{PF}$ replaced by the matrix $M^{TL}$, where
		\be M^{TL} = \(1 - \frac{2(N-1)}{N^2}\) \sum_{a} \proj{a,a} + \frac2{N^2}\sum_{a\neq b} \kb{a,a}{b,b} + \sum_{a\neq b} \proj{a,b}, \ee 
		which follows from \eqref{tl_op_evol} and reduces to $M^{PF}$ when $N=2$. Thus compared with $M^{PF}$, when $N>2$ TL dynamics has a smaller probability for pairs to flip. In particular, pairs completely cease to flip in the $N\ra\infty$ limit. 
		As with pair-flip, $\mcm$ is doubly stochastic and irreducible, and hence the uniform distribution is $\mcm$'s unique steady state. 

		\section{Rigorous bounds on relaxation and mixing times} \label{app:bounds}
		
		In this section we will prove bounds on the relaxation and mixing times of the Markov process $\mcm$ associated with the pair-flip RU dynamics studied in the main text, as defined in Sec.~\ref{app:stochastics}.

		Before we begin, some remarks and reminders on notation. The initial states we consider will always be computational basis product states, referred to simply as ``product states'' in what follows. $\psi$ will be used to denote an arbitrary product state. We will write $\psi(t)$ for a sequence of product states obtained as a particular realization of the Markov process defined by $\mcm$. This is to be distinguished from the {\it probability distribution} one obtains from evolving a given state $\psi$ for time $t$, which we write as $\k{\mcm^t\psi}$. 		
		As mentioned above, the double stochasticity of $\mcm$ implies that $\mcm$'s equilibrium distribution is uniform on the space of product states, which we will denote as $\mch$ by a slight abuse of notation.

		The irreducible string associated with a product state $\psi$---which is obtained by iteratively removing pairs of adjacent identical characters in $\psi$---will be written as $\bfs_\psi$. In various places below we will write $\bfs_d$ when we wish only to emphasize that the irreducible string in question has length $d$. 
		As defined previously, $\mck_\bfs$ will denote a specific Krylov sector with irreducible string $\bfs$; when we wish to denote an arbitrary $\mck_\bfs$ with $|\bfs|=d$ we will instead similarly write $\mck_d$, and when we wish to explicitly specify the system size we will write $\mck_d^{(L)}$. 
		In all of what follows we will assume for simplicity of notation that $L$ is even, although all results can be readily generalized to odd $L$. 
		
		We will refer to the graph whose vertices are Krylov sectors and whose edges are drawn according to how the coupling to the bath connects the sectors as the {\it Krylov graph}. Since the constraint-breaking term moves the endpoint of the random walk by a distance of exactly 2, each sector in the Krylov graph is connected to $N(N-1)$ other sectors, and thus the Krylov graph is formed by the even / odd (depending on $L\mod2$) sublattice of the symmetric depth $L$ $N$-regular tree (see Fig.~\ref{fig:trees}~a).  
		
		\ss{Markov chains and graph expansion}
		
		We begin by reviewing some central concepts in the theory of Markov chains,\footnote{See e.g. \cite{lyons2017probability} for a good introduction.} letting $\mcm$ to denote the generator of a given Markov process. A key notion in what follows will be that of the expansion:
		\begin{definition}
			Let $R \subset \mch$. The {\it expansion} of $R$ is defined as the amount of probability flow that the uniform distribution experiences out of $R$ during one step of the Markov process $\mcm$:
			\be \label{expansiondef} \cp(R) \equiv \frac1{|R|} \sum_{\psi \in R}\sum_{\psi' \in R^c} \lan \psi' | \mcm | \psi\ran,\ee 
			where the sums run over sets of computational basis product states spanning $R$ and $R^c$, respectively.
		\end{definition}
		
		The utility of this definition is that when $\cp(R)$ is small, states initially in $R$ take a long time to diffuse to its complement $R^c$ \cite{lyons2017probability}: 
		\begin{proposition}
			The probability that a product state $\psi(0)$ randomly drawn from $R$ will leave $R$ in time $t$ is upper bounded by 
			\be \label{pdiffbound} P(\psi(t) \in R^c \, | \, \psi(0) \in R) \leq t \cp(R).\ee 
		\end{proposition} 
		\begin{proof}
			We proceed following Ref.~\cite{lyons2017probability}. Using the Markovity of the dynamics, 
			\bea P(\psi(t) \in R^c \, | \, \psi(0)\in R)
			& \leq \frac{|\mch|}{|R|} \sum_{r = 1}^t P(\psi(r) \in R^c, \, \psi(r-1) \in R) \\ 
			& = \frac{t|\mch|}{|R|} P(\psi(1) \in R^c , \, \psi(0) \in R) 
			\\ 
			&= t \cp(R),\eea 
			where we have used that the probability of selecting any particular state $\psi$ is $1/|\mch|$. 
		\end{proof} 
		
		Regions with smaller $R$ are more ``cut off'' from their complements $R^c$. The most isolated subregion of $\mch$ defines the expansion of $\mcm$: 
		\begin{definition}
			The expansion of $\mcm$ is defined by the minimal expansion of a subregion of $\mch$: 
			\be \cp(\mcm) \equiv \min_{R\subset \mch\, : \, |R| \leq \frac12 |\mch|} \cp(R).\ee 
		\end{definition} 
		
		$\cp(\mcm)$ thus provides a fundamental measure of the slowness of the dynamics. Its primary utility is that it can be used to bound two important timescales characterizing the slowness of $\mcm$, defined as follows: 
		\begin{definition}
			The {\it relaxation time} is defined as the inverse gap of $\mcm$: 
			\be t_{\rm rel} \equiv \frac1{\De_\mcm},\ee 
			where $\De_\mcm \equiv 1 - \l_2$, with $\l_2$ the second largest eigenvalue of $\mcm$. 
			The {\it mixing time} is defined as the amount of time required for the distribution $\mcm^t \psi$ to become close to the uniform distribution $\pi$: 
			\be t_{\rm mix} \equiv \min \{ t \, : \, \max_{\psi \in \mch} || \mcm^t \psi - \pi ||_1 \leq \frac12 \},\ee 
			where the maximum is over all initial product states in $\mch$.  
		\end{definition} 
		
		$t_{\rm rel}$ and $t_{\rm mix}$ are essentially equivalent in their ability to capture the slowness of $\mcm$, as follows from the general bounds $(t_{\rm rel}-1) \ln 2 \leq t_{\rm mix} \leq t_{\rm rel} \ln(4 |\mch|) $ \cite{levin2017markov}. For concreteness we will focus on $t_{\rm rel}$ in what follows. 
		
		\ms 
		
		The expansion $\cp(\mcm)$ bounds $t_{\rm rel}$ via a fundamental result known as Cheeger's inequality (see e.g. \cite{lyons2017probability}): 
		\be \label{cheeger} \frac{\cp(\mcm)^2}2 \leq  \De_\mcm  \leq 2 \cp(\mcm).\ee 
		In what follows, we will calculate $\cp(\mcm)$ for the Markov chains of interest and will use the above inequality to provide a rather tight constraint on the relaxation time. We will also see how the calculation of $\cp(R)$ for appropriate choices of $R$ can be used to bound entanglement growth and the relaxation times of local operators. 
		
		\ss{Local and non-local chains}
		
		The Markov generator derived in Sec.~\ref{app:stochastics} was obtained by considering brickwork RU dynamics composed of three alternating layers: one layer of constrained 2-site gates on the even sublattice ($\mcm_e$ in the notation of Sec.~\ref{app:stochastics}), one layer on the odd sublattice ($\mcm_o$), and one layer consisting solely of depolarizing noise applied to the spin on site $L$ ($\mcm_L$). In this section, we will write $\mcm_{\rm loc}$ for the Markov generator so obtained: 
		\be \label{mloc}	\mcm_{\rm loc}  \equiv \mcm_o \mcm_e \mcm_L .\ee  
		This dynamics is to be constrasted with a simpler ``non-local'' Markov process, obtained in the limit in which the depolarizing noise is weak, with the number of constrained RU brickwork layers being much larger than the number of applications of the depolarizing noise. In this limit, the constrained part of the dynamics thermalizes within each Krylov sector much faster than the time scale over the constraint-breaking term acts. In this situation, the internal structure of the dynamics within each Krylov sector is trivial: as soon as a state reaches a new sector, it instantly spreads out uniformly across that sector, and the thermalization dynamics is consequently controlled solely by transitions between sectors. We will denote the Markov generator of this dynamics by $\mcm_{\rm nonloc}$: 
		\bea \label{mnonloc} 
		\mcm_{{\rm nonloc}} \equiv  (\mcm_o\mcm_e)^\infty \mcm_L = \Pi_{\rm unif} \mcm_L, \eea 
		where 
		\be  \Pi_{\rm unif} = \sum_{v_s} \frac1{|\mck_{v_s}|} \sum_{s,s' \in \mck_s} \kb{s}{s'}\ee 
		projects product states onto the uniform distribution over the Krylov sector they belong to. 
		
		Analytic bounds on the gap of $\mcm_{\rm nonloc}$ are naturally easier to obtain than bounds on that of $\mcm_{\rm loc}$, since for the former, the locality of the dynamics does not come into play. In almost all of this section we will thus focus on $\mcm_{\rm nonloc}$, rather than $\mcm_{\rm loc}$.	 
		This is done without loss of generality since we are primarily interested in obtaining {\it upper} bounds on the gap, and naturally the lack of locality means that 
		\be \De_{\mcm_{\rm nonloc}} \geq \De_{\mcm_{\rm loc}},\ee 
		so that an upper bound on $\De_{\mcm_{\rm nonloc}}$ will also upper bound $\De_{\mcm_{\rm loc}}$. On physical grounds we in fact expect 
		\be \label{locnonlocscaling} \De_{\mcm_{\rm loc}} \sim \frac{\De_{\mcm_{\rm nonloc}}}{L},\ee 
		with the $1/L$ coming from the time needed for the dynamics to make transitions induced by the bath felt across the system. Indeed, we will see shortly that numerical determinations of the gap agree with this scaling.

		\ss{$N>2$}
		
		In this subsection we prove a variety of results establishing exponentially slow thermalization in the strongly fragmented models obtained when $N>2$. The case of $N=2$, where thermalization is expected to be much faster, is dealt with in a subsequent subsection. 
		
		\sss{The spectral gap}
		
		We begin by proving the following theorem: 
		\begin{theorem} \label{thm:mcmexp}
			For $N>2$, the expansion of the Markov process $\mcm_{\rm nonloc}$ at large $L$ satisfies 
			\be \label{cpmcmbound} \cp(\mcm) \leq  \frac C{L^{3/2}} \r^{L},\ee 
			where $C$ is an unimportant $O(1)$ constant 
			\be \label{specrad} \r \equiv \frac{2\sqrt{N-1}}{N} < 1\ee 
			is the spectral radius of the symmetric $N$-ary tree. 
		\end{theorem} 
		
		\begin{figure}
			\centering 
			\includegraphics{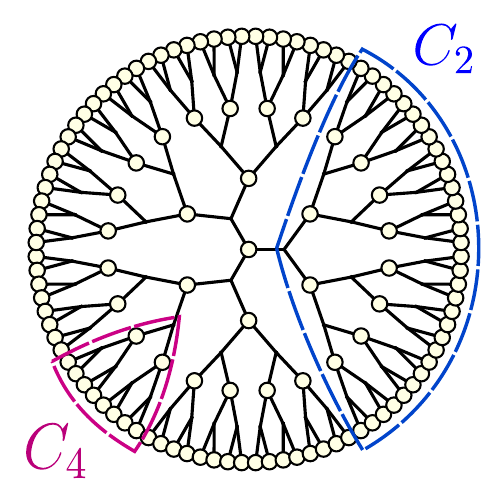} 
			\caption{\label{fig:cones_illustration} The krylov graph for $L=6,N=3$, with examples of regions $C_2, C_4$ indicated by the dashed lines.   }
		\end{figure}
		
		To prove this, we need the following definition: 
		\begin{definition}
			For a length $d-1$ irreducible string $\bfs_{d-1}$, define the {\it cone} $C_{\bfs_{d-1}} \subset \mch $ as 
			the subregion of Hilbert space spanned by the $N-1$ sectors which lie at depth $d$ and whose parent vertex is associated with the string $s_{d-1}$, together with all sectors which are children of these sectors. More formally, we have 
			\be \label{conedef} C_{s_{d-1}} \equiv \bigoplus_{s'\, : \, (s'_1\dots s'_{d-1}) = s_{d-1}, \, |s'| \geq d} \mck_{s'},\ee 
			where the sum is over all irreducible strings of length $\geq d$ whose first $d-1$ entries are equal to $s_{d-1}$. When the exact string $s_{d-1}$ is not important, we will simply write $C_d$ instead of $C_{s_{d-1}}$. A graphical illustration of this definition is shown in Fig.~\ref{fig:cones_illustration}.	
		\end{definition} 
		
		Our proof of Theorem~\ref{thm:mcmexp} will rely on the following Lemma: 
		\begin{lemma} \label{cdlemma}
			
			The expansion of $C_d$ is 
			\be \label{cdphi} \cp(C_d) = \frac{N-1}N \frac{|\mck_{d-1}^{(L-1)}|}{\sum_{c=0}^{(L-d)/2} |\mck^{(L)}_{d+2c}| (N-1)^{2c+1} }. \ee 
			In particular, when both $L,d$ are large and $L-d = \ct(L)$, $\cp(C_d)$ behaves as 
			\bea \label{approxphicd}
			\cp(C_d) \approx \frac{2(N-1)e^{(d/L-v_N)(1-v_N)}}{N^2} \( \ct(d-v_NL) (d/L-v_N) + \ct(v_NL-d) \frac{e^{-\frac{(d-v_NL)^2}{2L}}}{\sqrt{\twp L}} \),\eea 
			where $v_N = 1-2/N$ as before. 
			
		\end{lemma}
		\begin{proof}
			
			\begin{figure}
				\centering
				\includegraphics{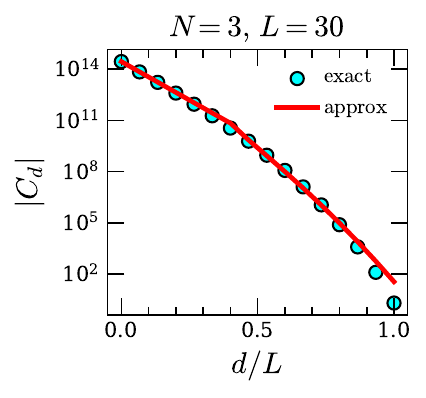}
				\caption{\label{fig:conevolume} The number of states in the cone $C_d$ for $N=3,L=30$. The exact result (circles) is compared with the asymptotic expression \eqref{approxconevol} (solid line).
				}
			\end{figure}
			
			To determine $\cp(C_d)$, we need to know the sizes of both $C_d$ and its boundary. 
			The exact size of $C_d$ is 
			\bea  \label{exactcdvol} |C_d| & = \sum_{c=0}^{(L-d)/2} |\mck_{d+2c}|(N-1)^{2c+1}.\eea 
			When both $L$ and $d$ are large, we may use \eqref{approxdim} to write 
			\bea |C_d| & \approx \frac{ (N-1)^{2-d}N^{L-1}} {\sqrt{2\pi L}}\int_0^{L-d} dx\, \exp\(-\frac{(x+d-v_NL)^2}{2L}  \). 
			\eea 
			The value of this expression depends on the sign of $d/L - v_N$, which we will assume is $\ct(L)$. If $d/L > v_N$, then the saddle point of the integrand does not lie within the integration domain, and we may approximate the integral by $e^{-(d-v_NL)^2/2L}/(d/L-v_N)$. If $d/L<v_N$ we may use the saddle point approximation, with the integral becoming $\sqrt{2\pi L}$. Thus 
			\be \label{approxconevol} |C_d| \approx  \frac{ (N-1)^{2-d}N^{L-1}} {\sqrt{2\pi L}} \( \ct(d-v_NL) \frac{ e^{-(d-v_NL)^2/(2L)}}{d/L-v_N} + \ct(v_NL-d) \sqrt{\twp L} \),\ee 
			which is shown plotted against the exact result \eqref{exactcdvol} in Fig.~\ref{fig:conevolume}. 
			This result means that when $d>v_NL$, most of the states in $C_d$ are concetrated at a depth of $d$, while when $d<v_NL$, most of the states are concentrated within a window of width $\sim \sqrt L$ around $v_NL$. This can be understood simply from the concentration of the biased random walk discussed near \eqref{approxdim}.

			To get $\cp(C_d)$, we need to calculate the probability for states in $C_d$ to move to the complement $C_d^c$ under one step of $\mcm_{\rm nonloc}$. Clearly the only states which can do so are those in the $N-1$ sectors at depth $d$. Under $\mcm_{\rm nonloc}$, an arbitrary state is equally likely (with probability $1/N$) to move to each of the $N-1$ distinct sectors it is connected to. Therefore since there are $N-1$ sectors in $C_d$ whose states can be connected to $C_d^c$, 
			\be \sum_{\psi \in C_d} \sum_{\psi' \in C_d^c} \lan \psi' |\mcm_{\rm nonloc} |\psi\ran  = \frac{N-1}N|\{\psi\in \mck_d\, : \, \p \psi \cap \mck_{d-2} \neq 0\}|\ee 
			where $\mck_d$ is one of the depth $d$ sectors in $C_d$, $\mck_{d-2}$ is the depth $d-2$ sector it is connected to, and $\p \psi$ denotes those states that have nonzero overlap with $\mcm_{\rm nonloc} \k\psi$. If a state in $\mck_d$ is to have $\p\psi \cap \mck_{d-2}\neq 0$, the length-$L$ walk on $T_N$ associated to $\psi$ must lie at depth $d-1$ at step $L-1$. Therefore 
			\be |\{ \psi \in \mck_d \, : \, \p \psi \cap \mck^{(L)}_{d-2} \} | = |\mck_{d-1}^{(L-1)}|,\ee 
			from which the exact expression \eqref{cdphi} then follows. The approximate expression in \eqref{approxphicd} is then obtained using \eqref{approxdim} and a bit of algebra. 
		\end{proof}
		
		The most important aspect of \eqref{approxphicd} is that $\cp(C_d)$ is exponentially small in $L$ if $d< v_NL$, while it is $O(1)$ and roughly $L$-independent when $d>v_NL$. From \eqref{expansiondef}, this means that a state initially localized in a random state in $C_d$ will take exponentially long to diffuse out of $C_d$ when $d<v_NL$, while it can take only $O(1)$ time when $d>v_NL$. This means that diffusion on the Krylov graph is fast for states corresponding to random walks that reach a distance further from the center of the graph than the expected distance of $v_NL$, and slow for states that reach a distance less than $v_NL$. The crossover between these two regimes occurs over a window of depths centered on $v_NL$ and of width $\sim \sqrt L$, which is where most of the states in $\mch$ lie.

		As a particular case of the previous lemma and \eqref{asymptotic_kmax}, we have 
		\begin{corollary} \label{cor:phicd}
			The conductance of the region $C_2$ is 
			\be \label{c2cond} \cp(C_2) = (N-1) \frac{|\mck^{(L-1)}_1|}{N^L - |\mck_0^{(L)}|} \sim L^{-3/2} \r^L,\ee 
			where $\sim$ denotes equality in the asymptotic scaling sense. 
		\end{corollary} 
		This then gives the desired bound appearing in Theorem \ref{thm:mcmexp}. \qed

		\ms 
		\begin{figure}
			\centering  \includegraphics{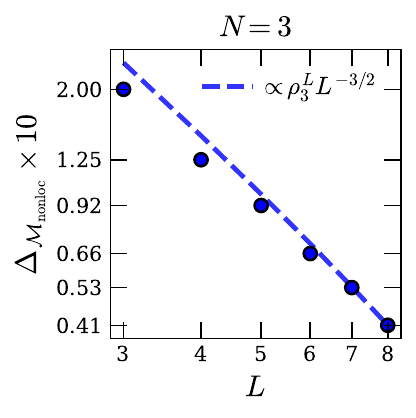} 		\includegraphics{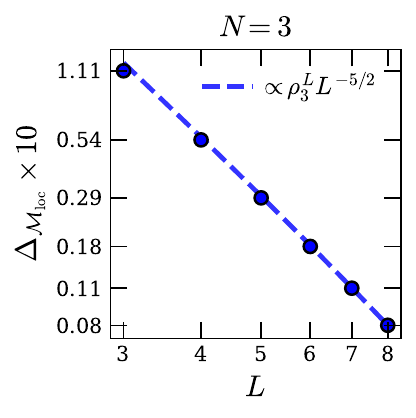} 
			\caption{\label{fig:N3_markov_gap} Markov gaps for $N=3$, computed with exact diagonalization. {\it Left:} The gap $\De_{\mcm_{\rm nonloc}}$ of the non-local chain (blue circles) fit to the analytic bound $\propto \r_3^L L^{-3/2}$ (dashed line). {\it Right:} As left, but for the local chain $\mcm_{\rm loc}$, and fit to $\propto \r_3^L L^{-5/2}$.  }
		\end{figure}
		By Cheeger's inequality, we thus have the following corollary: 
		\begin{corollary} \label{cor:trelbounds}
			The relaxation time $t_{\rm rel} \equiv \De_{\mcm_{\rm nonloc}}\inv$ satisfies 
			\be t_{\rm rel} \geq C' L^{3/2} \r^{-L} \ee 
			for an $O(1)$ constant $C'$.
		\end{corollary} 
		
		The above arguments have only established an upper bound on the expansion, but in fact we expect this bound to be fairly tight, as we expect $C_2$ to be fairly close to the true region of minimal expansion. As steps in this direction, we note first that the region with minimal expansion will always be connected.\footnote{This follows simply from the fact that for positive numbers $A_{1,2}, V_{1,2}$, 
			\be \frac{A_1+A_2}{V_1+V_2} = \frac{A_1}{V_1}p + \frac{A_2}{V_2}(1-p) \geq \min(A_1/V_1, A_2/V_2),\ee 
			where $p\equiv (1/V_2) / (1/V_1 + 1/V_2)$. } 
		Suppose now that $\cp(R)$ is minimized by a connected region whose boundary defines a cut between vertices on the Krylov graph, i.e. suppose that for all $\mck_s$, either $\mck_s \subset R$ or $\mck_s \cap R = \emptyset$. If this is true then $\cp(R)$ is obviously minimized for $R = C_d$ for some $d$, since for a given $\mck_s \in R$, including every child sector of $\mck_s$ in $R$ increases $|R|$ but leaves $\sum_{\psi \in R} \sum_{\psi'\in R^c} \lan \psi' | \mcm_{\rm nonloc}|\psi\ran$ unchanged. This $C_2$ defines a minimal expansion region if one can show that a minimal expansion cut must always be made {\it between} nodes on the Kyrlov graph, instead of being made {\it within} any particular node. This may not be true in complete generality, but we expect $C_2$ to be close enough to the region of minimal expansion that $\cp(\mcm_{\rm nonloc})$ still follows the same asymptotic scaling as the upper bound \eqref{cpmcmbound}. We thus conjecture that there exist constants $C_1,C_2$ such that 
		\be \label{guess_upper} C_1 L^{-3} \r^{2L} \leq \De_{\mcm_{\rm nonloc}} \leq C_2 L^{-3/2} \r^L.\ee 
		While $\r<1$---implying exponentially large relaxation times---the factors of $L^{-3/2}, L^{-3}$ appearing in the above inequality dominate over the exponentiall parts for modest values of $L$, meaning that for smaller system sizes we expect a mostly power-law scaling of $t_{\rm rel}$. 
		
		In Fig.~\ref{fig:N3_markov_gap} we determine the gaps of both $\mcm_{\rm nonloc}$ and its local variant $\mcm_{\rm loc}$ for very small system sizes using exact diagonalization. For the small values of $L$ available, we observe a scaling of $\De_{\mcm_{\rm nonloc}} \sim \r_3 L^{-3/2}$, consistent with a saturation of the upper bound on $t_{\rm rel}$ obtained from \eqref{guess_upper}. We likewise observe a good fit of $\De_{\mcm_{\rm loc}}$ to $\r_3 L^{-5/2}$, with locality thus providing an extra factor of $1/L$, as advocated for around \eqref{locnonlocscaling}.

		In the following subsections, we prove that the exponentially long thermalization time is also manifested in both the growth of entanglement entropy, and in the expectation values of certain local operators.

		\sss{Entanglement entropy}
		
		We now use the formalism developed in the previous section to bound entropy growth. 		
		We will first focus on the case where the dynamics is that of a random unitary circuit perturbed by depolarizing noise on the boundary. In this setting, our diagnostic of thermalization will be the von Neumann entropy of the time-evolved state: 
		\be  \ob{S}(t;\psi)  \equiv \oEE_{\mcc_t}  S( \mcc_t(\proj\psi)),\ee
		where the average is over depth-$t$ quantum circuits $\mcc_t$ defined as in the main text, and $\psi$ is a product state of our choosing. More generally, for a subspace $R\subset \mch$ spanned by product states, we will be interested in the average of the entanglement entropy when the initial states are sampled uniformly from $R$: 
		\be \ob{S}(t;R) \equiv \oEE_{\psi \in R} \ob{S}(t;\psi).\ee 
		The ultimate fixed point of the dynamics we consider is always the maximally mixed state $\unit$, so that regardless of $R$, $\ob S(t\ra\infty;R) = L \ln N$. 
		
		We will prove the following: 
		
		\begin{theorem}	\label{thm:entropy}
			Let $C_d$ be as in \eqref{conedef}, and suppose that $d$ satisfies
			\be \label{dassump} d < v_NL, \qq v_NL-d = \ct(L). \ee 
			Then 
			\be\label{entanglementbound} \ob S(t;C_d) \lesssim L \ln(N) \( 1 - \frac{d}L \frac{\ln (N-1)}{\ln (N)} + t \frac {F_d}{\sqrt L} e^{-L\frac{(d/L-v_N)^2}2}\) + c,\ee 	
			where $c=1/e + 2\ln(N-1)-\ln(N)$ and $F_d$ is an $O(1)$ constant:
			\be \label{fddef} F_d = \frac{4(N-1)}{\sqrt\twp N^2} e^{(d/L-v_N)(1-v_N)}. \ee 
		\end{theorem} 
		
		This means that the entropy of the system will take a time exponentially long in system size to 
		saturate if initialized in a random state in $C_d$, although it may quickly approach the sub-maximal volume-law value of 
		\be \ln (|C_d| )\approx L \ln(N) \( 1 - \frac{d}L \frac{\ln (N-1)}{\ln (N)}\).\ee 
		In particular, it is easy to see that $\ob S(t;\proj \psi) \leq S(t;C_d)$ for all $d$ if $\k\psi$ is any state on the boundary of the Krylov graph, and thus a bound on the thermalization times of such boundary states may be obtained by minimizing the RHS of \eqref{entanglementbound} over $d$ satisfying \eqref{dassump}. 
		
		\begin{proof}
			By the concavity of the entropy and linearity of the trace, we have 
			\be \ob{S}(t;C_d) \leq S(\oEE_{\mcc_t} \oEE_{\psi \in R_d} \mcc_t(\proj\psi)) \equiv S(\ob{\r}(t;d)),\ee 
			where $\ob\r(t;d)$ is the channel- and state-averaged density matrix: 
			\be \ob{\r}(t;d) = \oEE_{\mcc_t} \oEE_{\psi \in C_d} \mcc_t( \proj{\psi}) .\ee 
			Normally this inequality is of limited use when studying random unitary circuits, since the averaged reduced density matrix is usually rendered trivial by the average over the RU part of the dynamics. In our case, the slow mixing of $\mcm$ will mean this is not so; indeed our above result on $t_{\rm rel}$ will be seen to imply that for exponentially long times, $\ob\r(t;d)$ has most of its support on only an exponentially small fraction of Hilbert space. 
			
			Anticipating this, we will first show that if $\s$ is a density matrix mostly supported on some subspace $R$, and if $\s$ does not connect $R$ with its complement, then $S(\s)$ cannot be much more than $\ln |R|$. To this end, define 
			\be \s^R \equiv \frac{\Pi_R \s \Pi_R}{\Tr[\Pi_R\s]},\ee 
			where $\Pi_R$ projects onto $R$. If $\Pi_R \s \Pi_R^\perp = 0$ (where $\Pi^\perp_R \equiv \unit - \Pi_R$), then 
			the trace distance between $\s$ and $\s^R$ is 
			\bea T(\s,\s^R) & = \Tr|\s^R- \s| \\ 
			& = \Tr\left| \frac{\Pi_R \s \Pi_R}{\Tr[\Pi_R \s]} - \Pi_R \s \Pi_R - \Pi_R^\perp \s \Pi_R^\perp \right| \\ 
			& = 1 - \Tr[\Pi_R\s] + \Tr[\Pi_R^\perp \s] \\ 
			& = 2 \Tr[\Pi_R^\perp \s]. \eea 
			Fannes' inequality then implies\footnote{This is not the strongest version of Fannes' inequality, but strengthening it only modifies the unimportant $L$-independent part.}
			\be \label{fannes} |S(\s) - S(\s^R)| \leq T(\s,\s^R)\ln( N^L)+ \frac1e = 2L\ln (N)\Tr[\Pi_R^\perp \s] + \frac1e.\ee 
			Since ${\rm rank}(\s^R) \leq |R|$, we have $S(\s^R) \leq \ln |R|$. Therefore 
			\be \label{ssigbound} S(\s) \leq \ln|R| + 2L\ln(N) \Tr[\Pi_R^\perp\s] +\frac1e.\ee 
			
			We now apply the above inequality to the model under study. From the definition of $\ob\r_A(t;d)$, 
			\be \Tr[\Pi^\perp_{C_d} \ob\r(t;d)] = P(\psi(t) \in C_d^c \, | \ \psi(0) \in C_d),\ee 
			where the probability on the RHS is calculated using $\mcm$. From \eqref{pdiffbound}, we then know that 
			\be \Tr[\Pi_{C_d}^\perp \ob\r(t;d)] \leq t \cp(C_d).\ee 
			Since $\ob\r(t;d)$ is diagonal in the computational basis, we have $\Pi_{C_d} \ob \r(t;d) \Pi^\perp_{C_d} = 0$, which allows us to apply \eqref{ssigbound} to give
			\be \ob S(t;C_d) \leq \ln(|C_d|) + 2tL\ln(N) \cp(C_d) + \frac1e.\ee 
			
			By \eqref{approxphicd} and our assumption that $d<v_NL$, the expansion of $C_d$ is 
			\be \cp(C_d) \approx  \frac{F_d}{2\sqrt{L}} e^{-L \frac{(d/L-v_N)^2}2 },\ee 
			where the $O(1)$ constant $F_d$ is defined as in \eqref{fddef}. Relatedly, \eqref{approxconevol} gives $|C_d| \approx (N-1)^{2-d} N^{L-1}$ since $d<v_NL$, and so  
			\be \ob S(t;C_d) \lesssim L \ln(N) \( 1 - \frac{d}L \frac{\ln (N-1)}{\ln (N)} + t \frac {F_d}{\sqrt L} e^{-L\frac{(d/L-v_N)^2}2}\),  \ee 
			where the $\lesssim$ indicates that we have dropped the constant $c$ appearing in \eqref{entanglementbound}. 
		\end{proof}
		
		This Theorem shows that the entropy for typical states in $C_d$ will take a time exponential in system size to fully saturate provided $v_N-d/L$ is positive and order $L^0$. To make this more concrete, we define the entanglement saturation time $t_S(\g;\psi)$ as the time needed for $\ob{S}(t;\psi)$ to reach a fraction $\g$ of its maximal value $\ob S(t\ra\infty;\psi) = L\ln N$: 
		\be t_S(\g;\psi) \equiv \min\{ t\, : \, \ob{S}(t;\psi) \geq \g L\ln N\}.\ee 
		We then have: 
		\begin{corollary} \label{cor:longts}
			Suppose that $\g$ satisfies 
			\be \label{gammaconds} 1>\g> \g_*, \qq \g_* \equiv 2\(1-v_N\frac{\ln (N-1)}{\ln (N)}\),\ee 
			with $|\g-\g_*| = \ct(L^0)$. 
			Then if $\k\psi$ is a state on the boundary of the Krylov graph,
			\be \label{longts} t_S(\g;\psi) \geq \frac{\g}{2 F_{d_\g}} \sqrt L e^{L\frac{(d_\g/L - v_N)^2}2},\ee 
			where 
			\be d_\g = L(1-\g/2)\frac{\ln(N)}{\ln(N-1)}.\ee  
		\end{corollary} 
		\begin{proof}
			This follows directly from Theorem~\ref{thm:entropy} after fixing $d=d_\g$, so that $1-(d/L)\ln(N-1)/\ln(N) = \g/2$. 
		\end{proof}
		
		Note that $|\g-\g_*| = \ct(L^0)$ implies $|d_\g - v_NL| = \ct(L)$, so that $t_S(\g;\psi)$ is exponentially large in $L$.
		The result \eqref{entanglementbound} suggests that a product state at the boundary of the Krylov graph thermalizes in the following two-stage process. In the first stage, the system undergoes a period of rapid entanglement growth as it quickly occupies sectors at depths $d \geq v_N L$, with the entropy reaching a value of $\ob S\approx \ln |C_{v_NL}|$. In the second stage, it undergoes a gradual $\ln(t)$ growth which takes exponentially long in $L$ to saturate to the steady state value. Note that the length of the first stage becomes smaller as $N$ gets larger, on account of the fact that $v_{N\ra\infty} = 1$.

		While there are exponentially many states on the boundary of the Krylov graph, such states are still an exponentially small fraction of all computational basis product states. Nevertheless, a random product state is exponentially likely to have an entanglement saturation time scaling in the same way as states on the boundary of the Krylov graph: 
		\begin{corollary}
			For $\g$ satisfying \eqref{gammaconds}, 
			\be P_{\psi}\[ t_S(\g;\psi) \geq \frac{\g}{2F_{d_\g}}\sqrt L e^{L\frac{(d_\g/L - v_N)^2}2}\] \gtrsim 1- \frac{\sqrt L}{\sqrt\twp d_\g} e^{-L\frac{(d_\g/L-v_N)^2}2},\ee 
			where $\psi$ is sampled uniformly from all computational basis product states. 
		\end{corollary}
		Thus with unit probability in the $L\ra\infty$ limit, the dynamics initialized from $\k\psi$ takes exponentially long to thermalize.
		\begin{proof}
			As we have seen above, if $\psi$ is sampled uniformly from $\mch$, for large $L$, $|\irr(\psi)|$ will be distributed according to a biased random walk on $\nn$ with velocity $v_N$, and hence $P_\psi[|\irr(\psi)| = d] \approx \frac1{\sqrt{\twp L}} e^{-L(d/L-v_N)^2/2}$. Thus 
			\be P_\psi[|\irr(\psi)| \leq d_\g] \lesssim \frac1{\sqrt{\twp L}} \int_0^{d_\g}dx\,  e^{-L\frac{(x/L-v_N)^2}2} \approx \frac{\sqrt L}{\sqrt\twp d_\g} e^{-L \frac{(d_\g/L - v_N)^2}2},\ee
			which is exponentially small in $L$ since $|d_\g  - v_NL| = \ct(L)$. Thus a randomly drawn $\psi$ is exponentially likely to be contained in $C_{d_\g}$. Furthermore, the concentration of $|\irr(\psi)|$ about $v_NL$ means that the saturation time of this randomly chosen state will be bounded using $1/\cp(C_{d_\g})$ in the same way as in Corollary~\ref{cor:longts}.  \end{proof} 
		
		Thus far we have focused on the entropy of the full system's density matrix when undergoing evolution by the open dynamics $\mcc_t$. In the setting with closed system time evolution performed by an appropriate random unitary circuit $\mcu_t$, 
		\begin{corollary} 
			For a spatial bipartition $AB$, $|A|=|B|=L/2$, define the circuit-averaged bipartite entanglement entropy as 
			\be \ob S_A(t;\psi) \equiv \oEE_{\mcu_t} S(\mathrm{Tr}_B[\mcu_t^\da \proj \psi \mcu_t]).\ee 
			Then $\ob S_A(t;C_d) \equiv \oEE_{\psi \in C_d} \ob S_A(t;\psi)$ satisfies 
			\be \label{sabound}  
			\ob S_A(t;C_d) \lesssim L \ln(N) \( 1 - \frac{d}L \frac{\ln (N-1)}{\ln (N)} + t \frac {2F_d}{\sqrt L} e^{-L\frac{(d/L-v_N)^2}2}\) + c,\ee 	
			whose only difference with respect to the bound \eqref{entanglementbound} is a factor of 2 in the term proportional to $t$. 
		\end{corollary} 
		
		\begin{proof} 
			The reasoning is almost exactly the same as in the proof of Theorem~\ref{thm:entropy}. The only difference is in the anologue of \eqref{fannes} that one obtains, which instead reads
			\be \label{fanneshalfchain} |S(\s_A) - S(\s_A^R)| \leq T(\s_A,\s^R_A)\ln( N^{L/2})+ \frac1e \leq L\Tr[\Pi_R^\perp \s]\ln N + \frac1e.\ee 
			This follows from \eqref{fannes} after using the fact that the trace distance is monotonically decreasing under partial trace, so that the reduced density matrices $\s_A, \s_A^R$ of the $\s,\s^R$ appearing in \eqref{fannes} satisfy 
			\be T(\s_A,\s_A^R) \leq 2 \Tr[\Pi_R^\perp \s].\ee 
		\end{proof} 
		Note that the maximum possible value of $\ob S_A(t;C_d)$ is $\frac L2\ln N$, which is smaller than the $t=0$ value of \eqref{sabound} only if
		\be \ln (|C_{v_NL}|) = L \ln(N) \( 1- \frac dL \frac{\ln(N-1)}{\ln(N)}\) < \frac L2 \ln N \implies  N \geq 5,  \ee 
		from which we conclude that \eqref{sabound} provides a meaningful bound only if $N\geq 5$. 

		\sss{Operator relaxation times} 
		
		We now examine how operator expectation values diagnose the long relaxation times computed above. 
		For an operator $\mco$ of unit norm and a computational basis product state $\k\psi$, we define the relaxation time $t_{\mco}(\g;\psi)$ by the time needed for the expectation value of $\mco$ in the circuit-averaged state $ \ob\r(t;\psi) \equiv  \oEE_{\mcc_t}\mcc_t(\proj\psi)$ to relax to within an amount $\g$ of its circuit-averaged equilibrium value, where we require that $0<\g<1$, $\g = \ct(L^0)$. Since the circuit-averaged density matrix at long times is simply $\unit / |\mch|$, this definition reads  
		\be t_\mco(\g;\psi) \equiv \min\{t \, : \, |\lan \mco\ran_{\ob\r(t;\psi)} - \frac1{|\mch|}\Tr[\mco] |\leq \g \}.\ee 
		
		It is not obvious that operators with exponentially long relaxation times exist. If one was willing to give up locality, a naive guess would be to let $\k\psi$ be a state at the edge of the Krylov graph, and to set $\mco = \proj{\wt \psi}$,  
		where $\k{\wt\psi}$ is any product state on the edge of the tree whose Hamming distance with $\k\psi$ is $L$. In this case $\lan \mco \ran_{\ob \r_\psi(t)}$ vanishes at $t=0$ and indeed takes a time of $\sim t_{\rm rel}$ to increase to its equilibrium value, but that value is $\Tr[\proj{\wt\psi}]/|\mch| = N^{-L}$, whose smallness means that $t_{\proj{\wt \psi}}(\g;\psi) = 0$ by virtue of our requirement that $\g = \ct(L^0)$. 
		
		Fortunately, there nevertheless exist local operators whose relaxation times are exponentially long. These are the normalized charge operators 
		\be Q_a = \frac2L\sum_i (-1)^i\proj{a}_i.\ee  
		Indeed, let $\psi_{{\rm max},a}$ be a product state with maximal $Q_a$ charge $\lan Q_a\ran_{\psi_{{\rm max},a}} =1$. For concreteness, we will fix $\psi_{{\rm max},a} = (ba)^{L/2}$ where $b = a +1 \mod N$. 
		In this section, we will prove the following theorem: 
		\begin{theorem}\label{thm:longcharge} 
			Let $0<\g<v_N/2$, $\g = \ct(L^0)$. Then the relaxation time of $Q_a$ in the state $\k{\psi_{{\rm max},a}}$ is exponentially long: 
			\be \label{tqadef} t_{Q_a} (\g;\psi_{{\rm max},a}) \gtrsim D_\g \sqrt L e^{L \frac{(2\g-v_N)^2}2},\ee 
			where the $O(1)$ constant $D_\g$ is defined as 
			\be D_\g \equiv \frac{2(1+2\g)(N-1)}{N^2 \sqrt\twp} e^{-(2\g-v_N)(1-v_N)}.\ee 
		\end{theorem}
		
		We will in fact prove a slightly more general version of this theorem which allows for more freedom in the choice of the initial state. 

		\begin{figure*}
			\centering 
			\includegraphics[width=\tw]{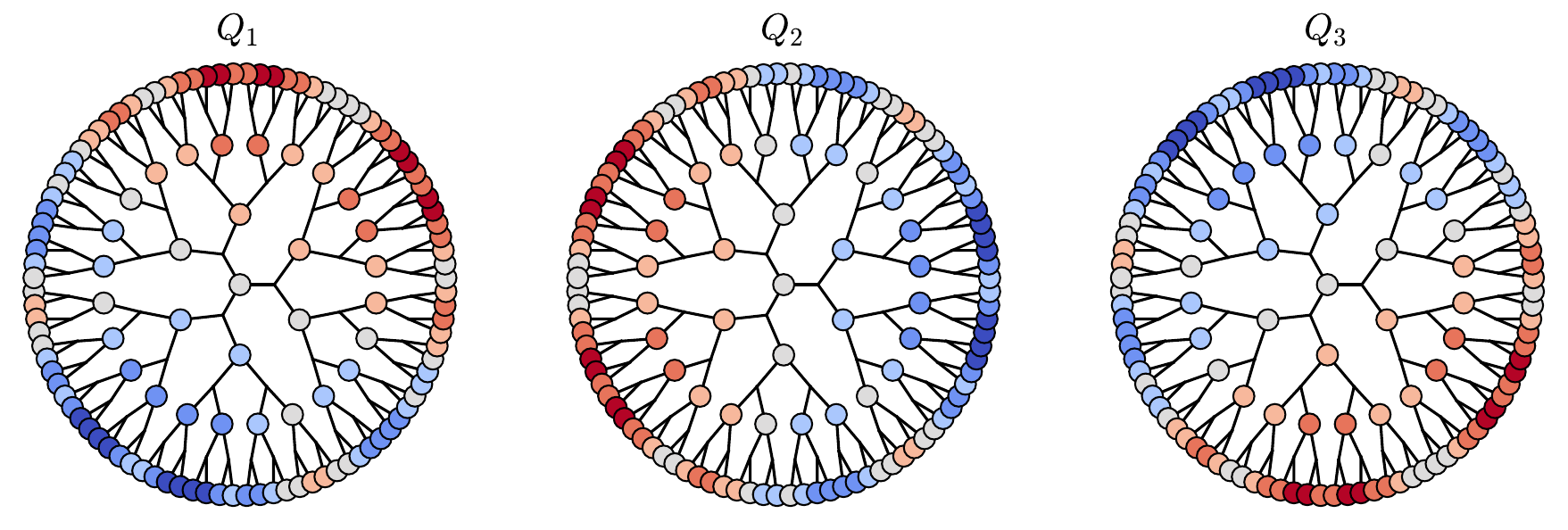}  
			\caption{\label{fig:sector_charges} The quantum numbers of each Krylov sector under the symmetries $Q_a$ for $N=3$ and a system of size $L=6$. The darkest red sectors have the maximum value of $Q_a = L/2$; the darkest blue have $Q_a = -L/2$. } 
		\end{figure*} 

		\begin{proof} 
			Our proof strategy is to define a subspace $A$ of states in $\mch$ which all have a large nonzero expectation value of $Q_a$, and then argue that a system initialized in $A$ typically takes an exponentially long time to move to the complement $A^c$ of $A$ in $\mch$. While knowing the values of $Q_a$ in a given product state do not allow one to distinguish where in the Krylov graph that state lies, the distribution of $Q_a$ charges in the Krylov graph is not homogeneous (see Fig.~\ref{fig:sector_charges}), and this can be used to select out an appropriate choice of $A$. Fixing $b \equiv (a + 1) \mod N$ as above, and assuming that $L\eta \in 2 \nn$ in what follows for simplicity of notation, the space we choose is the cone  
			\be A = C_{(ba)^{L\eta/2}} =  \{ \psi \, : \, \irr(\psi) = (ba)^{\lfloor L\eta/2\rfloor } \times \S^*\},\ee 
			where $\S^*$ is the set of the irreducible strings of all the product states of length less than $L(1-\eta)$, namely $A$ is the space of all product states whose irreducible strings have the first $2\lfloor L\eta/2\rfloor$ elements equal to $(ba)^{L\eta/2}$ (c.f. \eqref{conedef}).
			
			Let us compute the expected value of $Q_a$ obtained after evolving a random state in $A$ for time $t$, 
			\be \lan Q_a\ran_A(t) \equiv  \oEE_{\psi \in A} \lan Q_a \ran_{\ob\r(t;\psi)} .\ee 
			This is 
			\bea \label{qalpht} \lan Q_a\ran_A(t) & = \oEE_{\psi \in A} \( \sum_{\psi' \in A} P(\psi(t) = \psi' | \psi(0) = \psi) \lan Q_a\ran_{\psi'} + \sum_{\psi' \in A^c} P(\psi(t) = \psi' | \psi(0) = \psi) \lan Q_a\ran_{\psi'}  \) \\ 
			& \geq \oEE_{\psi \in A} \( \sum_{\psi' \in A} P(\psi(t) = \psi' | \psi(0) = \psi) \lan Q_a\ran_{\psi'} \) -  P(\psi(t) \in A^c  \, | \,  \psi(0) \in A), \eea 
			where we used $\min_{\psi \in \mch} \lan Q_a \ran_\psi = -1$. To deal with the first term, we need 
			
			\begin{lemma} 
				The average charge of states which begin in $A$ and remain in $A$ at time $t$ satisfies 
				\be \oEE_{\psi \in A}  \sum_{\psi' \in A} P(\psi(t) = \psi' \, |\, \psi(0) = \psi) \lan Q_a\ran_{\psi'}  \geq  \eta(1- P(\psi(t) \in A^c \, | \, \psi(0) \in A)).\ee 
			\end{lemma}
			\begin{proof}
				This is true because the average value of $Q_a$ for states in $A$ is at least $\eta$. Showing this is  complicated slightly by the fact that there exist states in $A$ with $\lan Q_a\ran_\psi$ as small as $\eta  - (1-\eta) = 2\eta-1$. Our strategy will be to show that these states always pair up with states of larger charge to give an average charge bounded below by $\eta$. 
				
				To this end, for each $\psi \in A$, write $\irr(\psi) = (ba)^{L\eta/2} \times s_\psi$ for some length $|\irr(\psi)|-L\eta$ irreducible string $s_\psi$. 
				Suppose first that the second entry of $s_\psi$ is not equal to $a$, $[s_\psi]_2 \neq a$. Then define $\wt \psi$ as the state whose irreducible string differs from $\irr(\psi)$ by a cyclic permutation on the last $|s_\psi|$ entries: 
				\be  s_{\wt \psi} = T(s_\psi),\ee 
				where $T$ is the cyclic permutation $T(a_1\cdots a_n) = a_2a_3\cdots a_{n-1} a_1$. Note that $\irr(\wt\psi) = \irr(\psi) \times s_{\wt \psi}$ is an allowed irreducible string since $[s_{\wt \psi}]_2 \neq a$ by assumption. Thus
				\be \oEE_{\psi' \in A} P(\psi(t) = \psi \, |\, \psi(0) = \psi') = \oEE_{\psi' \in A} P(\psi(t) = \wt \psi \, |\,\psi(0) = \psi').\ee 
				We may thus write 
				\bea\label{qqin} \oEE_{\psi \in A}  \sum_{\psi' \in A\, } P(\psi(t) = \psi' | \psi(0)  = \psi) \lan Q_a\ran_{\psi'} & =
				\oEE_{\psi \in A}\sum_{\psi' \in A\, \, : \, [s_{\psi'}]_2 \neq a}  P(\psi(t) = \psi' | \psi(0) = \psi) \frac{ \lan Q_a\ran_{\psi'} + \lan Q_a \ran_{\wt \psi'} }2 \\ & \qq + \oEE_{\psi \in A}\sum_{\psi' \in A \, : \, [s_{\psi'}]_2 = a } \oEE_{\psi \in A}P(\psi(t) = \psi'\, | \, \psi(0) = \psi) \lan Q_a \ran_\psi, \eea 
				The point of writing things like this is that $(\lan Q_a\ran_{\psi'} + \lan Q_a \ran_{\wt \psi'} )/2= \eta $, simply because the $Q_a$ charge of $s_{\wt \psi}$ is opposite to that of $s_{\psi}$ (on account of the fact that $T$ interchanges sublattices and thus $TQ_a T\inv = -Q_a$), meaning that $(\lan Q_a\ran_{\psi'} + \lan Q_a \ran_{\wt \psi'} )/2$ receives contributions only from the first $L\eta$ characters of $\irr(\psi')$, which carry a $Q_a$ charge of $\eta$. Therefore 
				\bea\label{qqin2} \oEE_{\psi \in A}  \sum_{\psi' \in A\, } P(\psi(t) = \psi' | \psi(0)  = \psi) \lan Q_a\ran_{\psi'} & = \eta 
				\oEE_{\psi \in A}\sum_{\psi' \in A\, \, : \, [s_{\psi'}]_2 \neq a}  P(\psi(t) = \psi' | \psi(0) = \psi)  \\ & \qq +\oEE_{\psi \in A} \sum_{\psi' \in A \, : \, [s_{\psi'}]_2 = a }  P(\psi(t) = \psi'\, | \, \psi(0) = \psi) \lan Q_a \ran_\psi. \eea 
				The second summand in \eqref{qqin2} can be dealt with similarly. Since this summand contains only states with $[s_{\psi'}]_2 = a$, $(ba)^{L\eta/2} \times T(s_{\psi'})$ is not an allowed irreducible string. We thus instead split up the states in the sum as $\psi' = (ba)^{L\eta/2} ca \times p_{\psi'}$ for some $c \neq a$, where $p_{\psi'}$ is a length $L-L\eta/2-2$ irreducible string with $[p_{\psi'}]_1 \neq a$. We can then pair up the subset of these states with $[p_{\psi'}]_2 \neq a$ in the same manner as was done above by defining an appropriate $\wt \psi'$ obtained from cyclically shifting $p_{\psi '}$; each pair appearing in the sum is then seen to have an average $Q_a$ charge of $\eta  +2/L$. Repeating this process, the successive paired states one generates are all seen to have average charge $\eta  + 2n/L$, with $0<n \leq  (1-\eta)L/2$. Since these average charges are all strictly greater than $\eta $, we obtain the bound 
				\bea  \oEE_{\psi \in A}  \sum_{\psi' \in A\, } P(\psi(t) = \psi' | \psi(0)  = \psi) \lan Q_a\ran_{\psi'} & \geq \eta\oEE_{\psi \in A} \sum_{\psi' \in A}P(\psi(t) = \psi' \, | \, \psi(0) = \psi) \\ 
				& = \eta P(\psi(t) \in A \, | \, \psi(0) \in A) \\ 
				& = \eta  (1 - P(\psi(t) \in A^c \, | \, \psi(0) \in A)),\eea 
				which is what we wanted to show. 
			\end{proof}
			
			This result lets us write \eqref{qalpht} as 
			\bea \label{qbound} \lan Q_a\ran_A(t) & \geq \eta - (1+\eta) P(\psi(t) \in A^c \, |\, \psi(0) \in A)),\eea 
			with the second term being bounded from above by $t\cp(A)$ as in \eqref{pdiffbound}. Since $A$ is a union of cones, the expansion of $A$ is simply
			\be \cp(A) = \cp(C_{L\eta+2}),\ee 
			and so 
			\bea  \lan Q_a\ran_A(t) & \geq \eta - t(1+\eta) \cp(C_{L\eta+2}).\eea 
			As we saw in \eqref{approxphicd}, $\cp(C_d)$ is exponentially small only when $d < v_NL, |d/L-v_N| = \ct(L^0)$. To get a long relaxation time, we thus will need to assume that 
			\be \eta < v_N,\ee 
			with $|\eta - v_N| = \ct(L^0)$. If this is the case, we conclude from \eqref{approxphicd} that 
			\bea  \lan Q_a\ran_A(t) & \geq \eta - \frac{t}{D_{\eta/2} \sqrt L e^{L \frac{(\eta-v_N)^2}2}},\eea 
			with the constant 
			\be D_{\eta/2} \equiv  \frac{N^2 \sqrt{\twp}}{2(1+\eta)(N-1)} e^{-(\eta-v_N)(1-v_N)}.\ee 
			The advertised bound on $t_{Q_a}(\g;\psi_{{\rm max},a})$ is then obtained by setting $\eta = 2\g$. Notably, when $\gamma\to 0$, $\lan Q_a \ran_A(t)\gtrsim 2\gamma - t\Phi(C_2)$, so that the lower bound of $t_Q$ is equal to that of the relaxation time, i.e.,
			\be t_{Q_a}(\gamma\to 0;\psi_{\rm max, a})\gtrsim 1/\Phi(C_2)\sim L^{3/2}\rho_3^{-L}. \ee
		\end{proof}

		\ss{$N=2$} 
		
		We now compute the expansion for $N=2$. In this case we expect a large expansion---and hence a fast mixing time---due to the absence of strong Hilbert space bottlenecks. 
		
		In this subsection we will find it most convenient to label the Krylov sectors by their charge $Q$, defined as 
		\be Q \equiv \sum_i (-1)^i Z_i,\ee 
		with $\lan Q\ran_\psi$ measuring the endpoint of the random walk defined by the product state $\psi$. 
		For a length-$L$ system, there are thus $L+1$ sectors $\mck_Q$, with $Q \in \{ -L,-L+2,\cdots,L-2,L\}$ (in the language of the charges $Q_a$ discussed for $N>2$, $Q = Q_1-Q_2$). The dimensions of these sectors are accordingly determined as 
		\be |\mck_Q| = {L \choose \frac{L+Q}2}.\ee 
		
		We now argue that charge relaxation in this case is polynomially fast. In particular, we will argue that the spectral gap of $\mcm_{\rm nonloc}$ satisfies 
		\be \frac1{\pi L} \leq \De_{\mcm_{\rm nonloc}} \leq \sqrt{ \frac{8}{\pi L}}.\ee 
		We will give a rigorous proof of the upper bound, and a slightly less rigorous one for the lower bound. As in our analysis of the $N>2$ case, the lack of rigour for the lower bound comes from making the assumption that the subset $S\subset\mch$ with minimal expansion is determined by a cut which passes ``between'' two Krylov sectors, rather than cutting ``within'' a given sector. Even if this is not true, we expect the minimal expansion to be asymptotically the same as the minimum expansion of a region defined by making only inter-sector cuts.

		With this assumption, it is straightforward to see that the $S$ with minimal expansion can be taken without loss of generality to be of the form
		\be S_{Q} = \bigcup_{Q' \geq Q} \mck_{Q'},\ee 
		with the minimal $S_Q$ having $Q \geq 0$ without loss of generality; in the language of our $N>2$ discussion this is simply a cone $C_Q$. To find the minimal $S_Q$, we use the recursion relation 
		\be |\mck_Q| = |\p S_Q| + |\p S_{Q+2}|,\ee 
		which holds for all $Q \geq 0$ and follows from the fact that each state in $\mck_Q$ is connected to exactly one state in  $\mck_{Q+2}\cup\mck_{Q-2}$. Solving this recursion relation for $|\p S_Q|$ yields 
		\be \label{recurs}|\p S_{Q}| = \sum_{Q' =Q}^{L} (-1)^{\frac{Q'-Q}2} |\mck_{Q'}|,\ee  
		where the sum accordingly only includes those $Q'$ with the same parity as $L$. 
		Now the difference in expansions between adjacent sectors is 
		\be \cp(S_{Q+2}) - \cp(S_Q) = \frac{|S_Q| |\p S_{Q+2}| - |S_{Q+2} ||\p S_Q|}{|S_Q||S_{Q+2}|},\ee 
		which can be evalulated using \eqref{recurs} and the dimensions $|\mck_Q|$, with some unilluminating algebra showing that the RHS is always positive, meaning that $\cp(S_Q)$ is minimized on the smallest value of $Q$ (viz. $Q_{min} =L\mod2$). Taking $L \in 2\nn+1$ for notational simplicity, this gives 
		\be \cp_* =\cp(S_1) =2^{1-L} \sum_{k=0}^{(L-1)/2} (-1)^k {L \choose \frac{L+1}2+k} = \frac{L+1}{L2^L} {L \choose \frac{L+1}2} \approx \sqrt{\frac{2} \pi L}, \ee 
		where we used $|S_1| = |\mch|/2$ in the second equality, $\sum_{k=0}^l (-1)^k {2l+1\choose l+k+1} = \frac{l+1}{2l+1} {2l+1\choose l+1}$ in the third, and Stirling's approximation in the fourth. Cheeger's inequality thus tells us that 
		\be \frac1{\pi L} \leq \De_{\mcm_{\rm nonloc}} \leq \sqrt{ \frac{8}{\pi L}},\ee 
		which is what we wanted to show. 
		
		\begin{figure}
			\centering 
			\includegraphics{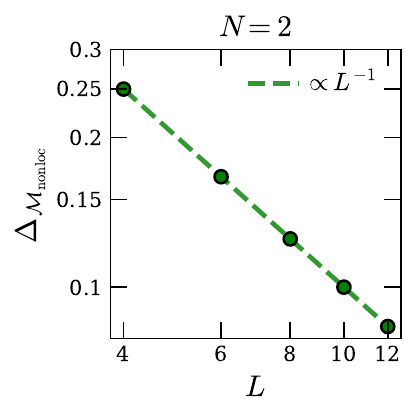} \includegraphics{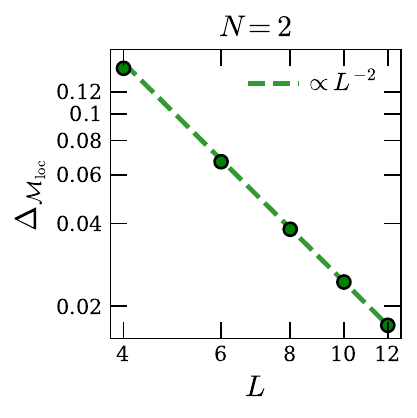} 
			\caption{\label{fig:N2_markov_gap} Markov gaps for $N=2$, computed with exact diagonalization. {\it Left:} The gap $\De_{\mcm_{\rm nonloc}}$ of the non-local chain (green circles) fit to the analytic bound $\propto L^{-1}$ (dashed line). {\it Right:} As left, but for the local chain $\mcm_{\rm loc}$, and fit to $\propto L^{-2}$. }
		\end{figure}
		
		Calculating the gaps of $\De_{\mcm_{\rm loc}}, \De_{\mcm_{\rm nonloc}}$ exactly for small values of $L$ with exact diagonalization yields the scaling shown in Fig.~\ref{fig:N2_markov_gap}. Even for very small values of $L$, the scaling of $\De_{\mcm_{\rm nonloc}}$ fits very well to the linear lower bound of $\sim L\inv$. The gap of the local chain $\mcm_{\rm loc}$ is (as expected) observed to scale slower by one power of $L$ as $\De_{\mcm_{\rm loc}} \sim L^{-2}$, consistent with the simulations of Fig.~\ref{fig:T_relax}. These finding are consistent with those reported in Ref.~\cite{gao2023information}.

		\section{Lack of thermalization in Temperley-Lieb models perturbed by a single-site impurity} \label{App:TL}
		
		In this section we show that a certain class of $SU(N)$ symmetric models---referred to as Temperely-Lieb Hamiltonians in what follows---are such that they remain fragmented even when perturbed by an arbitrary term that has support only on a {\it single site}. 		
		The Hamiltonians we consider are of the form \cite{batchelor1991temperley,aufgebauer2010quantum,moudgalya2022hilbert}
		\be H_{TL} = \sum_{i=1}^L g_i P_{i,i+1},\qq P_{i,i+1} \equiv \frac1N \sum_{a,b=1}^N |a,a\ran\lan b,b|_{i,i+1} \equiv  \proj\Psi_{i,i+1}\ee 
		with $N>2$ in all of what follows. 
		The projectors $P_{i,i+1}$ satisfy $P_{i,i+1}^2 = P_{i,i+1}$ and obey the Temperly-Lieb algebra 
		\bea \label{tlalg}
		P_{i,i+1}P_{j,j+1} P_{i,i+1} & = \frac1{N^2} P_{i,i+1} , \quad i = j\pm1 \\ 
		P_{i,i+1} P_{j,j+1} & = P_{j,j+1} P_{i,i+1} ,\quad |i-j|>1.\eea

		The product states $\k s = \bot_{i=1}^L \k{s_i}$ with $s_i \neq s_{i+1}$ for all $i=1,\dots,L-1$ are clearly annihilated by $H_{TL}$. However, these are far from the only types of states in the kernel of $H_{TL}$. Frozen states can be constructed using ``singlets'' like $\k{\cp_{ab}} \propto |aa\ran - |bb\ran$, which are orthogonal to $\k \Psi$, as well as more complicated states. For example, when $L=3$ we may write down the state 
		\be |\L\ran \propto \sum_{a=1,\dots,N} \z_N^{a-1} (|1aa\ran  + |aa1\ran ) - |111\ran,\ee 
		where $\z_N = e^{\twp i / N}$. $\k\L$ is annihilated by both $P_{1,2}$ and $P_{2,3}$ but is not constructible from the $|\cp_{ab}\ran$ or the $\k s$. This makes enumerating $H_{TL}$'s frozen states rather complicated.

		Nevertheless, owing to the TL algebra obeyed by the projectors $P_{i,i+1}$, quite a large amount of information about the spectrum of $H_{TL}$ can be determined analytically, even without explicitly constructing any eigenstates. We will only need to know a few facts about the counting of $H_{TL}$'s degenerate levels, the first of which is \cite{aufgebauer2010quantum}
		\begin{proposition}
			On an open chain of length $L$, the number of zero-energy eigenstates of $H_{TL}$ is 
            {\be |\O_L| =
			\begin{cases} \label{gsd} \frac{(N+\sqrt{N^2-4})^{L+1} - (N-\sqrt{N^2-4})^{L+1}}{2^{L+1} \sqrt{N^2-4}},\quad N>2\\
            L+1,\quad N=2
            \end{cases} \ee}
		\end{proposition}
		If $g_i \geq 0$ for all $i$, $H_{TL}$ is frustration free, and the states in $\O_L$ are in one-to-one correspondence with the ground states of $H_{TL}$. Our results however will hold for arbitrary $g_i$. 
		
		\begin{proof}
			To keep our presentation self-contained, we will reproduce the proof from Ref.~\cite{aufgebauer2010quantum}. Let $\O_L$ denote the space of states annihilated by all of the $P_{i,i+1}$: 
			\be \O_L \equiv \bigcap_i \ker P_{i,i+1}.\ee 
			We are interested in obtaining the dimension $|\O_L|$ of this space. 
			
			We proceed by induction. Given $\O_{L-1}$, we determine $\O_L$ as 
			\be \O_L = \ker (P_{1,2}  \, : \, \mch \tp \O_{L-1} \ra \k\Psi \tp \O_{L-2} ),\ee 
			where $\mch$ is the onsite Hilbert space. It is easy to check that the map $P_{1,2}$ is surjective, and thus $|\O_L| $ is determined as 
			\be |\O_L| = \dim[ \mch \tp \O_{L-1} ] - \dim[ |\Psi\ran \tp \O_{L-2} ] = N |\O_{L-1}| - |\O_{L-2}|.\ee 
			
			The initial values needed to set up a recurrence relation are $|\O_0| = 1$, $|\O_1| = N$. The soltuion to this recurrence relation is precisely \eqref{gsd}. 
		\end{proof}
		
		We now ask about the spectrum of the model 
		\be H = H_{TL} + H_{\rm imp},\ee 
		where $H_{\rm imp}$ is an arbitrary $N\times N$ single-site Hamiltonian acting on the first site only. Cases with impurities acting in the middle of the chain, or with multiple non-adjacent impurities, can be treated similarly at the expense of more complicated notation. 
		
		We can use a similar approach as the one used in the computation of $|\O_L|$ to determine a large number of degenerate states of $H$:
		\begin{proposition}
			Let $\k\a$ be the eigenstates of $H_{\rm imp}$ and $\ep_\a$ be the corresponding eigenvalues. Let $\O_L^\a$ denote the eigenstates of $H$ with eigenvalue $\ep_\a$. Then $\O_L^\a$ is always non-empty, and in particular has degeneracy 
			\be \label{impdeg} |\O^{\ep_\a}_L| = |\O_{L-1}| - |\O_{L-2}|.\ee 
		\end{proposition}
		
		This proposition is rather surprising at face value: since $[H_{\rm imp},H_{TL}] \neq 0$ we would not generically expect $H$ to have eigenvalues equal to those of $H_{\rm imp}$, and \eqref{impdeg} says that the number of such eigenvalues is in fact exponentially large in $L$. Indeed though, such states can be readily constructed, as we now prove. 
		
		\begin{proof}
			Consder the space $ \k\a \tp \O_{L-1}$. States in this space are almost the desired eigenstates of $H$, but they are not annihilated by $P_{1,2}$. The desired eigenstates are thus identified with the kernel \be \O_L^\a =  \ker (P_{1,2}  \, : \, \k\a \tp \O_{L-1} \ra \k\Psi \tp \O_{L-2}).\ee 
			The map $P_{1,2}$ here is surjective for almost all choices of $H_{\rm imp}$. Therefore using similar reasoning as above, 
			\bea |\O_L^\a| & = \dim [ \k\a \tp \O_{L-1} ] - \dim [\k \Psi \tp \O_{L-2}] \\ & = |\O_{L-1}| - | \O_{L-2}| \\ & = \frac1N |\O_L| - \(1-\frac 1N\) |\O_{L-2}|.\eea  
		\end{proof}

		A similar result holds in a situation where one places impurities on both ends of the chain: 
		
		\begin{corollary}
			Let 
			\be H_{two-imp} = H_{TL} + H_{{\rm imp},1} + H_{{\rm imp},L},\ee 
			where $H_{{\rm imp},1/L}$ are independently random single-site Hamiltonians on the left and right ends of the chain, respectively. Let $\ep_{\a_{1/L}}$ be their corresponding eigenvalues. Let also $\O^{\a_1+\b_L}_L$ be the set of eigenstates of $H_{two-imp}$ with energy $\ep_{\a_1} + \ep_{\b_L}$. Then 
			\be |\O^{\a_1 + \b_L}_L| = |\O_{L-2}| - |\O_{L-3}|.\ee 
		\end{corollary}
		This result is true even if the left and right impurities do not possess a common eigenbasis; all that matters is that they act only on single sites. 
		\begin{proof}
			The proof proceeds as in the previous proposition, except with the starting state drawn from the vector space $\O^{\a_1}_{L-1} \tp |\b_L\ran$. 
		\end{proof}

		\ms 
		
		The above results tell us that even in the presence of local impurities, the spectrum of $H$ contains exponentially large degeneracies. This leads to initially frozen states possessing memory of their initial conditions for infinitely long times, since the large number of degeneracies mean that a large number of energy levels of $H$ do not dephase relative to one another. The following theorem makes this intuition precise: 
		
		\begin{theorem}[non-thermalization of the TL + impurity model]
			Let $\mcm$ be the memory that typical fozen states in $\O_L$ have of their initial conditions: 
			\be \mcm \equiv \mathop{\mathbb{E}}_{f\in \O_L}  \lim_{T\ra \infty} \frac1T \int_0^T dt\, | \lan f | e^{-iHt} | f\ran |^2,\ee 
			where $\EE_{f \in \O_L}$ denotes an average over states in $\O_L$. Then in the $L\ra \infty$ limit, 
			\be \mcm =  \frac1 N \( 1 - \frac{4(N-1)}{(N+\sqrt{N^2 -4})^2 } \)^2 + \cdots \ee 
			is an order one constant (with the $\cdots$ denoting terms vanishing as $L\ra\infty$). 
		\end{theorem}
		
		\begin{proof}
			Our proof will proceed by making use of the degenerate eigenstates in $\O_L^\a$. We start by writing 
			\bea \mcm & = \oEE_{f\in \O_L} \sum_{\mu,\nu \in \spec(H)} \d_{E_\mu, E_\nu} |\lan \mu | f\ran |^2 |\lan \nu | f \ran |^2 \\ 
			& \geq \oEE_{f\in \O_L} \sum_{\a \in \spec(H_{\rm imp})} \sum_{\mu,\nu \in \O^\a_L} |\lan \mu | f\ran |^2 |\lan \nu | f \ran |^2  \\
			& = \oEE_{f\in \O_L} \sum_{\a \in \spec(H_{\rm imp})} \lan f | \Pi_{\O^\a_L} | f \ran^2 ,\eea 
			where 
			\be \Pi_{\O^\a_L} = \sum_{\mu \in \O^\a_L} \proj \mu\ee 
			and the $\k\mu,\k\nu$ are orthonormal eigenstates of $H$. 

			We then use $||v||^2_2 \geq \frac1N ||v||_1^2$ for any $v\in \rr^N$ together with an application of Jensen's inequality $\oEE_x [ f(x)^2] \geq (\oEE_x[f(x)])^2$ to write  
			\bea \mcm & \geq \frac1 N \oEE_{f\in \O_L} \( \sum_{\a \in \spec(H_{\rm imp})} \lan f | \Pi_{\O^\a_L} | f\ran  \)^2  \\ 
			& \geq \frac1 N \( \sum_{\a \in \spec(H_{\rm imp})}  \oEE_{f\in \O_L}   \lan f | \Pi_{\O^\a_L} | f\ran  \)^2.
			\eea 
			The average over frozen states gives $\oEE_{f\in \O_L} \proj f = \Pi_{\O_L} / |\O_L|$, so 
			\be \mcm \geq \frac 1{N |\O_L|^2} \(  \sum_{\a \in \spec(H_{\rm imp})}  \Tr [\Pi_{\O^\a_L} \Pi_{\O_L}]\)^2.\ee 
			Since $\O^\a_L\subset \O_L$, the trace is simply $|\O^\a_L|$ for all $\a$. Thus 
			\bea \mcm & \geq N \(\frac{|\O^\a_L|}{|\O_L|} \)^2  \\ 
			& = \frac1 N \(1 - (N-1) \frac{|\O_{L-2}|}{|\O_L|} \)^2 \\ 
			& \sim \frac1 N \( 1 - \frac{4(N-1)}{(N+\sqrt{N^2 -4})^2 } \)^2, \eea 
			where the $\sim$ in the tast line denotes the leading scaling in the $L\ra \infty$ limit. 
		\end{proof}
		
		Since $\mcm \sim O(1)$ but $|\O_L|$ is exponentially large, almost all frozen states are guaranteed to retain memory of their intitial conditions for infinitely long times. From our numerical results we believe this result should remain true even when a term 
		\be H_Z = \sum_i \sum_{a=1}h^a_i \proj{a}_i\ee 
		is added to $H$, although the proof techniques for this case must necessarily be different on account of the fact that the degeneracy of $H$'s spectrum is completely lifted for a generic choice of $h^a_i$.

		\bigskip \bigskip

\end{document}